\documentclass{llncs}

\usepackage{graphicx}

\usepackage[utf8]{inputenc}
\usepackage[american]{babel}
\usepackage[babel]{csquotes}

\usepackage{amscd}
\usepackage{amsmath}
\usepackage{amssymb}
\usepackage{amsthm}

\usepackage[usenames,dvipsnames]{xcolor}
 \usepackage{array}

\usepackage{xspace}
\usepackage{paralist}
\usepackage{xifthen}
\usepackage{url}

\usepackage{tikz}
\usetikzlibrary{trees,decorations,arrows,automata,shadows,positioning,plotmarks,backgrounds,shapes}
\usetikzlibrary{calc,matrix,fit,petri,decorations.markings,decorations.pathmorphing,patterns,intersections,decorations.text}

\tikzstyle{mystate}=[state,inner sep=3pt,minimum size=20pt,line width=0.5mm]
\tikzstyle{fstate}=[state,accepting,inner sep=2pt,minimum size=3pt]
\tikzstyle{istate}=[state,initial,inner sep=2pt,minimum size=3pt]
\tikzstyle{mysquare}=[inner sep=3pt,minimum size=15pt,line width=0.5mm]
\tikzstyle{fmysquare}=[inner sep=3pt,minimum size=15pt,line width=0.5mm,accepting]
\newcommand{\SFSAutomatEdge}[5]{\path[->](#1) edge[#4,line width=0.5mm] node[#5] {\ensuremath{#2}} (#3);}
\usepackage{tabularx}
\usepackage{booktabs}

\newcommand{\comment}[1]{}

\newcommand{\new}[1]{\textcolor{black}{#1}}
\newcommand{\full}[1]{\textcolor{black}{#1}}

\newcommand{\newb}{\color{black}}
\newcommand{\newe}{\color{black}}

\newcommand{\fullb}{\color{black}}
\newcommand{\fulle}{\color{black}}

\newcommand{\OpPre}[1]{\mathop{\mathrm{pfx}(#1)}}

\newcommand{\OpInf}[1]{\mathop{\mathrm{Inf}(#1)}}

\newcommand{\OpPreTD}[1]{\mathop{\mathsf{CondPre}\ifthenelse{\isempty{#1}}{}{(#1)}}}
\newcommand{\OpPreTDc}[1]{\mathop{\overline{\mathsf{CondPre}}\ifthenelse{\isempty{#1}}{}{(#1)}}}

\newcommand{\nocontentsline}[3]{\relax}
\newcommand{\tocless}[1]{\bgroup\let\addcontentsline=\nocontentsline#1\egroup}
\newcommand{\EndBox}{\hspace{\stretch{100}}$\Box$}

\newcommand{\equationwmem}[2]{
  \newcounter{#1}
  \setcounter{#1}{\value{equation}}
  \expandafter\newcommand\csname equationcom#1\endcsname{#2}%
  \begin{equation} \label{#1}
    \csname equationcom#1\endcsname
  \end{equation}
}

\newcommand{\equationfrommem}[1]{
    \ifcsname c@tempcounter\endcsname
    \setcounter{tempcounter}{\value{equation}}
  \else
    \newcounter{tempcounter}
    \setcounter{tempcounter}{\value{equation}}
  \fi
  \setcounter{equation}{\value{#1}}
  \begin{equation}
    \csname equationcom#1\endcsname
  \end{equation}
  \setcounter{equation}{\value{tempcounter}}
}

\newcommand{\BR}[1]{\left( #1 \right)}
\newcommand\set[1]{\{ #1 \}}
\newcommand\semantics[1]{[\![ #1 ]\!]}

\newcommand{\REFlem}[1]{\text{Lemma~\ref{#1}}}
\newcommand{\REFthm}[1]{\text{Theorem~\ref{#1}}}

\newcommand{\REFsec}[1]{Section~\ref{#1}}

\newcommand{\REFfig}[1]{Figure~\ref{#1}}

\newcommand{\REFapp}[1]{Appendix~\ref{#1}}
\newcommand{\REFproblem}[1]{Problem~\ref{#1}}
\newcommand{\REFtable}[1]{Table~\ref{#1}}

\let\exampleOrig\endexample
\def\endexample{\hspace*{0pt}\hfill$\triangleleft$\exampleOrig}

\newcommand{\ON}[1]{\mathsf{#1}}

\def\clap#1{\hbox to 0pt{\hss#1\hss}}

\newcommand{\DiCases}[4]{\ensuremath{\begin{cases}%
#1,&~#2\\%
#3,&~#4%
\end{cases}}}%
\newcommand{\TriCases}[6]{\ensuremath{\begin{cases}%
#1,&~#2\\%
#3,&~#4\\%
#5,&~#6%
\end{cases}}}%
\makeatletter 
\newif\ifFIRST
\newif\ifSECOND
\let\LISTOP\relax
\newcommand{\List}[4][\;]{#3#1%
        \FIRSTtrue
        \@for\i:=#2\do{%
        \ifFIRST\LISTOP{\i}\FIRSTfalse\else,\LISTOP{\i}\fi%
        }%
        #1#4%
        \let\LISTOP\relax
}
\makeatother
\newcounter{DINGLIST}
\stepcounter{DINGLIST}
\newcommand{\markD}[3][\;\;]{\text{\ding{\the\numexpr171+\theDINGLIST}\stepcounter{DINGLIST}}#1#3}

\makeatletter

\newcommand{\propNeg}{\@ifstar\propNegStar\propNegNoStar}
\newcommand{\propNegStar}[1]{\ensuremath{\left(\propNegNoStar{#1}\right)}}
\newcommand{\propNegNoStar}[2][\cdot]{\ensuremath{\neg\ifthenelse{\isempty{#2}}{#1}{#2}}}

\newcommand{\propConj}{\@ifstar\propConjStar\propConjNoStar}
\newcommand{\propConjStar}[2]{\ensuremath{\left(\propConjNoStar{#1}{#2}\right)}}
\newcommand{\propConjNoStar}[3][\cdot]{\ensuremath{\ifthenelse{\isempty{#2}}{#1}{#2}\wedge\ifthenelse{\isempty{#3}}{#1}{#3}}}

\newcommand{\propDisj}{\@ifstar\propDisjStar\propDisjNoStar}
\newcommand{\propDisjStar}[2]{\ensuremath{\left(\propDisjNoStar{#1}{#2}\right)}}
\newcommand{\propDisjNoStar}[3][\cdot]{\ensuremath{\ifthenelse{\isempty{#2}}{#1}{#2}\vee\ifthenelse{\isempty{#3}}{#1}{#3}}}

\newcommand{\propImp}{\@ifstar\propImpStar\propImpNoStar}
\newcommand{\propImpStar}[2]{\ensuremath{\left(\propImpNoStar{#1}{#2}\right)}}
\newcommand{\propImpNoStar}[3][\cdot]{\ensuremath{\ifthenelse{\isempty{#2}}{#1}{#2}\Rightarrow\ifthenelse{\isempty{#3}}{#1}{#3}}}

\newcommand{\propAequ}{\@ifstar\propAequStar\propAequNoStar}
\newcommand{\propAequStar}[2]{\ensuremath{\left(\propAequNoStar{#1}{#2}\right)}}
\newcommand{\propAequNoStar}[3][\cdot]{\ensuremath{\ifthenelse{\isempty{#2}}{#1}{#2}\Leftrightarrow\ifthenelse{\isempty{#3}}{#1}{#3}}}

\newcommand{\AllQ}{\@ifstar\AllQStar\AllQNoStar}
\newcommand{\AllQStar}[3][\;]{\ensuremath{\left(\forall #2#1.#1#3\right)}}
\newcommand{\AllQNoStar}[3][\;]{\ensuremath{\forall #2#1.#1#3}}
\newcommand{\AllQu}{\@ifstar\AllQuStar\AllQuNoStar}
\newcommand{\AllQuStar}[3][\;]{\ensuremath{\left(\forall^{\infty} #2#1.#1#3\right)}}
\newcommand{\AllQuNoStar}[3][\;]{\ensuremath{\forall^{\infty} #2#1.#1#3}}

\newcommand{\ExQ}{\@ifstar\ExQStar\ExQNoStar}
\newcommand{\ExQStar}[3][\;]{\ensuremath{\left(\exists #2#1.#1#3\right)}}
\newcommand{\ExQNoStar}[3][\;]{\ensuremath{\exists #2#1.#1#3}}

\newcommand{\NExQ}{\@ifstar\NExQStar\NExQNoStar}
\newcommand{\NExQStar}[3][\;]{\ensuremath{\left(\nexists #2#1.#1#3\right)}}
\newcommand{\NExQNoStar}[3][\;]{\ensuremath{\nexists #2#1.#1#3}}

\newcommand{\UniqueQ}{\@ifstar\UniqueQStar\UniqueQNoStar}
\newcommand{\UniqueQStar}[3][\;]{\ensuremath{\left(\exists! #2#1.#1#3\right)}}
\newcommand{\UniqueQNoStar}[3][\;]{\ensuremath{\exists! #2#1.#1#3}}

\newenvironment{propConjA}{\left(\def\unionAtest{1}\begin{array}{@{\if\unionAtest1\gdef\unionAtest{0}\phantom{\wedge}\else\wedge\fi}l@{}}}{\end{array}\right)}

\newenvironment{propDisjA}{\left(\def\unionAtest{1}\begin{array}{@{\if\unionAtest1\gdef\unionAtest{0}\phantom{\vee}\else\vee\fi}l@{}}}{\end{array}\right)}

  \newlength{\SFS@HEIGHT}
  \newlength{\SFS@WIDTH}
  \newcommand{\SplitX}[2]{
            \settoheight{\SFS@HEIGHT}{$#2$}
            \settowidth{\SFS@WIDTH}{$#2$}
            \mbox{\begin{tikzpicture}[baseline=(current bounding box.center)]
            \node[] (E) at (0,0) {$#1$};
            \node[inner sep=0pt] (F) at ($(E.south west)+(1ex,-1ex)+(3ex+.5\SFS@WIDTH,-\SFS@HEIGHT)$) {$#2$};
            \node[] (E) at (0,0) {\phantom{$#1$}};
            \draw[fill] ($(E.east)+(1ex,0ex)$) circle (.2ex);
            \draw[-] ($(E.east)+(1ex,0ex)$) -- ($(E.south east)+(1ex,-0.5ex)$) -- ($(E.south west)+(1ex,-0.5ex)$) -- ($(E.south west)+(1ex,-1ex)-(0,\SFS@HEIGHT)$) -- ($(E.south west)+(2.5ex,-1ex)-(0,\SFS@HEIGHT)$);
            \draw[fill] ($(E.south west)+(2.5ex,-1ex)-(0,\SFS@HEIGHT)$) circle (.2ex);
            \end{tikzpicture}}}
  \newcommand{\SplitS}[2]{
            \settoheight{\SFS@HEIGHT}{$#2$}
            \settowidth{\SFS@WIDTH}{$#2$}
            \mbox{\begin{tikzpicture}[baseline=(current bounding box.center)]
            \node[] (E) at (0,0) {$#1$};
            \node[inner sep=0pt] (F) at ($(E.south west)+(1ex,0.5ex)+(0ex+.5\SFS@WIDTH,-\SFS@HEIGHT)$) {$#2$};
            \end{tikzpicture}}}

  \newcommand{\SetCompSplit}[2]{\left\{\SplitS{#1\mid}{#2}\right\}}

\newcommand{\Set}[2][]{\List[#1]{#2}{\{}{\}}}
\newcommand{\VSet}[2][]{\let\LISTOP\val\List[#1]{#2}{\{}{\}}}

\newcommand{\Tuple}[2][]{\List[#1]{#2}{(}{)}}

\newcommand{\VTuple}[2][]{\let\LISTOP\val\List[#1]{#2}{(}{)}}

\newcommand{\UNION}{\@ifstar\UNIONStar\UNIONNoStar}
\newcommand{\UNIONStar}[2]{\ensuremath{\left(\UNIONNoStar{#1}{#2}\right)}}
\newcommand{\UNIONNoStar}[2]{\ensuremath{\ifthenelse{\isempty{#1}}{\cdot}{#1}\cup\ifthenelse{\isempty{#2}}{\cdot}{#2}}}

\newcommand{\UNIOND}{\@ifstar\UNIONDStar\UNIONDNoStar}
\newcommand{\UNIONDStar}[2]{\ensuremath{\left(\UNIONDNoStar{#1}{#2}\right)}}
\newcommand{\UNIONDNoStar}[2]{\ensuremath{\ifthenelse{\isempty{#1}}{\cdot}{#1}\uplus\ifthenelse{\isempty{#2}}{\cdot}{#2}}}

\newcommand{\SETMINUS}{\@ifstar\SETMINUSStar\SETMINUSNoStar}
\newcommand{\SETMINUSStar}[2]{\ensuremath{\left(\SETMINUSNoStar{#1}{#2}\right)}}
\newcommand{\SETMINUSNoStar}[2]{\ensuremath{\ifthenelse{\isempty{#1}}{\cdot}{#1}\setminus\ifthenelse{\isempty{#2}}{\cdot}{#2}}}

\newcommand{\INTERSECT}{\@ifstar\INTERSECTStar\INTERSECTNoStar}
\newcommand{\INTERSECTStar}[2]{\ensuremath{\left(\INTERSECTNoStar{#1}{#2}\right)}}
\newcommand{\INTERSECTNoStar}[2]{\ensuremath{\ifthenelse{\isempty{#1}}{\cdot}{#1}\cap\ifthenelse{\isempty{#2}}{\cdot}{#2}}}

\newcommand{\CARTPROD}{\@ifstar\CARTPRODStar\CARTPRODNoStar}
\newcommand{\CARTPRODStar}[2]{\ensuremath{\left(\CARTPRODNoStar{#1}{#2}\right)}}
\newcommand{\CARTPRODNoStar}[2]{\ensuremath{\ifthenelse{\isempty{#1}}{\cdot}{#1}\times\ifthenelse{\isempty{#2}}{\cdot}{#2}}}

\newcommand{\FINCOUNT}{\@ifstar\FinCountStar\FinCountNoStar}
\newcommand{\FinCountStar}[1]{\ensuremath{\#(\ifthenelse{\isempty{#1}}{\cdot}{#1})}}
\newcommand{\FinCountNoStar}[1]{\ensuremath{\#\left(\ifthenelse{\isempty{#1}}{\cdot}{#1}\right)}}

\makeatother 

\newcommand{\SetCompX}[3][]{\left\{#1#2#1\middle\vert#1#3#1\right\}}

\newcommand{\twoup}[1]{\ensuremath{2^{#1}}}

  \newcommand{\y}{\ensuremath{y}}

  \newcommand{\Q}{\ensuremath{Q}}
  \newcommand{\Qo}{\ensuremath{Q^1}}
  \newcommand{\Qz}{\ensuremath{Q^0}}

  \newcommand{\Fz}{\ensuremath{F^0}}
  \newcommand{\Fo}{\ensuremath{F^1}}
  
  \newcommand{\Fsz}{\ensuremath{\mathcal{F}^0}}
  \newcommand{\Fso}{\ensuremath{\mathcal{F}^1}}
  
  \newcommand{\FA}{\ensuremath{F_{\mathcal{A}}}}
  
  \newcommand{\FG}{\ensuremath{F_{\mathcal{G}}}}
    
  \newcommand{\Fc}{\ensuremath{\mathcal{F}}}
  \newcommand{\FcA}{\ensuremath{\mathcal{F}_{\mathcal{A}}}}
  \newcommand{\FcG}{\ensuremath{\mathcal{F}_{\mathcal{G}}}}

 \newcommand{\Lomega}{\ensuremath{\mathcal{L}}}

  \newcommand{\Gg}{\ensuremath{H}}

 \newcommand{\Zt}{\ensuremath{\overline{Z}}}
 \newcommand{\Yt}{\ensuremath{\overline{Y}}}
 \newcommand{\Xt}{\ensuremath{\overline{X}}}
 \newcommand{\Wt}{\ensuremath{\overline{W}}}
   \newcommand{\FAt}{\ensuremath{\overline{F}_{\mathcal{A}}}}
  \newcommand{\FGt}{\ensuremath{\overline{F}_{\mathcal{G}}}}

   \newcommand{\ps}[1]{\ensuremath{{}^{#1}\!}}
  \newcommand{\Za}[1]{\ensuremath{\ifthenelse{\isempty{#1}}{\ps{a}Z}{\ps{#1}Z}}}
 \newcommand{\Ya}[1]{\ensuremath{\ifthenelse{\isempty{#1}}{\ps{a}Y}{\ps{#1}Y}}}
 \newcommand{\Xa}[1]{\ensuremath{\ifthenelse{\isempty{#1}}{\ps{a}X}{\ps{#1}X}}}
 \newcommand{\Wa}[1]{\ensuremath{\ifthenelse{\isempty{#1}}{\ps{a}W}{\ps{#1}W}}}
   \newcommand{\FAa}[1]{\ensuremath{\ifthenelse{\isempty{#1}}{\ps{b}\FA}{\ps{#1}F_{\mathcal{A}}}}}
  \newcommand{\FGa}[1]{\ensuremath{\ifthenelse{\isempty{#1}}{\ps{a}\FG}{\ps{#1}F_{\mathcal{G}}}}}

   \newcommand{\Zta}[1]{\ensuremath{\ifthenelse{\isempty{#1}}{\ps{a}\overline{Z}}{\ps{#1}\overline{Z}}}}
 \newcommand{\Yta}[1]{\ensuremath{\ifthenelse{\isempty{#1}}{\ps{a}\overline{Y}}{\ps{#1}\overline{Y}}}}
 \newcommand{\Xta}[1]{\ensuremath{\ifthenelse{\isempty{#1}}{\ps{a}\overline{X}}{\ps{#1}\overline{X}}}}
 \newcommand{\Wta}[1]{\ensuremath{\ifthenelse{\isempty{#1}}{\ps{a}\overline{W}}{\ps{#1}\overline{W}}}}
   \newcommand{\FAta}[1]{\ensuremath{\ifthenelse{\isempty{#1}}{\ps{a}\FAt}{\ps{#1}\overline{F}_{\mathcal{A}}}}}
  \newcommand{\FGta}[1]{\ensuremath{\ifthenelse{\isempty{#1}}{\ps{a}\FGt}{\ps{#1}\overline{F}_{\mathcal{G}}}}}

  \newcommand{\Trz}{\ensuremath{\delta^0}}
  \newcommand{\Tro}{\ensuremath{\delta^1}}
  \newcommand{\Tr}[1]{\ensuremath{\delta^{#1}}}

    \newcommand{\qz}{\ensuremath{q^0}}
  \newcommand{\qo}{\ensuremath{q^1}}

  \newcommand{\fz}[1]{\ensuremath{f^{0}\ifthenelse{\isempty{#1}}{}{\Tuple{#1}}}} 
  \newcommand{\fo}[1]{\ensuremath{f^{1}\ifthenelse{\isempty{#1}}{}{\Tuple{#1}}}}  
    \newcommand{\fI}[1]{\ensuremath{f^{i}\ifthenelse{\isempty{#1}}{}{\Tuple{#1}}}}  
  \newcommand{\PreE}{\ensuremath{\ON{Pre}^{\exists}}}
  \newcommand{\PreA}{\ensuremath{\ON{Pre}^{\forall}}}
  \newcommand{\Preo}{\ensuremath{\ON{Pre}^1}}
  \newcommand{\Prez}{\ensuremath{\ON{Pre}^0}}
  
\newcommand{\Pre}{\ensuremath{\ON{Pre}}}

    \newcommand{\rank}{\ensuremath{\ON{rank}}}

\newcommand{\ex}[1]{\ensuremath{\check{#1}}}

 \newcommand{\Tre}[1]{\ensuremath{\ex{\delta}^{#1}}}

\def\N{\hspace{4pt}\raise 3pt \hbox{\circle{7}} \hspace{0pt} }

\newcommand{\rs}{\hspace{-0.8mm}}

\begin{document}


\title{Environmentally-friendly GR(1) Synthesis}
\author{Rupak Majumdar\inst{1} \and Nir
  Piterman\inst{2}\thanks{Supported by project ``d-SynMA'' that is funded by the European Research Council (ERC) under the European Union's Horizon 2020 research and innovation programme (grant agreement No 772459).}
  \and Anne-Kathrin Schmuck\inst{1}\thanks{corresponding author: akschmuck@mpi-sws.org}}
\institute{MPI-SWS, Kaiserslautern, Germany \and University of Leicester, Leicester, UK}
\maketitle

\begin{center}
\end{center}

\begin{abstract}
Many problems in reactive synthesis are stated using two formulas
---an \emph{environment assumption} and a \emph{system guarantee}---
and ask for an implementation that satisfies the guarantee in
environments that satisfy their assumption.
Reactive synthesis tools often produce strategies that
formally satisfy such specifications by actively preventing an
environment assumption from holding. 
While formally correct, such strategies do not capture the intention
of the designer. 
We introduce an additional requirement in reactive synthesis,
\emph{non-conflictingness},
which asks that a system strategy should always allow the environment
to fulfill its liveness requirements.
We give an algorithm for solving GR(1) synthesis that produces
non-conflicting strategies.
Our algorithm is given 
by a 4-nested fixed point in the $\mu$-calculus, in contrast to the
usual 3-nested fixed point for GR(1).
Our algorithm ensures that, in every environment that satisfies its
assumptions on its own, traces of the resulting implementation satisfy
both the assumptions and the guarantees.
In addition, the asymptotic complexity of our algorithm is the same as
that of the usual GR(1) solution. 
We have implemented our algorithm and show how its performance compares to the usual GR(1) synthesis algorithm.
\end{abstract}

\section{Introduction}\label{sec:Intro}

Reactive synthesis from temporal logic specifications provides a methodology to
automatically construct a system implementation from a declarative specification
of correctness.
Typically, reactive synthesis starts with a set of requirements on the system
and a set of assumptions about the environment.
The objective of the synthesis tool is to construct an implementation that ensures 
all guarantees are met in every environment that satisfies all the assumptions;
formally, the synthesis objective is an implication $A \Rightarrow G$.
In many synthesis problems, the system can actively influence whether an environment
satisfies its assumptions.
In such cases, an implementation that prevents the environment from satisfying its
assumptions is considered correct for the specification: since the antecedent of the implication
$A\Rightarrow G$ does not hold, the property is satisfied.

Such implementations satisfy the letter of the specification but not its intent.
Moreover, assumption-violating implementations are not a theoretical curiosity but are regularly
produced by synthesis tools such as \texttt{slugs} \cite{slugs}.
In recent years, a lot of research has thus focused on how to 
model environment assumptions
\cite{KleinPnueli-2010,DBPU10,BCGHHJKK14,BloemEhlersKoenighofer_2015,AssumptionsInSynthesis},
so that assumption-violating implementations are ruled out. 
Existing research either removes the ``zero sum'' assumption
on the game by introducing different levels of co-operation
\cite{BloemEhlersKoenighofer_2015}, by introducing
equilibrium notions inspired by non-zero sum games
\cite{BrenguierRaskinSankur_2017,KupfermanPV16,FismanKL10}, 
or by introducing richer quantitative objectives 
on top of the temporal specifications \cite{BloemCHJ09,AlmagorKRV17}.

\smallskip
\noindent\textbf{Contribution}
In this paper, we take an alternative approach.
We consider the setting of GR(1) specifications, where assumptions and
guarantees are both 
conjunctions of safety and B\"uchi properties \cite{Bloem_etal_2012}.
GR(1) has emerged as an expressive specification formalism \cite{rogersten11synthesis,xu2015specification,Johnson17jfr}
and, unlike full linear temporal logic, synthesis for GR(1) can be
implemented in time quadratic in the state/transition space.
In our approach, the environment is assumed to satisfy its assumptions
provided the system does not prevent this.
Conversely, the system is required to pick a strategy
that ensures the guarantees whenever the assumptions are satisfied,
but additionally ensures \emph{non-conflictingness}:  
along each finite prefix of a play according to the strategy, there
exists the persistent possibility for the environment to play such
that its liveness assumptions will be met.
\full{Note that non-conflictingness is not a trace property; we cannot ``compile away'' this additional requirement
into a different GR(1) or even $\omega$-regular 
objective.
}

Our main contribution is to show a $\mu$-calculus characterization 
of winning states (and winning strategies) that rules out system
strategies that are winning by preventing the environment from
fulfilling its assumptions. 
Specifically, we provide a $4$-nested fixed point that characterizes
winning states and strategies that are \emph{non-conflicting} and ensure 
all guarantees are met if all the assumptions are satisfied.
Thus, if the environment promises to satisfy its assumption if
allowed, the resulting strategy ensures both the assumption and the
guarantee.

Our algorithm does not introduce new notions of winning, or new logics
or winning conditions. 
Moreover, since $\mu$-calculus formulas with $d$ alternations can be
computed in $O(n^{\lceil d/2\rceil})$ time \cite{Seidl96,Browne96}, 
the $O(n^2)$ asymptotic complexity for the new symbolic algorithm is
the same as the standard GR(1) algorithm.

 \begin{figure}[t]
 \begin{center}
    \begin{tikzpicture}
 \begin{footnotesize}
 
 \def\h{0} \def\hm{3} \def\ha{0.3} 
 \def\v{0} \def\vm{2} \def\va{0.3}
 
 \def\nodesepx{1.2}
 \def\nodesepy{0.9}
 
 \def\wa{1.2pt}
 \def\radius{0.19cm}

 \fill [black!15] (\h,\v) rectangle (\h+\ha,\v+\va);
 \fill [black!15] (\h+2*\ha,\v+\va) rectangle (\h+3*\ha,\v+2*\va); 
 \fill [black!55] (\h+2*\ha,\v) rectangle (\h+3*\ha,\v+\va);
 \fill [black!55] (\h,\v+\va) rectangle (\h+\ha,\v+2*\va); 
 
  \node  (a) at (\h+0.5*\hm*\ha,\v+\vm*\va) {};
  \node  (b) at (\h+0.5*\hm*\ha,\v+\vm*\va+2*\va) {};
  \path[->](b) edge[line width=\wa] node[] {} (a);
 
 
\foreach \nodex/\nodey/\obsx/\obsy/\robx/\roby in {0/0/0/0/2/0,2/0/1/1/1/0,4/0/2/1/1/1,5/0/2/1/0/1,7/0/1/0/1/1,9/0/0/0/1/0} {
 
 
  \draw [line width=\wa] (\h+\nodesepx*\nodex,\v-\nodesepy*\nodey) rectangle (\h+\nodesepx*\nodex+\ha*\hm,\v-\nodesepy*\nodey+\va*\vm);
 \draw [step=\ha,line width=0.1pt] (\h+\nodesepx*\nodex,\v-\nodesepy*\nodey) grid (\h+\nodesepx*\nodex+\ha*\hm,\v-\nodesepy*\nodey+\va*\vm);
   \foreach \y/\xi/\xa in {1/0/1,1/2/\hm} {
 \draw [line width=\wa] (\h+\nodesepx*\nodex+\xi*\ha,\v-\nodesepy*\nodey+\y*\va) -- (\h+\nodesepx*\nodex+\xa*\ha,\v-\nodesepy*\nodey+\y*\va);
 }

 \node [draw, circle,minimum size=\radius, anchor=center,inner sep=0pt,line width=1pt] at (\h+\nodesepx*\nodex+\obsx*\ha+0.5*\ha,\v-\nodesepy*\nodey+\obsy*\va+0.5*\va) {};

 \node [draw, rectangle,minimum width=\radius,minimum height=\radius, anchor=center,inner sep=0pt,line width=1pt] at (\h+\nodesepx*\nodex+\robx*\ha+0.5*\ha,\v-\nodesepy*\nodey+\roby*\va+0.5*\va) {};
 
 \node (lab) at (\h+\nodesepx*\nodex+0.5*\hm*\ha,\v-\nodesepy*\nodey-0.3*\vm*\va) {\small $q_{\nodex}$};
  }

  
\foreach \nodex/\nodey/\obsx/\obsy/\obsxb/\obsyb/\robx/\roby in {1/0/0/0/1/0/2/0,3/0/0/1/1/1/1/0,6/0/1/1/2/1/0/1,8/0/1/0/2/0/1/1} {
 
  \draw [line width=\wa] (\h+\nodesepx*\nodex,\v-\nodesepy*\nodey) rectangle (\h+\nodesepx*\nodex+\ha*\hm,\v-\nodesepy*\nodey+\va*\vm);
 \draw [step=\ha,line width=0.1pt] (\h+\nodesepx*\nodex,\v-\nodesepy*\nodey) grid (\h+\nodesepx*\nodex+\ha*\hm,\v-\nodesepy*\nodey+\va*\vm);
   \foreach \y/\xi/\xa in {1/0/1,1/2/\hm} {
 \draw [line width=\wa] (\h+\nodesepx*\nodex+\xi*\ha,\v-\nodesepy*\nodey+\y*\va) -- (\h+\nodesepx*\nodex+\xa*\ha,\v-\nodesepy*\nodey+\y*\va);
 }
 \node [draw, circle,minimum size=\radius, anchor=center,inner sep=0pt,dotted,line width=1pt] at (\h+\nodesepx*\nodex+\obsx*\ha+0.5*\ha,\v-\nodesepy*\nodey+\obsy*\va+0.5*\va) {};
 \node [draw, circle,minimum size=\radius, anchor=center,inner sep=0pt,dotted,line width=1pt] at (\h+\nodesepx*\nodex+\obsxb*\ha+0.5*\ha,\v-\nodesepy*\nodey+\obsyb*\va+0.5*\va) {};

 \node [draw, rectangle,minimum width=\radius,minimum height=\radius, anchor=center,inner sep=0pt,line width=1pt] at (\h+\nodesepx*\nodex+\robx*\ha+0.5*\ha,\v-\nodesepy*\nodey+\roby*\va+0.5*\va) {};
  
  \node (help) at (\h+\nodesepx*\nodex+0.5*\hm*\ha,\v-\nodesepy*\nodey+\vm*\va) {};
  \path[->](help) edge[loop above,line width=1pt] node[] {} (help);
  
 \node (lab) at (\h+\nodesepx*\nodex+0.5*\hm*\ha,\v-\nodesepy*\nodey-0.3*\vm*\va) {\small$q_{\nodex}$};
  }

  \def\arc{0.2}
  \def\arcsep{0.05}
  
  \foreach \nodex in {1,...,9} {
  \node [inner sep=0pt] (a) at (\h+\nodesepx*\nodex-\arc-\arcsep,\v+\va) {};
  \node [inner sep=0pt] (b) at (\h+\nodesepx*\nodex-\arcsep,\v+\va) {};  
    \path[->] (a) edge[line width=1pt] node[] {} (b);
  }
  
  \node [inner sep=0pt] (a) at (\h+\nodesepx*9+\hm*\ha-0.05,\v-0.05) {};
  \node [inner sep=0pt] (b) at (\h+\nodesepx*9+\hm*\ha-0.05,\v-0.4) {};
  \node [inner sep=0pt] (c) at (\h+0.05,\v-0.4) {};
  \node [inner sep=0pt] (d) at (\h+0.05,\v-0.05) {};
  \path[->] (c) edge[line width=1pt] node[] {} (d);
  \draw[line width=1pt] (c) -| (a) (b) -| (d);
  
  \end{footnotesize}
 \end{tikzpicture}
 \end{center}
 \vspace*{-7mm}
  \caption{Pictorial representation of a \emph{desired} strategy for a
    robot (square) moving in a maze in presence of a moving obstacle
    (circle). Obstacle and robot start in the lower left and right
    corner, can move at most one step at a time (to non-occupied
    cells) and cells that they should visit infinitely often are
    indicated in light and dark gray (see $q_0$), respectively. Nodes
    with self-loops ($q_{\set{1,3,6,8}}$) can be repeated finitely
    often with the obstacle located at one of the dotted
    positions.}\label{fig:maze_3x2_4FP}
  \vspace*{-7mm}
\end{figure}
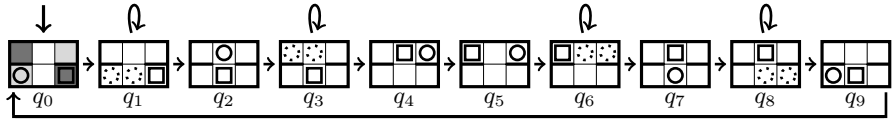

\smallskip
\noindent\textbf{Motivating Example}
Consider a small two-dimensional maze with 3x2 cells as depicted in
\REFfig{fig:maze_3x2_4FP}, state $q_0$. A robot (square) and an
obstacle (circle) are located in this maze and can move at most one
step at a time to non-occupied cells. There is a wall
between the lower and upper left cell and the lower and upper right
cell. The interaction between the robot and the object
is as follows: first the environment chooses where to move
the obstacle to, and, after observing the new location of the
obstacle, the robot chooses where to move.

Our objective is to synthesize a strategy for the robot s.t.\ it
visits both the upper left and the lower right corner of the maze
(indicated in dark gray in \REFfig{fig:maze_3x2_4FP}, state $q_0$)
infinitely often. Due to the walls in the maze the robot needs to
cross the two white middle cells infinitely often to fulfill this
task. If we assume an arbitrary, adversarial behavior of the
environment (e.g., placing the obstacle in one white cell and never
moving it again) this desired robot behavior cannot be
enforced. 
%
We therefore assume that the obstacle is actually another
robot that is required to visit the lower left and the upper right
corner of the maze (indicated in light gray in
\REFfig{fig:maze_3x2_4FP}, state $q_0$) infinitely often. While we do
not know the precise strategy of the other robot (i.e., the obstacle),
its liveness assumption is enough to infer that the obstacle will
always eventually free the white cells. Under this assumption the
considered synthesis problem has a solution. 

Let us first discuss one intuitive strategy for the robot in this
scenario, as depicted in \REFfig{fig:maze_3x2_4FP}. We start in $q_0$
with the obstacle (circle) located in the lower left corner and the
robot (square) located in the lower right corner.  
Recall that the obstacle will eventually move towards the upper right
corner. The robot can therefore  wait until it does so, indicated by
$q_1$. Here, the dotted circles denote possible locations of the
obstacle during the (finitely many) repetitions of $q_1$ by following
its self loop. Whenever the obstacle moves to the upper part of the
maze, the robot moves into the middle part ($q_2$). Now it waits until
the obstacle reaches its goal in the upper right, which is
ensured to happen after a finite number of visits to $q_3$. When the
obstacle reaches the upper right, the robot moves up as well
($q_4$). Now the robot can freely move to its goal in the upper left
($q_5$). This process symmetrically repeats for moving back to the
respective goals in the lower part of the maze ($q_6$ to $q_9$ and
then back to $q_0$). With this strategy, the interaction between
environment and system goes on for infinitely many cycles and the
robot fulfills its specification. 

\begin{figure}[t]
\begin{center}
   \begin{tikzpicture}
 \begin{footnotesize}
 
 \def\h{0} \def\hm{3} \def\ha{0.3} 
 \def\v{0} \def\vm{2} \def\va{0.3}
 
 \def\nodesepx{1.5}
 \def\nodesepy{0.9}
 
 \def\wa{1.2pt}
 \def\radius{0.19cm}

 \fill [black!15] (\h+\nodesepx,\v) rectangle (\h+\nodesepx+\ha,\v+\va);
 \fill [black!15] (\h+\nodesepx+2*\ha,\v+\va) rectangle (\h+\nodesepx+3*\ha,\v+2*\va); 
 \fill [black!55] (\h+\nodesepx+2*\ha,\v) rectangle (\h+\nodesepx+3*\ha,\v+\va);
 \fill [black!55] (\h+\nodesepx,\v+\va) rectangle (\h+\nodesepx+\ha,\v+2*\va); 
 
  \node  (a) at (\h+\nodesepx+0.5*\hm*\ha,\v+\vm*\va) {};
  \node  (b) at (\h+\nodesepx+0.5*\hm*\ha,\v+\vm*\va+2*\va) {};
  \path[->](b) edge[line width=\wa] node[] {} (a);
 
 
\foreach \nodex/\nodey/\obsx/\obsy/\robx/\roby in {1/0/0/0/2/0,3/0/1/1/1/0,0/0/0/0/1/0} {
 
 
  \draw [line width=\wa] (\h+\nodesepx*\nodex,\v-\nodesepy*\nodey) rectangle (\h+\nodesepx*\nodex+\ha*\hm,\v-\nodesepy*\nodey+\va*\vm);
 \draw [step=\ha,line width=0.1pt] (\h+\nodesepx*\nodex,\v-\nodesepy*\nodey) grid (\h+\nodesepx*\nodex+\ha*\hm,\v-\nodesepy*\nodey+\va*\vm);
   \foreach \y/\xi/\xa in {1/0/1,1/2/\hm} {
 \draw [line width=\wa] (\h+\nodesepx*\nodex+\xi*\ha,\v-\nodesepy*\nodey+\y*\va) -- (\h+\nodesepx*\nodex+\xa*\ha,\v-\nodesepy*\nodey+\y*\va);
 }

 \node [draw, circle,minimum size=\radius, anchor=center,inner sep=0pt,line width=1pt] at (\h+\nodesepx*\nodex+\obsx*\ha+0.5*\ha,\v-\nodesepy*\nodey+\obsy*\va+0.5*\va) {};

 \node [draw, rectangle,minimum width=\radius,minimum height=\radius, anchor=center,inner sep=0pt,line width=1pt] at (\h+\nodesepx*\nodex+\robx*\ha+0.5*\ha,\v-\nodesepy*\nodey+\roby*\va+0.5*\va) {};
 
  }

  
\foreach \nodex/\nodey/\obsx/\obsy/\obsxb/\obsyb/\robx/\roby in {2/0/0/0/1/0/2/0,4/0/1/0/2/0/0/0} {
 
  \draw [line width=\wa] (\h+\nodesepx*\nodex,\v-\nodesepy*\nodey) rectangle (\h+\nodesepx*\nodex+\ha*\hm,\v-\nodesepy*\nodey+\va*\vm);
 \draw [step=\ha,line width=0.1pt] (\h+\nodesepx*\nodex,\v-\nodesepy*\nodey) grid (\h+\nodesepx*\nodex+\ha*\hm,\v-\nodesepy*\nodey+\va*\vm);
   \foreach \y/\xi/\xa in {1/0/1,1/2/\hm} {
 \draw [line width=\wa] (\h+\nodesepx*\nodex+\xi*\ha,\v-\nodesepy*\nodey+\y*\va) -- (\h+\nodesepx*\nodex+\xa*\ha,\v-\nodesepy*\nodey+\y*\va);
 }
 \node [draw, circle,minimum size=\radius, anchor=center,inner sep=0pt,dotted,line width=1pt] at (\h+\nodesepx*\nodex+\obsx*\ha+0.5*\ha,\v-\nodesepy*\nodey+\obsy*\va+0.5*\va) {};
 \node [draw, circle,minimum size=\radius, anchor=center,inner sep=0pt,dotted,line width=1pt] at (\h+\nodesepx*\nodex+\obsxb*\ha+0.5*\ha,\v-\nodesepy*\nodey+\obsyb*\va+0.5*\va) {};

 \node [draw, rectangle,minimum width=\radius,minimum height=\radius, anchor=center,inner sep=0pt,line width=1pt] at (\h+\nodesepx*\nodex+\robx*\ha+0.5*\ha,\v-\nodesepy*\nodey+\roby*\va+0.5*\va) {};
  
  \node (help) at (\h+\nodesepx*\nodex+0.5*\hm*\ha,\v-\nodesepy*\nodey+\vm*\va) {};
  \path[->](help) edge[loop above,line width=1pt] node[] {} (help);
  
  }
  
  \def\nodex{4} 
  \foreach \obsx in {0,...,2} {
  \node [draw, circle,minimum size=\radius, anchor=center,inner sep=0pt,dotted,line width=1pt] at (\h+\nodesepx*\nodex+\obsx*\ha+0.5*\ha,\v+1*\va+0.5*\va) {};
  }
  
  \foreach \a/\b in {0/4,1/0,2/1,3/2,4/3} {
  \node (lab) at (\h+\nodesepx*\a+0.5*\hm*\ha,\v-0.3*\vm*\va) {\small$q_{\b}$};
  }
  
  \def\arc{0.5}
  \def\arcsep{0.05}
  
  \foreach \nodex in {2,...,4} {
  \node [inner sep=0pt] (a) at (\h+\nodesepx*\nodex-\arc-\arcsep,\v+\va) {};
  \node [inner sep=0pt] (b) at (\h+\nodesepx*\nodex-\arcsep,\v+\va) {};  
    \path[->] (a) edge[line width=1pt] node[] {} (b);
  }
  
    \node [inner sep=0pt] (a) at (\h+\nodesepx-\arc-\arcsep,\v+\va) {};
  \node [inner sep=0pt] (b) at (\h+\nodesepx-\arcsep,\v+\va) {};  
    \path[->] (b) edge[line width=1pt] node[] {} (a);
  
  
    \node (help) at (\h+\nodesepx*0+0.5*\hm*\ha,\v+\vm*\va) {};
  \path[->](help) edge[loop above,line width=1pt] node[] {} (help);
  
  \end{footnotesize}
 \end{tikzpicture}
\end{center}
 \vspace*{-9mm}
 \caption{Pictorial representation of the \emph{GR(1) winning
     strategy} synthesized by \texttt{slugs} for the robot (square) in
   the game described in
   \REFfig{fig:maze_3x2_4FP}.}\label{fig:maze_3x2_3FP} 
 \vspace*{-9mm}
\end{figure}
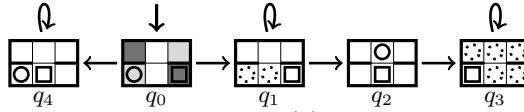

The outlined synthesis problem can be formalized as a two player game
with GR(1) winning condition. When solving this synthesis problem
using the tool \texttt{slugs} \cite{slugs}, we obtain the strategy
depicted in \REFfig{fig:maze_3x2_3FP} (not the desired one in
\REFfig{fig:maze_3x2_4FP}). The initial state, denoted by $q_0$ is the
same as in \REFfig{fig:maze_3x2_4FP} and if the environment moves the
obstacle into the middle passage ($q_1$) the robot reacts as before;
it waits until the object eventually proceeds to the upper part of the
maze ($q_2$). However, after this happens the robot takes the chance
to simply move to the lower left cell of the maze and stays there
forever ($q_3$). By this, the robot prevents the environment from
fulfilling its objective. Similarly, if the obstacle does not
immediately start moving in $q_0$, the robot takes the chance to place
itself in the middle passage and stays there forever ($q_4$). This
obviously prevents the environment from fulfilling its liveness
properties. 
 
In contrast, when using our new algorithm to solve
the given synthesis problem, we obtain the strategy
given in \REFfig{fig:maze_3x2_4FP}, which satisfies the guarantees while allowing
the environment assumptions to be satisfied.

%

\smallskip
\noindent\textbf{Related Work}
Our algorithm is inspired by supervisory controller synthesis for
non-terminating processes \cite{RamadgeOmega,Thistle94}, resulting in
a fixed-point algorithm over a Rabin-Büchi automaton. This algorithm
has been simplified for two interacting Büchi automata in
\cite{Moor_Report_omegaSCT} without proof. We adapt this algorithm to
GR(1) games and provide a new, self-contained proof in the framework
of two-player games, which is distinct from the supervisory controller
synthesis setting (see \cite{EhlersJdeds,SchmuckMoorMajumdar_2018} for a recent comparison of both frameworks).

The problem of correctly handling assumptions in synthesis has
recently gained attention in the reactive synthesis community
\cite{AssumptionsInSynthesis}.  
As our work does not assume precise knowledge about the environment
strategy (or the ability to impose the latter), it is distinct from
cooperative approaches such as assume-guarantee
 \cite{ChatterjeeHenzinger_2007} or rational synthesis \cite{RationalSynthesis_2010}. It is closest related to obliging games
\cite{ChatterjeeHornLoeding_ObligingGames_2010},
cooperative reactive synthesis \cite{BloemEhlersKoenighofer_2015}, and
assume-admissible synthesis \cite{BrenguierRaskinSankur_2017}.  
Obliging games \cite{ChatterjeeHornLoeding_ObligingGames_2010}
incorporate a similar notion of non-conflictingness as our work,  
but do not condition winning of the system on the environment
fulfilling the assumptions. This makes obliging games harder to win.   
Cooperative reactive synthesis \cite{BloemEhlersKoenighofer_2015}
tries to find a winning strategy 
enforcing $A\cap G$. If this specification is not realizable, it is
relaxed and the obtained system strategy enforces the 
guarantees if the environment cooperates \enquote{in the right
  way}. Instead, our work always assumes the same form of
cooperation; coinciding with just one cooperation lever in
\cite{BloemEhlersKoenighofer_2015}. 
Assume-admissible synthesis \cite{BrenguierRaskinSankur_2017} for two
players results in two individual synthesis problems. Given that both
have a solution, only implementing the system strategy ensures that
the game will be won if the environment plays \emph{admissible}. This
is comparable to the view taken in this paper, however, assuming that
the environment plays \emph{admissible} is stronger then our
assumption on an environment attaining its liveness properties if not
prevented from doing so. Moreover, we only need so solve one synthesis
problem, instead of two. 
However, it should be noted that 
\cite{ChatterjeeHornLoeding_ObligingGames_2010,BloemEhlersKoenighofer_2015,BrenguierRaskinSankur_2017} handle $\omega$-regular assumptions
and guarantees. We focus on the practically important GR(1) fragment and our method
better leverages the computational benefits for this fragment.  

\section{Two Player Games and the Synthesis Problem}\label{sec:prelim}


\subsection{Two Player Games}
%
\noindent\textbf{Formal Languages}
Let $\Sigma$ be a finite alphabet. 
We write $\Sigma^*$, $\Sigma^+$, and $\Sigma^\omega$ for the sets of finite words, non-empty finite words,
and infinite words over $\Sigma$. 
%
%
We write $w\le v$ (resp., $w<v$) if $w$ is a prefix of $v$ (resp., a
strict prefix of $v$). 
The set of all prefixes of a word $w\in\Sigma^\omega$ is denoted
$\OpPre{w}\subseteq \Sigma^*$. 
For $L\subseteq\Sigma^*$, we have $L\subseteq \OpPre{L}$. 
%
For ${\cal L} \subseteq \Sigma^\omega$ we denote by
$\overline{\cal L}$ its complement $\Sigma^\omega\setminus {\cal L}$. 

\smallbreak
\noindent\textbf{Game Graphs and Strategies}
A \emph{two player game graph} $\Gg = \Tuple{\Qz,\Qo, \Trz,\Tro,q_0}$ consists of two finite disjoint state sets $\Qz$ and $\Qo$, two transition functions $\Trz: \Qz \rightarrow \twoup{\Qo}$ and $\Tro: \Qo \rightarrow \twoup{\Qz}$,
and an initial state $q_0\in\Qz$. We write $\Q = \Qz\cup\Qo$.
%
Given a game graph $\Gg$, a \emph{strategy} for player $0$ is a function $\fz{}: (\Qz\Qo)^*\Qz\rightarrow \Qo$; it is \emph{memoryless} if $\fz{}(\nu q^0) = \fo{q^0}$ for all $\nu \in (\Qz\Qo)^*$ and all $q^0\in \Qz$. 
A \emph{strategy} $\fo{}: (\Qz\Qo)^+\rightarrow \Qz$ for player $1$ is defined analogously.
%
The infinite sequence $\pi\in(\Qz\Qo)^\omega$ is called a play over $\Gg$ if $\pi(0)=q_0$ and for all $k\in\mathbb{N}$ holds that $\pi(2k+1)\in\Trz(\pi(2k))$ and $\pi(2k+2)\in\Tro(\pi(2k+1))$; $\pi$ is compliant with $\fz{}$ and/or $\fo{}$ if 
additionally holds that
$\fz{}(\pi|_{[0,2k]}) = \pi(2k+1)$ and/or  $\fo{}(\pi|_{[0,2k+1]}) = \pi(2k+2)$. We denote by $\Lomega(H,\fz{})$, $\Lomega(H,\fo{})$ and $\Lomega(H,\fz{},\fo{})$ the set of plays over $\Gg$ compliant with $\fz{}$, $\fo{}$, and both $\fz{}$ and $\fo{}$, respectively.

\smallbreak
\noindent\textbf{Winning Conditions}
We consider winning conditions defined over sets of states of a given game graph $\Gg$.
Given $F\subseteq \Q$, we say a play $\pi$ satisfies the \emph{B\"uchi condition} $F$
if $\OpInf{\pi}\cap F \neq \emptyset$, where $\OpInf{\pi} = \set{q\in \Q \mid \pi(k)=q \mbox{ for infinitely many }k\in\mathbb{N}}$.
Given a set $\Fc = \Set{F_1,\hdots,F_m}$, where each $F_i\subseteq \Q$, we say
a play $\pi$ satisfies the \emph{generalized B\"uchi condition} $\Fc$
if $\OpInf{\pi}\cap F_i \neq \emptyset$ for each $i\in [1;m]$.
We additionally consider generalized reactivity winning conditions with rank 1 (GR(1) winning conditions in short).
Given two generalized B\"uchi conditions $\Fsz=\Set{\Fz_1,\hdots ,\Fz_m}$ and  $\Fso=\Set{\Fo_1,\hdots, \Fo_n}$,
a play $\pi$ satisfies the GR(1) condition
if either 
$\OpInf{\pi}\cap \Fz_i = \emptyset$ for some $i\in [1;m]$
or 
$\OpInf{\pi}\cap \Fo_j \neq \emptyset$ for each $j\in [1;m]$.
That is, whenever the play satisfies $\Fsz$, it also satisfies $\Fso$.
We use the tuples $(\Gg,F)$, $(\Gg,\Fc)$ and $(\Gg,\Fsz,\Fso)$ to denote a B\"uchi, generalized B\"uchi and GR(1) game over $\Gg$, respectively, and collect all winning plays in these games in the sets $\mathcal{L}(\Gg,F)$, $\mathcal{L}(\Gg,\Fc)$ and $\mathcal{L}(\Gg,\Fsz,\Fso)$. 
A strategy $f^l{}$ is \emph{winning} for player $l$ in a B\"uchi, generalized B\"uchi, or GR(1) game, if $\Lomega(H,f^l)$ is contained in the respective set of winning plays.

\smallbreak
\noindent\textbf{Set Transformers on Games}
Given a game graph $\Gg$, we define the existential, universal, and player
$0$-, and player $1$-controllable pre-operators. Let $P\subseteq Q$.
\begin{align}
\allowdisplaybreaks
\PreE(P) =
 &\SetCompX{\qz\in\Qz}{\Trz(\qz)\cap P\neq\emptyset}
 \cup\SetCompX{\qo\in\Qo}{\Tro(\qo)\cap P\neq\emptyset},~\text{and}\label{equ:PreE}\\
 \PreA(P)=
 &\SetCompX{\qz\in\Qz}{\Trz(\qz)\subseteq P}\cup\SetCompX{\qo\in\Qo}{\Tro(\qo)\subseteq P},\label{equ:PreA}\\
\Prez(P) =
 &\SetCompX{\qz\in\Qz}{\Trz(\qz)\cap P\neq\emptyset}
 \cup\SetCompX{\qo\in\Qo}{\Tro(\qo)\subseteq P},~\text{and}\label{equ:Prez}\\
 \Preo(P)=
 &\SetCompX{\qz\in\Qz}{\Trz(\qz)\subseteq P}\cup\SetCompX{\qo\in\Qo}{\Tro(\qo)\cap P\neq\emptyset}.\label{equ:Preo}
\end{align} 
Observe that $Q \setminus\PreE(P)=\PreA(Q \setminus P)$ and 
$Q \setminus\Preo(P)=\Prez(Q \setminus P)$. 

\comment{
Given a game graph $\Gg$, we define the existential and universal pre-operators for some $P\subseteq Q$ as
\begin{align}
\allowdisplaybreaks
\PreE(P) =
 &\SetCompX{\qz\in\Qz}{\Trz(\qz)\cap P\neq\emptyset}
 \cup\SetCompX{\qo\in\Qo}{\Tro(\qo)\cap P\neq\emptyset},~\text{and}\label{equ:PreE}\\
 \PreA(P)=
 &\SetCompX{\qz\in\Qz}{\Trz(\qz)\subseteq P}\cup\SetCompX{\qo\in\Qo}{\Tro(\qo)\subseteq P},\label{equ:PreA}
\end{align} 
and observe that $Q \setminus\PreE(P)=\PreA(Q \setminus P)$. 
For $l\in\set{0,1}$, we furthermore define
the \emph{player-$l$ controllable predecessor} for some $P\subseteq Q$ by
\begin{align}
\allowdisplaybreaks
\Prez(P) =
 &\SetCompX{\qz\in\Qz}{\Trz(\qz)\cap P\neq\emptyset}
 \cup\SetCompX{\qo\in\Qo}{\Tro(\qo)\subseteq P},~\text{and}\label{equ:Prez}\\
 \Preo(P)=
 &\SetCompX{\qz\in\Qz}{\Trz(\qz)\subseteq P}\cup\SetCompX{\qo\in\Qo}{\Tro(\qo)\cap P\neq\emptyset}.\label{equ:Preo}
\end{align} 
Intuitively, $\Pre^l(P)$ computes all states from which player $l$ can
force a visit to $P$ in one step.
It holds that $Q \setminus\Preo(P)=\Prez(Q \setminus P)$. 
}

We combine the operators in \eqref{equ:PreE}-\eqref{equ:Preo} to define a \emph{conditional predecessor} 
$\OpPreTD{}$ and its dual $\OpPreTDc{}$ for sets $P,P'\subseteq \Q$ by
\begin{align}
\allowdisplaybreaks
\OpPreTD{P,\,P'}:= &\PreE(P) \cap \Preo(P\cup P'),~\text{and}\label{equ:newPre}\\
\OpPreTDc{P,\,P'}:= &\PreA(P) \cup \Prez(P\cap P').\label{equ:newPre_neg}
\end{align}
\full{
Intuitively, $\OpPreTD{}$ computes the set of states from which $P$
is reachable in one step and player $1$ can force a visit to $P\cup P'$ in one step. 
Likewise, $\OpPreTDc{}$ computes the set of states from which either
player $0$ can force a visit to $P\cap P'$ in one step or neither player can
force the game to leave $P$ in one step.
}
We see that 
$Q \setminus\OpPreTD{P,\,P'}=\OpPreTDc{Q\setminus P,\,Q \setminus P'}$.

\smallbreak
\noindent\textbf{$\mu$-Calculus}
We use the  $\mu$-calculus as a convenient logical notation used to define a symbolic algorithm
(i.e., an algorithm that manipulates sets of states rather then
individual states) for computing a set of states with a particular
property over a given game graph $\Gg$. 
The formulas of the $\mu$-calculus, interpreted over a two-player game graph $\Gg$, 
are given by the grammar
\[
\varphi\; ::= \; p \mid X \mid \varphi \cup \varphi \mid \varphi_1 \cap \varphi_2 \mid \mathit{pre}(\varphi) \mid \mu X.\varphi \mid \nu X.\varphi
\]
where $p$ ranges over subsets of $Q$, $X$ ranges over a set of formal variables,
$\mathit{pre}\in \set{\PreE,\PreA,\Prez,\Preo,\OpPreTD{},\OpPreTDc{}}$ ranges over set transformers,
and $\mu$ and $\nu$ denote, respectively, the least and greatest fixpoint of the functional
defined as $X \mapsto \varphi(X)$.
Since the operations $\cup$, $\cap$, and the set transformers $\mathit{pre}$ are all monotonic,
the fixpoints are guaranteed to exist.
A $\mu$-calculus formula evaluates to a set of states over $\Gg$, and the set can be computed
by induction over the structure of the formula, where the fixpoints are evaluated by iteration.
We omit the (standard) semantics of formulas \cite{Kozen83}.

\subsection{The Considered Synthesis Problem}

The GR(1) synthesis problem asks to synthesize a winning strategy for
the system player (player $1$) for a given GR(1) game
$(\Gg,\FcA,\FcG)$ or determine that no such strategy exists. This can
be equivalently represented in terms of $\omega$-languages, by asking
for a system strategy $\fo{}$ over $\Gg$ s.t.\
\begin{equation*}
 \emptyset\neq\Lomega(\Gg,\fo{})\subseteq \overline{\Lomega(\Gg,\FcA)} \cup \Lomega(\Gg,\FcG).
\end{equation*}
That is, the system wins on plays $\pi\in\Lomega(\Gg,\fo{})$
if either $\pi\notin\Lomega(\Gg,\FcA)$ or
$\pi\in\Lomega(\Gg,\FcA)\cap\Lomega(\Gg,\FcG)$. 
The only mechanism to ensure that \emph{sufficiently} many
computations will result from $\fo{}$ is the usage of the environment
input, which enforces a minimal branching structure.
However, the system could still win this game by \emph{falsifying the
  assumptions}; i.e., by generating plays
$\pi\notin\Lomega(\Gg,\FcA)$ that prevent the environment from
fulfilling its liveness properties.  

We suggest an alternative view to the usage of the assumptions on the
environment $\FcA$ in a GR(1) game.
The condition $\FcA$ can be interpreted abstractly as modeling an
underlying mechanism that ensures that the environment player (player
$0$) generates only inputs (possibly in response to observed outputs)
that conform with the given assumption.
In this context, we would like to ensure that the system (player $1$) allows the
environment, as much as possible, to fulfill its liveness and only \emph{restricts} the environment behavior 
if needed to enforce the guarantees. 
%
We achieve this by forcing the system player to ensure that the
environment is always able to play such that it fulfills its liveness, i.e.
\begin{equation*}
 \OpPre{\Lomega(\Gg,\fo{})} = \OpPre{\Lomega(\Gg,\fo{})\cap\Lomega(\Gg,\FcA)}.
\end{equation*}
As the $\supseteq$-inclusion trivially holds, the constraint is given by the
$\subseteq$-inclusion.
Intuitively, the latter holds if every finite
play $\alpha$ compliant with $\fo{}$ over $\Gg$
can be extended (by a suitable environment strategy) to an infinite play $\pi$ compliant with
$\fo{}$ that fulfills the environment liveness assumptions. 
It is easy to see that not every solution to the GR(1) game
$(\Gg,\FcA,\FcG)$ (in the classical sense) supplies this additional
requirement.
We therefore propose to synthesize a system strategy $\fo{}$ with the
above properties, as summarized in the following problem statement.

\begin{problem}\label{PS}
Given a GR(1) game $(\Gg,\FcA,\FcG)$ synthesize a system strategy
$\fo{}$
\begin{subequations}\label{equ:PS}
 \begin{align}
 &\text{s.t.}\quad \emptyset\neq\Lomega(\Gg,\fo{})\subseteq \overline{\Lomega(\Gg,\FcA)} \cup \Lomega(\Gg,\FcG),\label{equ:LanguageSpecOld}\\
  &\text{and}\quad  \OpPre{\Lomega(\Gg,\fo{})} = \OpPre{\Lomega(\Gg,\fo{})\cap\Lomega(\Gg,\FcA)}\label{eq:SCT_nonblocking}
\end{align}
\end{subequations}
both hold, or verify that no such system strategy exists. \EndBox
\end{problem}

\REFproblem{PS} asks for a strategy $\fo{}$ s.t.\ every play $\pi$ compliant with $\fo{}$ over $\Gg$ fulfills the system
guarantees, i.e., $\pi\in\Lomega(\Gg,\FcG)$, if the environment
fulfills its liveness properties, i.e., if
$\pi\in\Lomega(\Gg,\FcA)$ (from \eqref{equ:LanguageSpecOld}),
while the latter always remains possible (by a suitably playing environment) due to
\eqref{eq:SCT_nonblocking}.
Inspired by algorithms solving the supervisory controller synthesis
problem for non-terminating processes \cite{RamadgeOmega,Thistle94},
we propose a solution to \REFproblem{PS} in terms of a vectorized
4-nested fixed-point in the remaining part of this paper. We show
that \REFproblem{PS} can be solved by a finite-memory strategy, if a
solution exists.

\new{
We note that \eqref{eq:SCT_nonblocking} is not a linear time
but a branching time property and can therefore not be 
``compiled away'' into a different GR(1) or even $\omega$-regular objective.
Satisfaction of \eqref{eq:SCT_nonblocking} requires checking whether the set $\FA$ remains reachable from any reachable state in the game graph realizing $\Lomega(\Gg,\fo{})$%
\footnote{It can indeed be expressed by the CTL$^*$ formular $\mathsf{AGEF}\FA$ (see \cite{EhlersJdeds}, Sec. 3.3.2).}%
.
%
}
\fullb
This is made clear by the example in \REFfig{fig:branching}. The game graph $\Gg'$ (\REFfig{fig:branching}, left) realizes a language $\Lomega(\Gg,\fo{})$ which is non-conflicting for $\FA=\set{q_5}$ as $q_5$ is reachable from all states in $\Gg'$. However, reducing this language to the single trace $q_0 q_1 (q_2 q_3 q_4)^\omega$ realized by the game graph $\Gg''$ (\REFfig{fig:branching}, right) shows that the property does not hold anymore. Hence, non-conflictingness is not a trace property.
 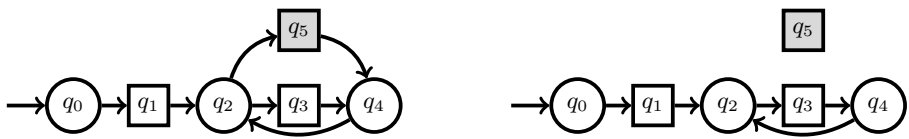
\begin{figure}
 \begin{center}
    \begin{tikzpicture}[auto,scale=1]
     
	    \node (init) at (-1,0) {};
	    \node[mystate] (q0) at (0,0) {$q_0$};
	    \node[draw,mysquare] (q1) at (1,0) {$q_1$};        
	    \node[mystate] (q2) at (2,0) {$q_2$};
	    \node[draw,mysquare] (q3) at (3,0) {$q_3$};        
	    \node[mystate] (q4) at (4,0) {$q_4$};
          \node[draw,mysquare, fill=black!15] (q5) at (3,1) {$q_5$};
          
%
%
%
%
\SFSAutomatEdge{init}{}{q0}{semithick}{}  
\SFSAutomatEdge{q0}{}{q1}{}{}  
\SFSAutomatEdge{q1}{}{q2}{}{}
\SFSAutomatEdge{q2}{}{q3}{}{}  
\SFSAutomatEdge{q3}{}{q4}{}{}
\SFSAutomatEdge{q4}{}{q2}{bend left}{}
\SFSAutomatEdge{q2}{}{q5}{bend left}{}
\SFSAutomatEdge{q5}{}{q4}{bend left}{}

     \end{tikzpicture}
     \hspace{1cm}
     \begin{tikzpicture}[auto,scale=1]
     
          \node (init) at (-1,0) {};
          \node[mystate] (q0) at (0,0) {$q_0$};
          \node[draw,mysquare] (q1) at (1,0) {$q_1$};        
          \node[mystate] (q2) at (2,0) {$q_2$};
          \node[draw,mysquare] (q3) at (3,0) {$q_3$};        
          \node[mystate] (q4) at (4,0) {$q_4$};
           \node[draw,mysquare, fill=black!15] (q5) at (3,1) {$q_5$};

\SFSAutomatEdge{init}{}{q0}{semithick}{}  
\SFSAutomatEdge{q0}{}{q1}{}{}  
\SFSAutomatEdge{q1}{}{q2}{}{}
\SFSAutomatEdge{q2}{}{q3}{}{}  
\SFSAutomatEdge{q3}{}{q4}{}{}
\SFSAutomatEdge{q4}{}{q2}{bend left}{}
     \end{tikzpicture}
 \end{center}
  \caption{Two game graphs $\Gg'$ (left) and $\Gg''$ (right). $\Gg'$ realizes an example language $\Lomega(\Gg,f^{1})$ which is non-conflicting for $\FA=\set{q_5}$ (indicated in gray). The right side shows a game graph $\Gg''$ realizing the sub language $\Lomega(\Gg'')=q_0 q_1 (q_2 q_3 q_4)^\omega\subset\Lomega(\Gg,f^{1})$. 
  }\label{fig:branching}
 \end{figure}
\fulle

\section{Algorithmic Solution for Singleton Winning Conditions}\label{sec:AlgoS}
We first consider the GR(1) game $(\Gg,\FcA,\FcG)$ with singleton winning
conditions $\FcA=\Set{\FA}$ and $\FcG=\Set{\FG}$, i.e., $n=m=1$.
It is well known that a system winning strategy $\fo{}$ for this game can be synthesized by
solving a three color parity game over $\Gg$.
This can
be expressed by the $\mu$-calculus formula  
(see \cite{EmersonJutla_1991}) 
\begin{align}\label{equ:3nestedFP}
\varphi_3:=\nu Z~.~\mu Y~.~\nu X~.~ (\FG\cap \Preo(Z)) \cup \Preo(Y) \cup (Q\setminus \FA\cap \Preo(X)).
\end{align}
It follows that $q_0\in \semantics{\varphi_3}$ if and only if the
synthesis problem has a solution and the winning strategy $\fo{}$ is obtained from a ranking
argument over the sets computed during the evaluation of
\eqref{equ:3nestedFP}.

To obtain a system strategy $\fo{}$ solving \REFproblem{PS} instead, we propose to extend \eqref{equ:3nestedFP} to a 4-nested fixed-point expressed by the $\mu$-calculus formula 
\begin{equation}\label{equ:new_FP}
  \begin{array}{l l}
  \varphi_4= & \nu Z~.~\mu Y~.~\nu X~.~\mu W~.\hfill\\
  & \multicolumn{1}{r}{
  (\FG\cap \Preo(Z)) ~\cup~
    \Preo(Y) ~\cup~ ((Q\setminus \FA)\cap \OpPreTD{W, X \setminus \FA}) }.
  \end{array}
\end{equation}

\newb
Compared to \eqref{equ:3nestedFP} this adds an inner-most largest fixed-point and
substitutes the last controllable pre-operator by the conditional
one. Intuitively, this distinguishes between 
states from which player $1$ can force visiting $\FG$ and states
from which player $1$ can force avoiding $\FA$. This is in contrast to 
\eqref{equ:3nestedFP} and allows to exclude strategies that allow
player $1$ to win by falsifying the assumptions. 
\full{This is further explained when discussing the example in \REFfig{fig:gamegraph1}.}

\newe

%
The remainder of this section shows that $q_0\in \semantics{\varphi_4}$ if and only if \REFproblem{PS} has a solution and the winning strategy $\fo{}$
fulfilling \eqref{equ:PS} can be obtained from a ranking
argument over the sets computed during the evaluation of
\eqref{equ:new_FP}.

\smallbreak
\noindent\textbf{Soundness}

\noindent
\label{sec:AlgoS:Soundness}
We prove soundness of \eqref{equ:new_FP} by showing that every state
$q\in\semantics{\varphi_4}$ 
is winning for the system player.
In view of \REFproblem{PS} this requires to show that there exists a system strategy $\fo{}$ s.t.\ all plays starting
in a state $q\in \semantics{\varphi_4}$ and evolving in accordance to
$\fo{}$ result in an infinite play that fulfills
\eqref{equ:LanguageSpecOld} and \eqref{eq:SCT_nonblocking}.  
%

%
We start by defining $\fo{}$ from a ranking argument over
the iterations of \eqref{equ:new_FP}.
Consider the last iteration of the fixed-point in \eqref{equ:new_FP}
over $Z$. As \eqref{equ:new_FP} terminates after this iteration we have $Z=Z^\infty=\semantics{\varphi_4}$.
Assume that the fixed point over $Y$ is reached after $k$ iterations. 
If $Y^i$ is the set obtained after the $i$-th iteration, we have that
$Z^\infty=\bigcup_{i=0}^k Y^i$ with $Y^i\subseteq Y^{i+1}$,  
$Y^0=\emptyset$ and $Y^k=Z^\infty$. 
Furthermore, let $X^i=Y^i$ denote the fixed-point of the iteration
over $X$ resulting in $Y^i$ and denote by $W^i_j$ the set obtained in
the $j$th iteration  
over $W$ performed while using the value $X^i$ for $X$ and $Y^{i-1}$
for $Y$.
Then it holds that $Y^i=X^i=\bigcup^{l_i}_{j=0} W_j^i$ with
$W_j^i\subseteq W_{j+1}^i$, $W_0^i=\emptyset$ and $W_{l_i}^i=Y^i$ for
all $i\in[0;k]$.  

%
%
Using these sets, we define a ranking for every state $q\in Z^\infty$
s.t.\ 
\begin{equation}\label{equ:ranking}
 \rank(q)=(i,j) ~\text{iff}~q\in \BR{Y^i\setminus
   Y^{i-1}}\cap\BR{W^i_{j}\setminus W^i_{j-1}}~\text{for}~i,j>0. 
\end{equation}
We order ranks lexicographically. 
\new{It further holds that (see \REFapp{sec:proof:single:soundness})}
\begin{subequations}\label{equ:rank_prop}
\begin{align}
 q\in D&~\Leftrightarrow~\rank(q)=(1,1)~&&~\Leftrightarrow~q\in\FG\cap  Z^\infty\label{equ:rankFG}\\
 q\in E^i&~\Leftrightarrow~\propConj{\rank(q)=(i,1)}{i>1}\hspace{-0.5cm}~&&~\Leftrightarrow~q\in(\FA\setminus\FG)\cap  Z^\infty\label{equ:rankFA}\\
 q\in R^i_j&~\Leftrightarrow~\propConj{\rank(q)=(i,j)}{j>1}\hspace{-0.5cm}~&&~\Leftrightarrow~q\in(Z^\infty\setminus(\FA\cup\FG)),\label{equ:ranknoR}
\end{align}
\end{subequations}
where $D$, $E^i$ and $R^i_j$ denote the sets \emph{added} to the winning state set by the first, second and third term of \eqref{equ:new_FP}, respectively, in the corresponding iteration. 
 
%

\begin{figure}[t]
 \begin{center}
   \input{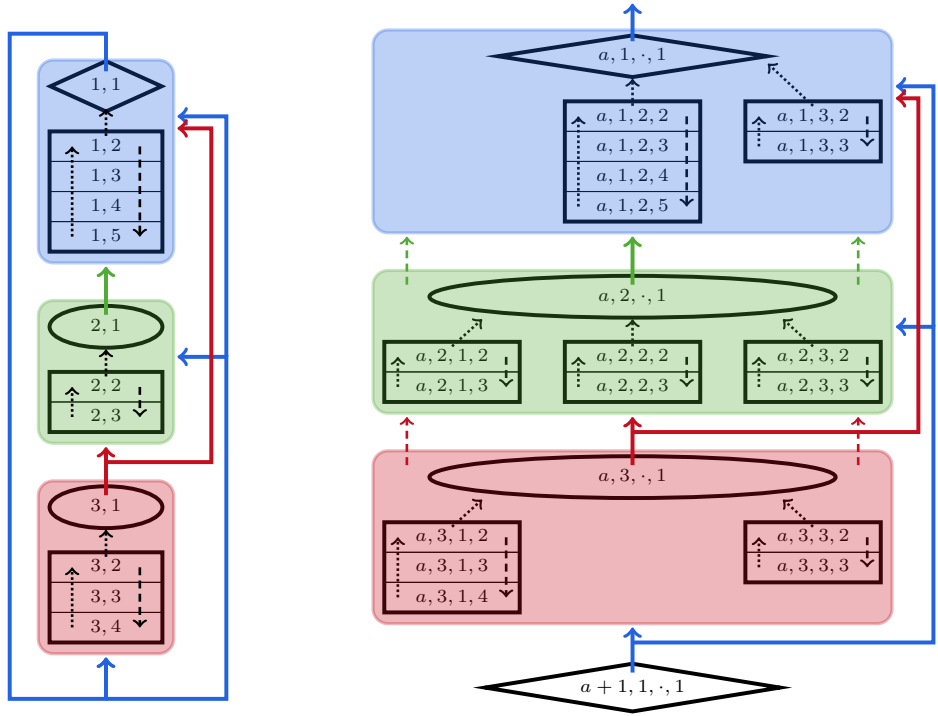}
 \end{center}
 \vspace*{-7mm}
  \caption{Schematic representation of the ranking defined in \eqref{equ:ranking} (left) and in \eqref{equ:ranking_a} (right). 
  Diamond, ellipses and rectangles represent the sets $D$, $E^i$ and $R^i_j$, while blue, green and red indicate the sets $Y^1$, $Y^2\setminus Y^1$ and $Y^3\setminus Y^2$ (annotated by $\ps{a}$/$\ps{ab}$ for the right figure).
  Labels $(i,j)$ and $(a,i,b,j)$ indicate that all states $q$ associated with this set fulfill $\rank(q)=(i,j)$ and $\ps{ab}\rank(q)=(i,j)$, respectively. Solid, colored arcs indicate system-enforceable moves, dotted arcs indicate existence of environment or system transitions and dashed arcs indicate possible existence of environment transitions.
  }\label{fig:strategy} 
 \vspace*{-6.5mm}
\end{figure}

\REFfig{fig:strategy} (left) shows a schematic representation of this construction for an example with $k=3$, $l_1=4$, $l_2=2$ and $l_3=3$. 
The set $D=\FG\cap  Z^\infty$ is represented by the diamond at the top where the label $(1,1)$ denotes the associated rank (see \eqref{equ:rankFG}). The ellipses represent the sets $E^i\subseteq(\FA\setminus\FG)\cap  Z^\infty$, where the corresponding $i>1$ is indicated by the associated rank $(i,1)$.  Due to the use of the controllable pre-operator in the first and second term of \eqref{equ:new_FP}, it is ensured that progress out of $D$ and $E^i$ can be enforced by the system, indicated by the solid arrows. 
This is in contrast to all states in $R^i_j\subseteq Z^\infty\setminus\FA\setminus\FG$, which are represented by the rectangular shapes in \REFfig{fig:strategy} (left). These states allow the environment to increase the ranking (dashed lines) as long as $Z^\infty\setminus\FA\setminus\FG$ is not left and there exists a possible move to decrease the $j$-rank (dotted lines).
While this does not strictly enforce progress, we see that whenever the environment plays such that states in $\FA$ (i.e., the ellipses) are visited infinitely often (i.e., the environment fulfills its assumptions), the system can enforce progress w.r.t.\ the defined ranking and states in $\FG$ (i.e., the diamond shape) is eventually visited. 
The system is restricted to take the existing solid or dotted transitions in \REFfig{fig:strategy} (left). With this, it is easy to see that the constructed strategy is winning if the environment fulfills its assumptions, i.e., \eqref{equ:LanguageSpecOld} holds. However, to ensure that \eqref{eq:SCT_nonblocking} also holds, we need an additional requirement. This is necessary as the used construction also allows plays to cycle through the blue region of \REFfig{fig:strategy} (left) only, and by this not surely visiting states in $\FA$ infinitely often. However, if $\Lomega(\Gg,\FG)\subseteq\Lomega(\Gg,\FA)$ we see that \eqref{eq:SCT_nonblocking} holds as well. It should be noted that the latter is a sufficient condition which can be easily checked symbolically on the problem instance but not a necessary one.

Based on the ranking in \eqref{equ:ranking} we define a memory-less system strategy $\fo{}:\Qo\cap Z^\infty\rightarrow \Qz\subseteq\Tro{}$ s.t.\ the rank is always decreased, i.e.,
\begin{equation}\label{equ:rank_new}
 \propImp{q'=\fo{q}}{
 \DiCases{\rank(q')<\rank(q)}{\rank(q)>(1,1)}{q'\in
   Z^\infty}{\text{otherwise}} 
 }.
\end{equation}
The next theorem shows that this strategy indeed solves \REFproblem{PS}.

\begin{theorem}\label{thm:Soundness_single}
 Let $(\Gg,\FcA,\FcG)$ be a GR(1) game with singleton winning
 conditions $\FcA=\Set{\FA}$ and $\FcG=\Set{\FG}$. Suppose $\fo{}$ is
 the system strategy in \eqref{equ:rank_new} based
 on the ranking in \eqref{equ:ranking}.
  Then it holds for all $q\in\semantics{\varphi_4}$ that%
 \footnote{Given a state $q\in Q=\Qz\cup\Qo$ we use the subscript $q$
   to denote that the respective set of plays is defined by using $q$ as
   the initial state of $\Gg$.}%
 \begin{subequations}\label{equ:Soundness}
 \begin{align}
  &\Lomega_q(\Gg,\fo{}) \subseteq \overline{\Lomega_q(\Gg,\FcA)} \cup \Lomega_q(\Gg,\FcG),\label{equ:Soundness:infinite}\\
  &\Lomega_q(\Gg,\fo{})\cap\Lomega_q(\Gg,\FcG)\neq\emptyset,~\text{and}\label{equ:Soundness:nonempty}\\
  &
  \Lomega_q(\Gg,\FcG)\rs\subseteq\rs\Lomega_q(\Gg,\FcA)
   \Rightarrow
   \OpPre{\Lomega_q(\Gg,\fo{})}\rs=\rs\OpPre{\Lomega_q(\Gg,\fo{})\rs\cap\rs\Lomega_q(\Gg,\FcA)}\rs.\label{equ:Soundness:nonblock}
 \end{align}
 \end{subequations}
\end{theorem}



\smallbreak
\noindent\textbf{Completeness}

\noindent
\label{sec:AlgoS:Completeness}
We show completeness of \eqref{equ:new_FP} by establishing that
every state $q\in Q\setminus\semantics{\varphi_4}=\semantics{\overline{\varphi}_4}$ is losing for the system player. In view of \REFproblem{PS} this requires to show that for all $q\in\semantics{\overline{\varphi}_4}$ and all system strategies $\fo{}$ either \eqref{equ:LanguageSpecOld} or \eqref{eq:SCT_nonblocking} does not hold. 
This is formalized 
in \REFapp{sec:proof:completeness} 
by first negating the fixed-point in \eqref{equ:new_FP} and deriving the induced ranking of this negated fixed-point. 
Using this ranking, we first show that the environment can 
(i) render the negated winning set $\overline{Z}^\infty$ invariant 
and 
(ii) can always enforce the play to visit $\FG{}$ only finitely often, resulting in a violation of the guarantees.
Using these observations we finally show that whenever \eqref{equ:LanguageSpecOld} holds for an arbitrary system 
strategy $\fo{}$ starting in $\semantics{\overline{\varphi}_4}$, then \eqref{eq:SCT_nonblocking} cannot hold. 
With this, completeness, as formalized in the following theorem, directly follows.

\begin{theorem}\label{thm:Completeness_single}
 Let $(\Gg,\FcA,\FcG)$ be a GR(1) game with singleton winning
 conditions $\FcA=\Set{\FA}$ and $\FcG=\Set{\FG}$.
 Then it holds for all $q\in \semantics{\overline{\varphi}_4}$ and all system strategies $\fo{}$ over $\Gg$ that either 
\begin{subequations}
 \begin{align}
&\emptyset\neq\Lomega_q(\Gg,\fo{})\subseteq \overline{\Lomega_q(\Gg,\FcA)} \cup \Lomega_q(\Gg,\FcG),~\text{or}\label{equ:LanguageSpecOld_q}\\
  \quad&\OpPre{\Lomega_q(\Gg,\fo{})} = \OpPre{\Lomega_q(\Gg,\fo{})\cap\Lomega_q(\Gg,\FcA)}~\text{does not hold.}\label{eq:SCT_nonblocking_q}
 \end{align}
\end{subequations}
\end{theorem}


\begin{figure}[h!]
   \begin{center}
     \begin{tikzpicture}[auto,scale=1.5]
     
          \node (init) at (-0.5,0) {};
          \node[mystate,fill=black!15] (q0) at (0,0) {$q_0$};
          \node[draw,mysquare] (q1) at (2,0) {$q_1$};        
          \node[mystate] (q2) at (2.75,-0.75) {$q_2$};
          \node[draw,mysquare] (q3) at (3.75,-0.75) {$q_3$};
          \node[draw,fmysquare,fill=black!15] (q4) at (2,-1.5) {$q_4$};
          \node[mystate,fill=black!15] (q5) at (0.75,-1.5) {$q_5$};
          \node[draw,mysquare] (q6) at (0,-0.75) {$q_6$};
          \node[mystate,fill=black!15] (q7) at (-1,-0.75) {$q_7$};
          \node[mystate] (q8) at (2.75,0.75) {$q_8$};
          \node[draw,mysquare] (q9) at (3.75,0.75) {$q_9$};
          
          \begin{tiny}
          \node (q0a) at (0.2,0.3) {$2,1$};
          \node (q1a) at (2.25,0.25) {$1,3$};        
          \node (q2a) at (2.95,-0.45) {$1,2$};
          \node (q3a) at (3.95,-0.5) {$1,3$};
          \node (q4a) at (2.2,-1.2) {$1,1$};
          \node (q5a) at (0.95,-1.2) {$3,1$};
          \node (q6a) at (0.2,-0.5) {$2,2$};
          \node (q7a) at (-0.8,-0.45) {$3,1$};
          \end{tiny}
          
          \begin{tiny}
          \node (q0b) at (0.2,-0.3) {$2$};
          \node (q1b) at (2.25,-0.3) {$1$};        
          \node (q2b) at (2.95,-1.05) {$1$};
          \node (q3b) at (3.95,-1.05) {$1$};
          \node (q4b) at (2.2,-1.8) {$1$};
          \node (q5b) at (0.95,-1.8) {$3$};
          \node (q6b) at (0.2,-1.05) {$2$};
          \node (q7b) at (-0.8,-1.05) {$3$};
          \node (q8b) at (2.95,0.45) {$1$};
          \node (q9b) at (3.95,0.45) {$1$};
          \end{tiny}

\SFSAutomatEdge{init}{}{q0}{semithick}{}  
\SFSAutomatEdge{q0}{}{q1}{}{}  
\SFSAutomatEdge{q1}{}{q8}{bend left}{}
\SFSAutomatEdge{q1}{}{q2}{bend right}{}
\SFSAutomatEdge{q8}{}{q9}{bend left}{}
\SFSAutomatEdge{q9}{}{q8}{bend left}{}
\SFSAutomatEdge{q2}{}{q3}{bend left}{} 
\SFSAutomatEdge{q3}{}{q2}{bend left}{} 
\SFSAutomatEdge{q2}{}{q4}{bend left}{}
\SFSAutomatEdge{q4}{}{q5}{}{}
\SFSAutomatEdge{q5}{}{q6}{bend left}{}
\SFSAutomatEdge{q6}{}{q7}{bend left}{}
\SFSAutomatEdge{q7}{}{q6}{bend left}{}
\SFSAutomatEdge{q6}{}{q0}{}{}

     \end{tikzpicture}
\end{center}
\caption{Game graph $H_1$. States in $\Qz$ and $\Qo$ are indicated by circles and squares, respectively. States in $\FA$ and $\FG$ are indicated by light gray with a single and a double boundary, respectively. The small numbers on the top right and bottom right of each state indicate the ranking induced for all states in $Z^\infty$ by \eqref{equ:new_FP} and \eqref{equ:3nestedFP}, respectively. }\label{fig:gamegraph1}
\end{figure}
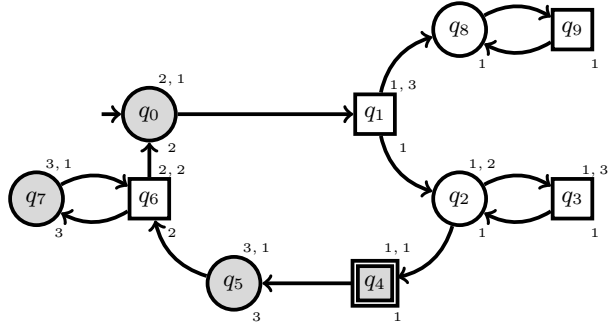

\fullb
\smallbreak
\noindent\textbf{Example}

\noindent
Consider the game graph $H_1$ in \REFfig{fig:gamegraph1}. Running the fixed-point in \eqref{equ:new_FP} for $H_1$ induces the ranking defined in \eqref{equ:ranking} as indicated in the top right of every winning state. Here, the evaluation of the fixed-point is particularly simple as the smallest fixed-points over $X$ and $Z$ never remove states; we therefore concentrate on the maximal fixed-points over $W$ and $Y$. In the first iteration over $W$, we start with $\FG$ (i.e., $q_4$) and successively enlarge this set by states that can reach $W$ (i.e., have a path to $\FG$) and can be forced by player $1$ to stay within $Q\setminus\FA$. This is true for all states with rank $(1,\cdot)$	, i.e., $q_1$ to $q_4$, giving $D=\set{q_1}$, $R^1_2=\set{q_2}$ and $R^1_3=\set{q_1,q_3}$. It is easy to see, that the environment can increase the rank during a play by going from $q_2$ to $q_3$ (i.e., moving from $R^1_2$ to $R^1_3$). However, whenever the environment fulfills its liveness property, it has to eventually transition to $q_4$ and, hence, the ranking is only increased finitely often. 

After this local state set is constructed, the pre-operator over $Y$ adds the assumption state $q_0$ to the fixed-point, indicated by the rank $(2,1)$ and resulting in $E^2=\set{q_0}$. Running the new fixed-point over $W$ now only adds $q_6$, giving $R^2_1=\set{q_6}$. Finally, $q_5$ and $q_7$ are added by the next iteration over $Y$, indicated by the rank $(3,1)$ and resulting in $E^3=\set{q_5,q_7}$. 

It is easy to see that  $q_8$ and $q_9$ are never added to the winning region, as they do not have a path to any $W$ constructed during the iteration over \eqref{equ:new_FP}, i.e., do not allow to reach $\FG$. By this, the strategy induced by this ranking via \eqref{equ:rank_new} always transitions from $q_1$ to $q_2$ and from $q_6$ to $q_0$, thereby avoiding to win by falsifying the assumptions.

Now consider the fixed-point in \eqref{equ:3nestedFP}, which induces a ranking over $Y$ as indicated in the bottom right of every winning state in \REFfig{fig:gamegraph1} (see \cite{Bloem_etal_2012} for a definition of the used ranking). Due to the missing inner fixed-point over $W$, the first iteration over $X$ is initialized directly with $\FG\cup(Q\setminus\FA)$, resulting in $Y^1=\set{q_1,\hdots,q_4,q_7,q_8}$. While the remaining iterations over $Y$ result in an equivalent $i$ ranking as in the new 4-nested fixed-point \eqref{equ:new_FP}, we see that $q_8$ and $q_9$ are now part of the winning region. Even worse, due to the structure of \eqref{equ:3nestedFP}, $q_8$ and $q_2$ have the same rank. I.e., the rank does not allow to distinguish between 
states from which player $1$ can force a visit to $\FG$ and states
from which player $1$ can force the play to stay inside
$Q\setminus\FA$. Therefore, is not possibly to construct a strategy via this ranking that avoids winning by falsifying the assumptions.

\fulle

\smallbreak
\noindent\textbf{A Solution for \REFproblem{PS}}

\noindent
\label{sec:AlgoS:PS}
We note that the additional assumption in 
\REFthm{thm:Soundness_single} is required only to ensure that the
resulting strategy fulfills \eqref{eq:SCT_nonblocking}.
Suppose that this assumption holds for the initial state $q_0$ of
$\Gg$.
That is, consider a GR(1) game $(\Gg,\FcA,\FcG)$ with singleton
winning conditions $\FcA=\Set{\FA}$ and $\FcG=\Set{\FG}$ s.t.\  
 $\Lomega(\Gg,\FG)\subseteq\Lomega(\Gg,\FA)$.
Then it follows from \REFthm{thm:Completeness_single} that
\REFproblem{PS} has a solution iff $q_0\in
\semantics{\varphi_4}$. Furthermore, if $q_0\in \semantics{\varphi_4}$,
based on the intermediate values maintained for the computation of
$\varphi_4$ in \eqref{equ:ranking} and the ranking defined in
\eqref{equ:rank_new}, we can construct $\fo{}$
that wins the GR(1) condition in \eqref{equ:LanguageSpecOld} and is
non-conflicting, as in \eqref{eq:SCT_nonblocking}. 

We can check symbolically whether
$\Lomega(\Gg,\FG)\subseteq\Lomega(\Gg,\FA)$. 
For this we
construct a game graph $\Gg'$ from $\Gg$ by removing all states in
$\FA$, and then check whether $\Lomega(\Gg',\FG)$ is empty. The latter
is decidable in logarithmic space and polynomial time. 
%
If this check fails, then
$\Lomega(\Gg,\FG)\not\subseteq\Lomega(\Gg,\FA)$.
Furthermore, we can replace $\Lomega(\Gg,\FcG)$ in
\eqref{equ:LanguageSpecOld} 
by $\Lomega(\Gg,\FcG)\cap\Lomega(\Gg,\FcA)$ without affecting the
restriction \eqref{equ:LanguageSpecOld} imposes on the choice of
$\fo{}$. Given singleton winning conditions $\FG$ and $\FA$, we see
that $\Lomega(\Gg,\FG)\cap\Lomega(\Gg,\FA)=\Lomega(\Gg,\set{\FG,\FA})$
and it trivially holds that
$\Lomega(\Gg,\set{\FG,\FA})\allowbreak\subseteq\Lomega(\Gg,\FA)$.
That is,
we fulfill the conditional by replacing the system guarantee $\Lomega(\Gg,\FcG)$ by
$\Lomega(\Gg,\set{\FG,\FA})$. However,
this results in a GR(1) synthesis problem with $m=1$ and
$n=2$, which we discuss next.

\section{Algorithmic Solution for GR(1) Winning Conditions}\label{sec:AlgoV}

We now consider a general GR(1) game $(\Gg,\FcA,\FcG)$ with
$\FcA=\Set{\FAa{1},\hdots,\FAa{m}}$ and
$\FcG=\Set{\FGa{1},\hdots,\FGa{n}}$ s.t.\ $n,m>1$. 
The known fixed-point for solving GR(1) games in \cite{Bloem_etal_2012} rewrites the three nested fixed-point in \eqref{equ:3nestedFP} in a vectorized version, which induces an order on the guarantee sets in $\FcG$ and adds a disjunction over all assumption sets in $\FcA$ to every line of this vectorized fixed-point. Adapting the same idea to the 4-nested fixed-point algorithm \eqref{equ:new_FP} results in 
\begin{align}\label{equ:4FP_vector}
\varphi_4~=~ &\nu 
\begin{bmatrix}
\Za{1}\\ \Za{2}\\ \vdots \\ \Za{n} 
\end{bmatrix}.
\begin{bmatrix}
 \mu~\Ya{1}~.~\BR{\bigvee_{b=1}^m~\nu~\Xa{1b}~.~\mu~\Wa{1b}~\ps{1b}\Omega}\\
 \mu~\Ya{2}~.~\BR{\bigvee_{b=1}^m~\nu~\Xa{2b}~.~\mu~\Wa{2b}~\ps{2b}\Omega}\\
 \vdots\\
\mu~\Ya{n}~.~\BR{\bigvee_{b=1}^m~\nu~\Xa{nb}~.~\mu~\Wa{nb}~\ps{nb}\Omega}
\end{bmatrix},
\end{align}
where, $\ps{ab}\Omega=(\FGa{a}\cap
  \Preo(\Za{a^+})) \cup \Preo(\ps{a}Y) \cup (Q\setminus \FAa{b}\cap
  \OpPreTD{W, X \setminus \FAa{b}})\notag$ and 
${a^+}$ denotes ${(a\mod n)+1}$.
  

The remainder of this section shows how soundness and completeness carries over from the 4-nested fixed-point algorithm \eqref{equ:new_FP} to its vectorized version in \eqref{equ:4FP_vector}.

\smallbreak
\noindent
    {\bf Soundness and Completeness}

\noindent
    We refer to intermediate sets obtained during the computation of the
fixpoints by similar notations as in
Section~\ref{sec:AlgoS:Soundness}.
For example, the set $\ps{a}Y^i$ is the $i$-th approximation of the
fixpoint computing $\ps{a}Y$ and $\ps{ab}W^i_j$ is the $j$-th
approximation of $\ps{ab}W$ while computing the $i$-th
approximation of $\ps{a}Y$, i.e., computing $\ps{a}Y^i$ and using
$\ps{a}Y^{i-1}$. 
Similar to the above, we define a mode-based rank for every state
$q\in \Za{}^\infty$; we 
track the currently chased guarantee $a\in [1;n]$ (similar to
\cite{Bloem_etal_2012}) and the currently avoided assumption set
$b\in[1,m]$ as an additional internal mode. In analogy to
\eqref{equ:ranking} we define 
\begin{equation}\label{equ:ranking_a}
 \ps{ab}\rank(q)=(i,j) ~\text{iff}~q\in \BR{\Ya{}^i\setminus \Ya{}^{i-1}}\cap\BR{\Wa{ab}^i_{j}\setminus \Wa{ab}^i_{j-1}}~\text{for}~i,j>0.
\end{equation}
Again, we order ranks lexicographically, and, in analogy to \eqref{equ:rank_prop}, we have%

\begin{subequations}\label{equ:arank}
\begin{align}
\allowdisplaybreaks
q\in
\ps{a}D&~\Leftrightarrow~\ps{a\cdot}\rank(q)=(1,1)
&&~\Rightarrow \new{q\in\FGa{}}
,\label{equ:arankFG}\\
q\in \ps{a}E^i&~\Leftrightarrow~\propConj{\ps{a\cdot}\rank(q)=(i,1)}{i\rs>\rs1}
,\label{equ:arankFA}\\
q\in \ps{ab}R^i_j&~\Leftrightarrow ~ \propConj{\ps{ab}\rank(q)=(i,j)}{j>1}\hspace{-0.5cm}&&~\Rightarrow~\new{q\notin\FAa{b}}\label{equ:arankFA_without}.
\end{align}
\end{subequations}
The sets $\Ya{}^i$, $\Wa{ab}^i_{j}$, $\ps{a}D$, $\ps{a}E^i$ and $\ps{ab}R^i_j$ are interpreted in direct analogy to \REFsec{sec:AlgoS:Soundness}, where $a$ and $b$ annotate the used line and conjunct in \eqref{equ:4FP_vector}. 

\REFfig{fig:strategy} (right) shows a schematic representation of the
ranking for an example with $\ps{a}k=3$, $\ps{a1}l_1=0$,
$\ps{a2}l_1=4$, $\ps{a3}l_1=2$, $\ps{a\cdot}l_2=2$, $\ps{a1}l_3=3$,
$\ps{a2}l_3=0$, and $\ps{a3}l_3=2$. Again, the set \new{$\ps{a}D\subseteq\FGa{}$} is represented by the diamond at the top of the figure.
%
Similarly, all ellipses represent sets
\new{$\ps{a}E^i$}
added in the $i$-th iteration over line $a$ of
\eqref{equ:4FP_vector}. 
%
Again, progress out of ellipses can be enforced by the system,
indicated by the solid arrows leaving those shapes. However, this
might not preserve the current $b$ mode. It might be the environment
choosing which assumption to avoid next%
.
\new{
Further, the environment might choose to change the $b$ mode along with decreasing the $i$-rank, as indicated by the colored dashed lines
\footnote{\new{The strategy extraction in \eqref{equ:rank_new_a} prevents the system from choosing a different $b$ mode. The strategy choice could be optimized w.r.t.\ fast
   progress towards $\FGa{}$ in such cases.}}%
  .
\full{This is possible as for $i>1$ we have $\ps{a}Y^{i-1}=\bigcup_{b'\in[1,m]}\ps{a,b'}W^{i-1}\subseteq \ps{ab}X^i$
and is further explained when discussing the example in \REFfig{fig:gamegraph2}}.
  }
Finally, the interpretation of the sets represented by rectangular shapes in \REFfig{fig:strategy} (right), corresponding to \eqref{equ:arankFA_without}, is in direct analogy to the case with singleton winning conditions. It should be noticed that this is the only place where we preserve the current $b$-mode when constructing a strategy.

Using this intuition we define a system strategy that uses enforceable
and existing transitions to decrease the rank if possible and
preserves the current $a$ mode until the diamond shape is reached. The
$b$ mode is only preserved within rectangular sets.  This is
formalized by a strategy 
\begin{subequations}\label{equ:rank_new_a}
 \begin{equation}
  \textstyle\fo{}:\bigcup_{a\in[1;n]} \BR{(\Q^1\cap \Za{}^\infty)\times a \times [1;m]}\rightarrow \Q^0\times [1;n]\times [1;m]
 \end{equation}
s.t.\ $(q',\cdot,\cdot)=\fo{q,\cdot,\cdot}$ implies $q'\in\Tr{1}(q)$ and $(q',a',b')=\fo{q,a,b}$ implies
\begin{equation}
 \TriCases
 {q'\in \ps{a^+}Z^\infty \wedge a'=a^+}{\ps{ab}\rank(q)=(1,1)}
 {\ps{a'b'}\rank(q')\leq(i-1,\cdot) \wedge a'=a}{\ps{ab}\rank(q)=(i,1),i>1}
 {\ps{a'b'}\rank(q')\leq(i,j-1) \wedge a'=a \wedge b'=b}{\ps{ab}\rank(q)=(i,j),j>1}.
\end{equation}
\end{subequations}
%
We say that a play $\pi$ over $\Gg$ is compliant with $\fo{}$ if there exist mode traces $\alpha\in[1;n]^\omega$ and $\beta\in[1;m]^\omega$ s.t.\ for all $k\in\mathbb{N}$ holds  $(\pi(2k+2),\alpha(2k+2),\beta(2k+2))=f^1(\pi(2k+1),\alpha(2k+1),\beta(2k+1))$, and 
\begin{inparaenum}[(i)]
 \item $\alpha(2k+1)=\alpha(2k)^+$ if $\ps{ab}\rank(\pi(2k+1))=(1,1)$, 
 \item $\alpha(2k+1)=\alpha(2k)$ if $\ps{ab}\rank(\pi(2k+1))=(i,1),i>1$, and
 \item $\alpha(2k+1)=\alpha(2k)$ and $\beta(2k+1)=\beta(2k)$ if $\ps{ab}\rank(\pi(2k+1))=(i,j),j>1$.
\end{inparaenum}

With this it is easy to see that the intuition behind \REFthm{thm:Soundness_single} directly carries over to every line of \eqref{equ:4FP_vector}. Additionally, using $\Preo(\Za{a^+})$ in $\ps{a}D$ allows to cycle through all the lines of \eqref{equ:4FP_vector}, which ensures that every set $\FGa{}\in\FcG$ is tried to be attained by the constructed system strategy in a pre-defined order. 
This is formalized in Appendix~\ref{sec:proof:vectorized} and summarized in \REFthm{thm:SoundnessCompleteness_vec} below.

To prove completeness, 
it is shown in Appendix~\ref{sec:proof:completeness:vec} 
that the negation of \eqref{equ:4FP_vector} can be 
over-approximated by negating every line separately. 
Therefore, the reasoning for every line of the negated fixed-point carries over from \REFsec{sec:AlgoS:Completeness}, resulting in the analogous completeness result.
With this we obtain soundness and completeness in direct analogy to \REFthm{thm:Soundness_single}-\ref{thm:Completeness_single}, formalized in \REFthm{thm:SoundnessCompleteness_vec}.

\begin{theorem}\label{thm:SoundnessCompleteness_vec}
 Let $(\Gg,\FcA,\FcG)$ be a GR(1) game with $\FcA=\Set{\FAa{1},\hdots,\FAa{m}}$ and $\FcG=\Set{\FGa{1},\hdots,\FGa{n}}$. Suppose $\fo{}$ is
 the system strategy in \eqref{equ:rank_new_a} based
 on the ranking in \eqref{equ:ranking_a}.
 Then it holds for all $q\in\semantics{\varphi_4^v}$ that \eqref{equ:Soundness} holds. Furthermore, it holds for all $q\notin \semantics{\varphi_4^v}$ and all system strategies $\fo{}$ over $\Gg$ that either \eqref{equ:LanguageSpecOld_q} or \eqref{eq:SCT_nonblocking_q} does not hold.
\end{theorem}

%


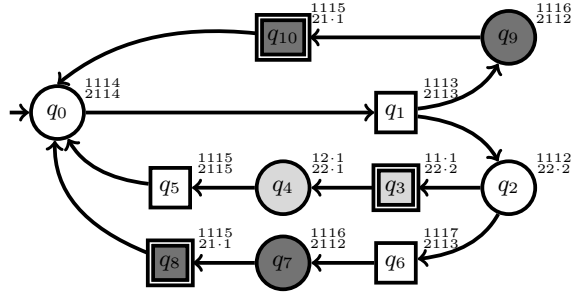
\begin{figure}[h!]
   \begin{center}
     \begin{tikzpicture}[auto,yscale=1,xscale=1.5]
     
          \node (init) at (-0.5,0) {};
          \node[mystate] (q0) at (0,0) {$q_0$};
          \node[draw,mysquare] (q1) at (3,0) {$q_1$};        
          \node[mystate] (q2) at (4,-1) {$q_2$};
          \node[draw,fmysquare,fill=black!15] (q3) at (3,-1) {$q_3$};
          \node[mystate,fill=black!15] (q4) at (2,-1) {$q_4$};
          \node[draw,mysquare] (q5) at (1,-1) {$q_5$};
          \node[draw,mysquare] (q6) at (3,-2) {$q_6$};
          \node[mystate,fill=black!55] (q7) at (2,-2) {$q_7$};
          \node[draw,fmysquare,fill=black!55] (q8) at (1,-2) {$q_8$};
          \node[mystate,fill=black!55] (q9) at (4,1) {$q_9$};
           \node[draw,fmysquare,fill=black!55] (q10) at (2,1) {$q_{10}$};
          
          \begin{tiny}
          \node (q0a) at (0.4,0.4) {$1114$};
          \node (q1a) at (3.4,0.4) {$1113$};        
          \node (q2a) at (4.4,-0.6) {$1112$};
          \node (q3a) at (3.4,-0.6) {$11\hspace{-1.5pt}\cdot\hspace{-1.5pt}1$};
          \node (q4a) at (2.4,-0.6) {$12\hspace{-1.5pt}\cdot\hspace{-1.5pt}1$};
          \node (q5a) at (1.4,-0.6) {$1115$};
          \node (q6a) at (3.4,-1.6) {$1117$};
          \node (q7a) at (2.4,-1.6) {$1116$};
          \node (q8a) at (1.4,-1.6) {$1115$};
          \node (q9a) at (4.4,1.4) {$1116$};
          \node (q10a) at (2.4,1.4) {$1115$};
          \end{tiny}
          
          \begin{tiny}
             \node (q0a) at (0.4,0.25) {$2114$};
             \node (q1a) at (3.4,0.25) {$2113$};        
          \node (q2a) at (4.4,-0.75) {$22\hspace{-1.5pt}\cdot\hspace{-1.5pt}2$};
          \node (q3a) at (3.4,-0.75) {$22\hspace{-1.5pt}\cdot\hspace{-1.5pt}2$};
          \node (q4a) at (2.4,-0.75) {$22\hspace{-1.5pt}\cdot\hspace{-1.5pt}1$};
          \node (q5a) at (1.4,-0.75) {$2115$};
          \node (q6a) at (3.4,-1.75) {$2113$};
          \node (q7a) at (2.4,-1.75) {$2112$};
          \node (q8a) at (1.4,-1.75) {$21\hspace{-1.5pt}\cdot\hspace{-1.5pt}1$};
          \node (q9a) at (4.4,1.25) {$2112$};
          \node (q10a) at (2.4,1.25) {$21\hspace{-1.5pt}\cdot\hspace{-1.5pt}1$};
          \end{tiny}

%

\SFSAutomatEdge{init}{}{q0}{semithick}{}  
\SFSAutomatEdge{q0}{}{q1}{}{} 
\SFSAutomatEdge{q1}{}{q2}{bend left}{}
\SFSAutomatEdge{q1}{}{q9}{bend right}{}
\SFSAutomatEdge{q2}{}{q3}{}{}
\SFSAutomatEdge{q2}{}{q6}{bend left}{}
\SFSAutomatEdge{q3}{}{q4}{}{} 
\SFSAutomatEdge{q4}{}{q5}{}{}
\SFSAutomatEdge{q5}{}{q0}{bend left}{}
\SFSAutomatEdge{q6}{}{q7}{}{}
\SFSAutomatEdge{q7}{}{q8}{}{}
\SFSAutomatEdge{q8}{}{q0}{bend left}{}
\SFSAutomatEdge{q9}{}{q10}{}{}
\SFSAutomatEdge{q10}{}{q0.north}{bend right}{} 

%

     \end{tikzpicture}
\end{center}
\caption{Game graph $H_2$. States in $\Qz$ and $\Qo$ are indicated by circles and squares, respectively. States in $\FcA$ and $\FcG$ are indicated by a state with a single and a double boundary, respectively, filled by light and dark gray for $a=1$ and $a=2$, respectively. The small numbers $aibj$ on the top right of each state indicate the ranking $(a,i,b,j)$ as in \REFfig{fig:strategy} induced by \eqref{equ:4FP_vector} via \eqref{equ:arank}.
}\label{fig:gamegraph2}
\end{figure}
\fullb
\smallbreak
\noindent\textbf{Example}

\noindent
We will explain the evaluation of the vectorized fixed-point in \eqref{equ:4FP_vector} using the game graph $H_2$ in \REFfig{fig:gamegraph2}. In this example, the fixed-point terminates after one iteration over every line of \eqref{equ:4FP_vector} with $\ps{1}Z=\ps{2}Z=Q$. Therefore, the ranking induced by the first iteration over $Z$ is also the final one. We discuss its construction for both lines separately.

$\mathbf{a=1}$:
First consider $a=1$ and $b=1$. In this case, the first iteration over $\ps{11}W$ starts with $\FGa{1}=\set{q_3}$ (giving $\ps{1}D=\set{q_3}$) and successively adds all states except $q_4$, as $q_4\in\FAa{1}$. Applying the smallest fixed-point over $\ps{11}X$ to this set has no effect and we have $\ps{11}X^1_\infty=\ps{11}W^1_\infty=Q\setminus\set{q_4}$ with $\ps{11}R^1_2=\set{q_2}$, $\ps{11}R^1_3=\set{q_1}$, $\ps{11}R^1_4=\set{q_0}$, $\ps{11}R^1_5=\set{q_{10},q_5,q_8}$, $\ps{11}R^1_6=\set{q_7,q_9}$ and $\ps{11}R^1_6=\set{q_6}$, as indicated by the upper four-digit number on the top-right of each state.
Now we consider $a=1$ and $b=2$. Again, the first iteration over $\ps{12}W$ starts with $\FGa{1}=\set{q_3}$ and successively adds all states except $q_7$ and $q_6$ (as $q_7\in\FAa{2}$ and $q_6$ is its predecessor), beginning with $q_2$. This results in $\ps{12}X=Q\setminus\set{q_6,q_7}$ for re-iterating $\ps{12}W$, which does not allow to add $q_2$ to $\ps{12}W$ as not all successors of $q_2$ (in particular $q_6$) are contained in $\ps{12}X\setminus\FAa{2}$. The re-calculation of the fixed-point therefore terminates with $\ps{12}X^1_\infty=\ps{12}W^1_\infty=\set{q_3}$, giving $\ps{12}R^1_j=\emptyset$ for all $j>1$. 
Now taking the union over the resulting fixed-points $\ps{11}X^1_\infty=Q\setminus\set{q_4}$ and $\ps{12}X^1_\infty=\set{q_3}$ gives $\ps{1}Y^1=Q\setminus\set{q_4}$ and it is easy to see that $q_4\in\Preo(\ps{1}Y^1)$, giving $\ps{1}E^2=\set{q_4}$. As all other states are already contained in $\ps{1}Y^1$ we have $\ps{12}R^2_j=\emptyset$ for all $j>1$. 

$\mathbf{a=2}$:
We first consider $a=2$ and $b=1$. In this case, the first iteration over $\ps{21}W$ starts with $\FGa{2}=\set{q_8,q_{10}}$ (giving $\ps{2}D=\set{q_8,q_{10}}$) and successively adds all states except $q_3$ and $q_4$ (as $q_4\in\FAa{1}$ and $q_3$ is its predecessor). Similarly to the case where $a=1$ and $b=2$ this results in the removal of $q_2$ from $\ps{21}W^1_\infty$ when re-iterating the fixed-point with $\ps{21}X=Q\setminus\set{q_3,q_4}$, as its successor $q_3$ is not contained in $\ps{21}X$. However, this does not effect the remaining iterations and we get $\ps{21}R^1_2=\set{q_7,q_9}$, $\ps{21}R^1_3=\set{q_6,q_1}$, $\ps{21}R^1_4=\set{q_0}$ and $\ps{21}R^1_5=\set{q_5}$, as indicated by the lower four-digit number on the top-right of each state.
Now we consider $a=2$ and $b=2$. Again, the first iteration over $\ps{22}W$ starts with $\FGa{2}=\set{q_8,q_{10}}$ but no further states are added as their only predecessors $q_7$ and $q_9$ are both in $\FAa{2}$. Hence, $\ps{22}R^1_j=\emptyset$ for all $j>1$. 
Now taking the union over the resulting fixed-points $\ps{21}X^1_\infty=\ps{21}W^1_\infty=Q\setminus\set{q_2,q_3,q_4}$ and $\ps{22}X^1_\infty=\ps{22}W^1_\infty=\set{q_8,q_{10}}$ gives $\ps{2}Y^1=Q\setminus\set{q_2,q_3,q_4}$ and it is easy to see that $q_4\in\Preo(\ps{2}Y^1)$, giving $\ps{2}E^2=\set{q_4}$. Now re-computing the fixed-points over $\ps{21}W$ and $\ps{22}W$ adds $q_2$ and $q_3$ in the first iteration in both cases. Hence $\ps{2\cdot}R^2_2=\set{q_2,q_3}$, as indicated by the lower four-digit number $22\cdot2$ on the top-right of both states.
%

Given this example we want to highlight that in $q_2$ the environment can decide to switch the $b$-mode from $2$ to $1$ by transitioning to $q_6$, which decreases the $i$-rank from $2$ to $1$. This is due to the fact that the re-evaluation of $\ps{22}W$ \enquote{copies} $\Preo{\ps{2}Y^1}$ to $\ps{22}W^2_1$, which contains $q_6$.

Further, we see that for $a=2$ the system strategy will always decide to move from $q_1$ to $q_9$, as this preserves the current $b$-mode. In this example, this also
allows to reach the target state $q_{10}\in\FGa{2}$ faster, which might not necessarily be the case. On the other hand, for $a=1$ the strategy will always transition from $q_1$ to $q_2$, as otherwise the rank increases. By this, the system must rely on the environment to eventually choose to transition from $q_2$ to $q_3$. While this might not always be the case (the environment is allowed to increase the $j$-rank by transitioning from $q_2$ to $q_6$), we see that whenever the environment plays such that the assumption is satisfied, i.e., $\FAa{1}=\set{q_4}$ is visited infinitely often, also $\FGa{1}=\set{q_3}$ is visited infinitely often, resulting in a winning play.  

\fulle

\smallbreak
\noindent
{\bf A Solution for \REFproblem{PS}}

\noindent
Given that $\Lomega(\Gg,\FcG)\subseteq\Lomega(\Gg,\FcA)$ it follows from \REFthm{thm:SoundnessCompleteness_vec} that \REFproblem{PS} has a solution iff $q_0\in \semantics{\varphi_4^v}$. 
Furthermore, if $q_0\in \semantics{\varphi_4^v}$ we can construct $\fo{}$
that wins the GR(1) condition in \eqref{equ:LanguageSpecOld} and is
non-conflicting, as in \eqref{eq:SCT_nonblocking}. 
%

Using a similar construction as in \REFsec{sec:AlgoS:PS}, we can
symbolically check whether
$\Lomega(\Gg,\FcG)\subseteq\Lomega(\Gg,\FcA)$. For this, we construct
a new game graph $\Gg_b$ for every $\FAa{b}$,~$b\in[1;m]$ by removing
the latter set from the state set of $\Gg$ and checking whether
$\Lomega(\Gg_b,\FcG)$ is empty. If some of these $m$ checks fail, we
have $\Lomega(\Gg,\FcG)\not\subseteq\Lomega(\Gg,\FcA)$. Now observe
that by checking every $\FAa{b}$ separately, we know which goals are
not necessarily passed by infinite runs which visit all $\FGa{a}$
infinitely often and can collect them in the set
$\FcA^{\mathrm{failed}}$.  
Using the same reasoning as in \REFsec{sec:AlgoS:PS}, 
we can simply add the set $\FcA^{\mathrm{failed}}$ to the system guarantee set to obtain an equivalent synthesis problem which is solvable by the given algorithm, if it is realizable. More precisely, consider the new system guarantee set $\FcG'=\FcG\cup\FcA^{\mathrm{failed}}$ and observe that $\Lomega(\Gg,\FcG')\subseteq\Lomega(\Gg,\FcA)$ by definition, and 
therefore substituting $\Lomega(\Gg,\FcG)$ by $\Lomega(\Gg,\FcG')$ in \eqref{equ:LanguageSpecOld} does not change the satisfaction of the given inclusion.

\section{Complexity Analysis}
We show that the search for a more elaborate strategy does not affect
the worst case complexity.
In \REFsec{section:experiments} we show that this is also the case in
practice.
We state this complexity formally below.

\begin{theorem}
  Let $(\Gg,\FcA,\FcG)$ be a GR(1) game.
  We can check whether there is a winning non-conflicting strategy
    $\fo{}$ by a symbolic algorithm that performs
    $O(|Q|^2|\FcG||\FcA|)$
  next step computations and by an enumerative algorithm that works in time
    $O(m|Q|^2|\FcG||\FcA|)$, where $m$ is the number of transitions
    of the game.
\end{theorem}

\begin{proof}
  Each line of the fixed-point is iterated $O(|Q|^2)$ times
  \cite{Browne96}.
  As there are $|\FcG||\FcA|$ lines the upper bound follows.
As we have to compute $|\FcG||\FcA|$ different ranks for
    each state, it follows that the complexity is $O(m|Q|^2|\FcG||\FcA|)$.
\end{proof}

We note that \emph{enumeratively} our approach is theoretically worse
than the classical approach to GR(1).
\newb
This follows from the straight forward reduction to the rank
computation in the rank lifting algorithm and the relative complexity
of the new rank when compared to the general GR(1) rank. 
We conjecture that more complex approaches, e.g., through a reduction
to a parity game and the usage of other enumerative algorithms, could
eliminate this gap. 
\newe

\section{Experiments}\label{section:experiments}


\begin{table}[b]
	\centering
	\vspace*{-2mm}
	\caption{Experimental results for the maze benchmark. The size
          of the maze is given in columns/lines, the number of goals
          is given per player. The states are counted for the returned
          winning strategies. Strategies preventing the environment
          from fulfilling its goals are indicated by a $^*$. Recorded
          computation times are rounded wall-clock
          times.}\label{table:maze_experiments}
        \vspace*{-2mm}
	\begin{footnotesize}
	\begin{tabular}{|c|c||c|c|c|c|c|c||c|c|c|c|c|c||} 
	        \hline
		&&\multicolumn{6}{c||}{falsifiable assumptions}&\multicolumn{6}{c||}{non-falsifiable assumptions}\\
		& &\multicolumn{2}{c|}{3FP} & \multicolumn{2}{c|}{4FP} & \multicolumn{2}{c||}{Heuristic}   &\multicolumn{2}{c|}{3FP} & \multicolumn{2}{c|}{4FP}& \multicolumn{2}{c||}{Heuristic}\\
		size & goals& states & time & states & time & states & time & states & time & states & time & states & time \\
		 \hline \hline
		 $3/2$&$2$ &	 $10^*$&$<1$s&			$46$&$<1$s&		$12$&$<1$s&		$35$&$<1$s&		$50$&$<1$s&		$40$&$<1$s	\\
 		\hline		 
		$3/10$&$10$&	$34^*$&$<1$s&			$1401$&$8$s&		$1307$&$3$s&		$1119$&$1$s&		$1513$&$13$s&		$1533 $&$5$s	\\
		$3/20$&$20$&	$64^*$&$21$s&			$5799$&$201$s&		$5732$&$337$s&		$3926$&$37$s&		$6000$&$163$s&		$6378 $&$105$s\\
		\hline
		$25/2$&$2$&	$94^*$&$<1$s&		$2144$&$4$s&		n.r.&$6$s&		$744$&$<1$s&	$2318$&$4$s			&n.r.&$5$s\\
		$63/2$&$2$&	$397^*$&$<1$s&		$14259$&$32$s&		n.r.&$101$s&		$4938$&$2$s&	$15465$&$54$s		&n.r.&$66$s\\
		\hline		
	\end{tabular}
	\end{footnotesize}
        \vspace*{-7mm}
\end{table}
%

We have implemented the 4-nested fixed-point algorithm in
\eqref{equ:4FP_vector} and the corresponding strategy extraction in
\eqref{equ:rank_new_a}. \new{It is available as an extension to the GR(1) synthesis tool
\texttt{slugs} \cite{slugs}}. In this section we show how this algorithm
(called 4FP) performs in comparison to the usual 3-nested fixed-point
algorithm for GR(1) synthesis (called 3FP) available in
\texttt{slugs}. All experiments were run on a computer with an Intel
i5 processor running an x86 Linux at 2\,GHz with 8\,GB of memory. 

We first run both algorithms on a benchmark set obtained from the maze example in the introduction by 
changing the number of rows and columns of the maze. 
We first increased the number of lines in the maze and added a goal state for both the obstacle and the robot per line. This results in a maze where in the first and last column, system and environment goals alternate and all adjacent cells are separated by a horizontal wall. 
Hence, both players need to cross the one-cell wide white space in the middle infinitely often to visit all their goal states infinitely often. The computation times and the number of states in the resulting strategy are shown in \REFtable{table:maze_experiments}, upper part, column 3-6. 
Interestingly, we see that the 3FP always returns a strategy that blocks the environment. 
In contrast, the non-conflicting strategies computed by the 4FP are relatively larger (in state size) and computed about 10 times 
slower compared to the 3FP (compare column 3-4 and 5-6). 
When increasing the number of columns instead (lower part of \REFtable{table:maze_experiments}), the number of goals is unaffected. We made the maze wider and left only a one-cell wide passage in the middle of the maze to allow crossings between its upper and lower row. 
Still, the 3FP only returns strategies that falsify the assumption, which have fewer states and are computed much faster than 
the environment respecting strategy returned by the 4FP.
Unfortunately, the speed of computing a strategy or its size is immaterial if the winning strategy so computed
wins only by falsifying assumptions.

To rule out the discrepancy between the two algorithms w.r.t.\ the size of strategies, we slightly modified the above maze benchmark s.t.\ the environment assumptions are not falsifiable anymore. We increased the capabilities of the obstacle by allowing it to move at most $2$ steps in each round and to \enquote{jump over} the robot. 
Under these assumptions we repeated the above experiments. The computation times and the number of states in the resulting strategy are shown in \REFtable{table:maze_experiments}, column 9-12. We see, that in this case the size of the strategies computed by the two algorithms are more similar. The larger number for the 4FP is due to the fact that we have to track both the $a$ and the $b$ mode, possibly resulting in multiple copies of the same $a$-mode state. We see that the state difference decreases with the number of goals (upper part of \REFtable{table:maze_experiments}, column 9-12) and increases with the number of (non-goal) states (lower part of \REFtable{table:maze_experiments}, column 9-12). In both cases, the 3FP still computes faster, but the difference decreases with the number of goals. 

In addition to the 3FP and the 4FP we have also tested a sound but
incomplete heuristic, which avoids the disjunction over all $b$'s in
every line of \eqref{equ:4FP_vector} by only investigating $a=b$. The
state count and computation times for this heuristic are shown in
\REFtable{table:maze_experiments}, column 7-8 for the original maze
benchmark, and in column 13-14 for the modified one. We see that in
both cases the heuristic only returns a winning strategy if the maze
is not wider then 3 cells. This is due to the fact that in all other
cases the robot cannot prevent the obstacle from attaining a
particular assumption state until the robot has moved from one goal to
the next. The 4FP handles this problem by changing between avoided
assumptions in between visits to different goals. 
Intuitively, the computation times and state counts for the heuristic
should be smaller then for the 4FP, as the exploration of the
disjunction over $b$'s is avoided, which is true for many scenarios of
the considered benchmark. It should however be noted that this is not
always the case (compare e.g. line 3, column 6 and 8). This stems from
the fact that restricting the synthesis to avoiding one particular
assumption might require more iterations over $W$ and $Y$ within the
fixed-point computation. 

\fullb
\smallskip
In addition to the maze benchmark, we have also run our algorithm on
the 3 safety-benchmarks that are included in the \texttt{slugs}
distribution. All three benchmarks do not have liveness assumptions
for either the system or the environment player. For all realizable
specifications, both the 3FP and the 4FP return the same strategy (as
there is only one maximal permissive strategy in a safety game) and
need almost the same time to compute this strategy.
\fulle

\section{Discussion}

We believe the requirement that a winning strategy be \emph{non-conflicting} is
a simple way to disallow strategies that win by actively preventing the environment
from satisfying its assumptions, without significantly changing the theoretical formulation
of reactive synthesis (e.g., by adding different winning conditions or new notions of equilibria).
It is not a trace property, but 
our main results show that adding this requirement retains the algorithmic niceties
of GR(1) synthesis: in particular, symbolic algorithms have the same asymptotic complexity.

However, non-conflictingness makes the implicit assumption of a ``maximally flexible'' environment:
it is possible that because of unmodeled aspects of the environment strategy, it is not possible
for the environment to satisfy its specifications in the precise way allowed by a non-conflicting strategy. 
In the maze example discussed in \REFsec{sec:Intro}, the environment needs to move the obstacle to precisely the 
goal cell which is currently rendered reachable by the system. 
If the underlying dynamics of the obstacle require it to go back to the lower left 
from state $q_3$ before proceeding to the upper right 
(e.g., due to a required battery recharge), 
the synthesized robot strategy prevents the obstacle from doing so.

Finally, if there is no non-conflicting winning strategy, one could look for a ``minimally violating'' strategy.
We leave this for future work.
Additionally, we leave for future work the consideration of
non-conflictingness for general LTL specifications or (efficient)
fragments thereof.

\clearpage

\bibliographystyle{abbrv}
\bibliography{bibliography_SCTvsRS}
\clearpage
\appendix

\section{Proofs for Singleton Winning Conditions}\label{sec:proof:single}

\subsection{Soundness}\label{sec:proof:single:soundness}



As mentioned, we compute
$W^i_j$ as part of $ Y^i$ and based on $ Y^{i-1}$ and
$ W^i_{j-1}$:
\begin{equation}\label{equ:proof:lastiteration}
W^i_j=(\FG\cap \Preo( Z^\infty))
\cup \Preo( Y^{i-1})
\cup \underbrace{(Q\setminus \FA\cap \OpPreTD{ W^i_{j-1},  Y^i \setminus \FA})}_{\Theta^i_j}
\end{equation}
Suppose that $\fo{}$ is
the system strategy in \eqref{equ:rank_new} based on the ranking in \eqref{equ:ranking}. 

%
We first show, that the property in \eqref{equ:rank_prop} holds.

\begin{lemma}\label{lem:rank_prop}
 Given the premises of \REFthm{thm:Soundness_single}, it holds that
 \begin{subequations}\label{equ:proof:rankFG}
\begin{align}
 &q\in (\FG\cap \Preo(Z^\infty))=:D~\label{equ:proof:rankFG:a}\\
 \Leftrightarrow~&\rank(q)=(1,1)~\label{equ:proof:rankFG:b}\\
 \Leftrightarrow~&q\in\FG\cap  Z^\infty,\label{equ:proof:rankFG:c}
\end{align}
\end{subequations}
 \begin{subequations}\label{equ:proof:rankFA}
\begin{align}
&q\in\Preo(Y^{i-1})\setminus Y^{i-1}=:E^i\neq\emptyset\label{equ:proof:rankFA:a}\\
\Leftrightarrow~&\propConj{\rank(q)=(i,1)}{i>1}~\label{equ:proof:rankFA:b}\\
\Leftrightarrow~&q\in(\FA\setminus\FG)\cap  Z^\infty.\label{equ:proof:rankFA:c}
\end{align}
\end{subequations}
 \begin{subequations}\label{equ:proof:rankNM}
\begin{align}
&q\in \Theta^i_j\setminus (W^i_{j-1}\cup Y^{i-1}\cup E^i\cup D)=:R^i_j\neq\emptyset\label{equ:proof:rankNM:a}\\
\Leftrightarrow~&\propConj{\rank(q)=(i,j)}{j>1}\label{equ:proof:rankNM:b}\\
\Leftrightarrow~&q\in(Z^\infty\setminus(\FA\cup\FG)).\label{equ:proof:rankNM:c}
\end{align}
\end{subequations}
\end{lemma}

\begin{proof}
We show all claims separately.

\textbf{Show \eqref{equ:proof:rankFG}:}
To see that \eqref{equ:proof:rankFG:a} $\Leftrightarrow$
\eqref{equ:proof:rankFG:c} holds, recall that $Z^\infty$ denotes the
fixed-point set.
We can show that $Z^\infty$ is closed under $\Preo(\cdot)$, which
immediately implies that
$(\FGa{}\cap \Preo(Z^\infty))=\FGa{}\cap Z^\infty$. 
Using \eqref{equ:proof:lastiteration} it can be easily observed
that for $i=j=1$ we have $ W^1_1=(\FG\cap \Preo(Z^\infty))=D$. 
As $ Y^0= W^0_0=\emptyset$ this implies that every state
$q\in D$ has $\rank(q)=(1,1)$ and vice versa.  
By the definition of the rank in \eqref{equ:ranking}, this in turn
means that $\rank(q')>(1,1)$ implies $q'\notin\FG\cap Z^\infty$, which proves \eqref{equ:proof:rankFG:a}$\Leftrightarrow$\eqref{equ:proof:rankFG:b}.

\textbf{Show \eqref{equ:proof:rankFA}:}
To see that \eqref{equ:proof:rankFA:b}$\Rightarrow$\eqref{equ:proof:rankFA:a}
holds, we pick $q$ s.t.\ $\rank(q)=(i,1)$ and $i>1$.
With $j=1$ we know that $ W^i_0=\emptyset$ and hence $\Theta^i_0=\emptyset$. It furthermore follows from \eqref{equ:proof:rankFG} and $i>1$ that $q\notin D$. As \eqref{equ:ranking} further implies  $q\in W^i_1$ we conclude from \eqref{equ:proof:lastiteration} that $q\in\Preo(Y^{i-1})$. It follows again from  \eqref{equ:ranking} that $q\notin Y^{i-1}$.
To see, that the other direction also holds, pick $q\in\Preo(Y^{i-1})\setminus Y^{i-1}=E^i$ and observe that $E^i\neq\emptyset$ iff $i>1$ as $Y^0=\emptyset$.
This implies $q\in W^i_j$ (from \eqref{equ:proof:lastiteration}) and hence $q\in\Ya{}^{i}$ by construction.
Now observe that \eqref{equ:ranking} determines the $j$-rank based on $W^i_j\setminus W^i_{j-1}$. As we know that $W^i_1$ contains $\Preo(Y^{i-1})$ from before, we conclude $j=1$.

We now show \eqref{equ:proof:rankFA:a}$\Rightarrow$\eqref{equ:proof:rankFA:c}.
By the nature of the fixed-point we have
$ Y^{i-1}= X^{i-1}=\bigcup_{j}^l  W^{i-1}_j= W^{i-1}_l$.
Hence, $q\in\Preo( W^{i-1}_l)$ and $q\notin W^{i-1}_l$ .  
As $ W^{i-1}_{l}$ is a fixed-point, we know that
$ W^{i-1}_{l}=\FG\cap \Preo( Z^\infty) \cup \Preo( Y^{i-2})\cup
  (Q \setminus \FA) \cap
\OpPreTD{ W^{i-1}_l, W^{i-1}_l\setminus \FA}$.
By definition we have\linebreak
$\OpPreTD{ W^{i-1}_{l}, W^{i-1}_{l}\setminus \FA}
=$
$\PreE( W^{i-1}_{l})
\cap \Preo( W^{i-1}_{l})=\Preo( W^{i-1}_{l})$, where the last equality follows from $\PreE( W^{i-1}_{l}) \supseteq
\Preo( W^{i-1}_{l})$. Hence, 
$ W^{i-1}_{l}=\FG\cap \Preo( Z^\infty) \cup \Preo( Y^{i-2})\cup
  \BR{(Q \setminus \FA) \cap\Preo( W^{i-1}_{l})}$.
It follows that every element in $\Preo( W^{i-1}_l)$ that is not in
$W^{i-1}_l$ must be in $\FA$. By recalling that $D\subseteq Y^{1}\subseteq Y^{i-1}$, we also have $q\notin\FG\cap Z^\infty$ from \eqref{equ:proof:rankFG}, what proves the statement.

To see that \eqref{equ:proof:rankFA:a}$\Leftarrow$\eqref{equ:proof:rankFA:c} also holds, fix $q\in(\FA\setminus\FG)\cap  Z^\infty$ s.t. $\rank(q)=(i,j)$. 
As $q\notin\FG$, it follows from \eqref{equ:proof:rankFG} that $i>1$ and $q\notin D$. With $q\in\FA\cap  Z^\infty$ we see that $q\notin\Theta^i_j$ either. With this, it follows from \eqref{equ:proof:lastiteration} that $q\in\Preo( Y^{i-1})$. As $Y^0$ there exists one $i$ for which $q\in\Preo( Y^{i-1})\setminus Y^{i-1}$.

\textbf{Show \eqref{equ:proof:rankNM}:}
First observe that for any $q$ s.t.\ $\rank(q)=(i,j)$ and $j>1$ we know that $q\in W^i_{j}\setminus W^i_{j-1}$ where $ W^i_{j-1}\neq\emptyset$ and $q\in Y^i\setminus Y^{i-1}$. 
As \eqref{equ:proof:rankFG} and \eqref{equ:proof:rankFA} holds, we furthermore know that $q\notin D$ and $q\notin E^i$. With this it follows from \eqref{equ:proof:lastiteration} that 
$q\in\Theta^i_j\setminus W^i_{j-1}\setminus Y^{i-1}\setminus E^i\setminus D$. This immediately proves \eqref{equ:proof:rankNM:b}$\Rightarrow$\eqref{equ:proof:rankNM:a}.
For the other direction, we see that $q\in\Theta^i_j$ implies $q\in W^i_{j}$ from \eqref{equ:proof:lastiteration}. As $q\notin W^i_{j-1}$ and $q\notin Y^{i-1}$, we know that $\rank(q)=(i,j)$. As $q\notin D$ and $q\notin E^i$, it immediately follows from \eqref{equ:proof:rankFG} and \eqref{equ:proof:rankFA} that $j>1$. 

To see that \eqref{equ:proof:rankNM:a}$\Rightarrow$\eqref{equ:proof:rankNM:c}, observe that \eqref{equ:proof:rankNM:a} and \eqref{equ:proof:lastiteration} imply that $q$ is contained in the last term of \eqref{equ:proof:lastinteration:a}, from which it is easy to see that $q\notin\FA$. Further, $q\notin\FG$ due to \eqref{equ:proof:rankFG}, what proves the statement.
\end{proof}

%
Now observe that the conditional predecessor in \eqref{equ:newPre} can be written as
\begin{align*}
\allowdisplaybreaks
\OpPreTD{P,\,P'}:= &\PreE(P) \cap \Preo(P\cup P')\\
 = &\SetCompX{\qz\in\Qz}{
 \begin{propConjA}
  \Trz(\qz)\cap P\neq\emptyset\\
   \Trz(\qz)\subseteq P\cup P'
 \end{propConjA}}
\cup\SetCompX{\qo\in\Qo}{\Tro(\qo)\cap P\neq\emptyset}.
\end{align*}
With this, \eqref{equ:proof:lastiteration} and \REFlem{lem:rank_prop} imply that for every system state $q\in\Qo\cap  Z^\infty$ one of the following three cases holds:
\begin{compactenum}[(a)]
\item $q\in\FG$ (i.e., $\rank(q)=(1,1)$ ) and there \emph{exists} $q'\in\Tro(q)\cap  Z^\infty$ with defined, arbitrary rank, or  
\item
  $q\in\FA\setminus\FG$, (i.e., $\rank(q)=(i,1)$, $i>1$ ) and there \emph{exists} $q'\in\Tro(q)\cap Z^\infty$ s.t.\ $\rank(q')\leq(i-1,\cdot)<\rank(q)$, or
\item
  $q\notin(\FA\cup\FG)$, (i.e., $\rank(q)=(i,j)$, $j>1$ ) and
  there \emph{exists} $q'\in\Tro(q)\cap Z^\infty$ s.t.\ $\rank(q')=(i,j')<\rank(q)$.
\end{compactenum}
Similarly, for every environment state $q\in\Qz\cap  Z^\infty$ holds
\begin{compactenum}[(a')]
 \item $q\in\FG$ (i.e., $\rank(q)=(1,1)$ ), 
   $\Trz(q)\subseteq  Z^\infty$, and all $q'\in\Trz(q)$ have a defined, arbitrary rank, or 
 \item $q\in\FA\setminus\FG$, (i.e., $\rank(q)=(i,1)$, $i>1$ ), $\Trz(q)\subseteq Z^\infty$ and
   $\rank(q')\leq(i-1,\cdot)<\rank(q)$ \emph{for all} $q'\in\Trz(q)$, or 
 \item $q\notin(\FA\cup\FG)$, (i.e., $\rank(q)=(i,j)$, $j>1$ ), $\Trz(q)\subseteq Z^\infty$, there
   \emph{exists} $q'\in\Trz(q)$ with $\rank(q')=(i,j')<\rank(q)$ and \emph{for all} $q'\in\Trz(q)$ holds
\new{
 \begin{compactenum}
 \item[(c'1)] $\rank(q')=(i,j')<\rank(q)$, or 
 \item[(c'2)] $\rank(q')=(i',\cdot)$ with $i'\leq i$ and $q'\notin\FA$.
 \end{compactenum}
 }
\end{compactenum}
It should be noted that the system strategy $\fo{}$ constructed in
\eqref{equ:rank_new} ensures that the transitions that are
existentially quantified in (a)-(c) are actually taken. Hence, case
(a) resets the rank, case (b) decreases the first component of
the rank and case (c) decreases the second component of the
rank. 


Based on this insight, we first show that any play over $\Gg$ started in a state $q\in
Z^\infty$ that complies with the system strategy $\fo{}$ and the
environment transition rules stays in $Z^\infty$.
\begin{lemma}\label{lem:stayinZ}
 Given the premises of \REFthm{thm:Soundness_single}, it holds for all $q\in Z^\infty$ that $\Tr{}(q)\in Z^\infty$.
\end{lemma}

\begin{proof}
Suppose $q\in \Qz\cap Z^\infty$. Then $\rank(q)$ is defined and one of the cases (a')-(c') holds. As for all cases holds $\Trz(q)\subseteq Z^\infty$, the claim follows.
  Suppose $q\in \Qo\cap Z^\infty$. Then $\rank(q)$ is defined and one of the cases (a)-(c) holds. If (a) holds, $q'=\fo{}(q)$ implies $q'\in  Z^\infty$ from the second line of \eqref{equ:rank_new}. If (b)-(c) holds $q'=\fo{}(q)$ implies $q'\in Z^\infty$ from the first line of \eqref{equ:rank_new}.  
\end{proof}

Next we show that every play $\pi$ on $\Gg$ consistent with $\fo{}$ and starting in $q\in Z^\infty$
satisfies the GR(1) winning condition.

\begin{lemma}\label{lem:GFa imp GFg}
Given the premises of \REFthm{thm:Soundness_single}, it holds
for all $q\in Z^\infty$ that $\Lomega_q(\Gg,\fo{}) \subseteq
\overline{\Lomega_q(\Gg,\FcA)} \cup \Lomega_q(\Gg,\FcG)$.
\end{lemma}

\begin{proof}
  Let $\pi\in\Lomega_q(\Gg,\fo{})$, i.e., $\pi(0)=q\in Z^\infty$. Then it follows from \REFlem{lem:stayinZ} that $\pi(k)\in Z^\infty$ for all $k\in\mathbb{N}$, i.e., one of the cases (a)-(c') holds for every $k$. 
  If $\pi$ visits every $\FA$ infinitely often, then case (b) or (b') occurs infinitely often.
  It follows that the first component decreases infinitely often.
  The only option that allows the first component to
  increase is by going through case (a) or (a').
  Hence, $\pi$ visits $\FG$ infinitely often.
\end{proof}

Next we show that there always exists a play $\pi$ on $\Gg$ that 
complies with $\fo{}$, starts in a state $q\in Z^\infty$ and visits every $\FG$ infinitely often. 

\begin{lemma}\label{lem:existE}
Given the premises of \REFthm{thm:Soundness_single}, it holds
for all $q\in Z^\infty$ that $\Lomega_q(\Gg,\fo{}) \cap \Lomega_q(\Gg,\FcG)\neq \emptyset$.
\end{lemma}

\begin{proof}
  We will construct an infinite computation $\pi$ in
  $\Lomega_q(\fo{})\cap \Lomega_q(\Gg,\FcG)$, hence $\pi(0)=q\in Z^\infty$.
  We construct $\pi$ by induction such that for every $k$ we have
  $\pi(k)\in Z^\infty$.
  As $\pi$ will be consistent with $\fo{}$ this follows from
  \REFlem{lem:stayinZ}.

  Let $\pi(k)=q'$, hence, by induction $\pi(k)\in Z^\infty$. Let $\rank(q')=(i,j)$, that is $q'\in  W^i_j$.
  Then one of the following cases holds:
  \begin{enumerate}
  \item
    $\rank(q')=(1,1)$ - then $q'\in \FG\cap \Preo( Z^\infty)$.
    We extend $\pi$ by choosing a successor $q''$ of $q'$ compatible
    with $\fo{}$ such that $q''\in  Z^\infty$.
  \item
    $\rank(q')=(i,1)$ for $i>1$ - then $q'\in\Preo( Y^{i-1})$.
    We extend $\pi$ by choosing a successor $q''$ of $q'$ compatible
    with $\fo{}$ such that $q''\in  Y^{i-1}$. That is, the first
    component in the rank of $q''$ is smaller than $i$.
  \item
    $\rank(q')=(i,j)$ for $j>1$ - then $q'\in (Q\setminus \FA)\cap
    \OpPreTD{ W^{i}_{j-1}, Y^{i}\setminus \FA}$.
    By definition of $\OpPreTD{}$ we have $q'\in \PreE( W^i_{j-1})$.
    We extend $\pi$ by choosing a successor $q''$ of $q'$ compatible
    with $\fo{}$ such that $q''\in  W^i_{j-1}$.
    That is, $\rank(q'')<\rank(q')$.

    We note that if $q'\in \Qo$ then the only option compatible with
    $\fo{}$ is $q''$.
    However, if $q'\in \Qz$ then $q''$ is compatible with $\fo{}$ but
    $q''$ is not enforceable by player~$1$.
  \end{enumerate}

  We show that $\pi\in\Lomega_q(\Gg,\FcG)$.
  In option $1$ above, $\FG$ is visited and the rank is possibly increased.
  In options $2$ and $3$ above, the rank of $\pi$ decreases.
  As $\pi$ is infinite, it follows that infinitely many times option
  $1$ is taken, implying that every $\FG$ is visited infinitely often, hence $\pi\in\Lomega_q(\Gg,\FcG)$. 
\end{proof}

As an immediate consequence of \REFlem{lem:stayinZ} and
\REFlem{lem:existE} we can now show that $\OpPre{\Lomega_q(\Gg,\fo{})}$ is contained
in $\OpPre{\Lomega_q(\Gg,\fo{})\cap \Lomega_q(\Gg,\FcA)}$.
Interestingly, this is only
true if $\Lomega(\Gg,\FcG)\subseteq\Lomega(\Gg,\FcA)$.


\begin{lemma}\label{lem:f_nonconflice}
 Given the premises of \REFthm{thm:Soundness_single}, let $q\in Z^\infty$ and
  $\Lomega_q(\Gg,\FcG)\subseteq\Lomega_q(\Gg,\FcA)$. Then, 
$
  \OpPre{\Lomega_q(\Gg,\fo{})} = \OpPre{\Lomega_q(\Gg,\fo{})\cap\Lomega_q(\Gg,\FcA)}\,.
$
\end{lemma}

\begin{proof}
  Observe that \enquote{$\supseteq$} above always holds.
  We therefore only prove the other direction.
  Pick $\pi \in \OpPre{\Lomega_q(\Gg,\fo{q})}$.
  Let $q'$ be the last state in $\pi$. 
  As $q\in Z^\infty$ it follows  from \REFlem{lem:stayinZ} that $q'\in
  Z^\infty$.
  Then we can use \REFlem{lem:existE} to pick $\beta$
  s.t.\ $\pi\beta\in \Lomega_q(\Gg,\fo{}) \cap \Lomega_q(\Gg,\FcG)$.
  As $\Lomega_q(\Gg,\FcG)\subseteq \Lomega_q(\Gg,\FcA)$ we therefore
  have $\pi\beta\in\Lomega_q(\Gg,\FcA)$ and hence $\pi\beta\in
  \Lomega_q(\Gg,\fo{})\cap \Lomega_q(\Gg,\FcA)$.
  With this we immediately have that $\pi\in
  \OpPre{\Lomega_q(\Gg,\fo{})\cap \Lomega_q(\Gg,\FcA)}$.
\end{proof}

\subsubsection{Proof of \REFthm{thm:Soundness_single}}
Combing the above properties of $\fo{}$ we see that
 \eqref{equ:Soundness:infinite} follows from \REFlem{lem:GFa imp GFg}, \eqref{equ:Soundness:nonempty} follows from \REFlem{lem:existE} and 
 \eqref{equ:Soundness:nonblock} follow from \REFlem{lem:f_nonconflice}.

 \subsection{Completeness}\label{sec:proof:completeness}
We start by negating \eqref{equ:new_FP}. We then use the induced ranking of this negated fixed-point to show that the environment can (i) render the negated winning set invariant, and (ii) can force the play to violate the guarantees. Based on this, we show that whenever \eqref{equ:LanguageSpecOld} holds for an arbitrary system strategy $\fo{}$ starting in $\semantics{\overline{\varphi}_4}$, then \eqref{eq:SCT_nonblocking} cannot hold. 

\subsubsection{Negating the fixed-point in \eqref{equ:new_FP}}
We use the negation rule of the $\mu$-calculus, i.e., $\neg (\mu
X~.~F(X))=\nu\overline{X}~.~\overline{F}(\overline{X})$, to negate
\eqref{equ:new_FP}.
This results in the fixed-point 
\begin{align}\label{equ:negatedFourFP_first}
 &\mu \Zt.\nu \Yt.\mu \Xt.\nu \Wt.~(\FGt\rs \cup\rs \Prez(\Zt)) \cap \Prez(\Yt) \cap (\FA\rs\cup\rs \OpPreTDc{\Wt, \Xt \rs\cup\rs \FA}).
\end{align}
By using de-Morgan laws on the right-hand side of
\eqref{equ:negatedFourFP_first} we obtain four disjuncts:
\begin{equation}\label{equ:rearrangeddemorgan}
\begin{array}{l l c c c c c r r}
 & ( & \FGt & \cap & \Prez(\Yt) &  \cap & \FA & ) & \qquad \langle L_1\rangle \\
 \cup & (& \Prez(\Zt) & \cap & \Prez(\Yt) & \cap & \FA &) & \qquad \langle L_2\rangle  \\
 \cup & (& \FGt  &\cap  &\Prez(\Yt)  &\cap&  \OpPreTDc{\Wt, \Xt \cup
   \FA} &)  & \qquad \langle L_3\rangle  \\
 \cup & (& \Prez(\Zt) &\cap & \Prez(\Yt)&  \cap&  \OpPreTDc{\Wt, \Xt
   \cup \FA}& )& \qquad \langle L_4\rangle \\
\end{array}
\end{equation}
From the structure of the fixed-points, we know that $\Zt\subseteq \Xt
\subseteq \Wt \subseteq\Yt$.
As $\Prez$ is monotonic, we have $\Prez(\Zt)\subseteq \Prez(\Yt)$.
It follows that $\langle L_2 \rangle$ above simplifies to $\Prez(\Zt)\cap \FA$ and
$\langle L_4 \rangle$ simplifies to $\Prez(\Zt)\cap \OpPreTDc{\Wt,\Xt\cup \FA}$. 
From de-Morgan rules, $\Wt\cap (\Xt\cup \FA)=\Xt\cap \Wt \cup \Wt\cap
\FA$.
From $\Xt \subseteq \Wt$ we have $\Xt\cap \Wt = \Xt$.
By definition  $\OpPreTDc{\Wt,\Xt\cup \FA}= \PreA(\Wt) \cup
\Prez(\Xt\cup (\Wt \cap\FA))$.
However, as $\Zt\subseteq \Xt$ we know that $\Prez(\Zt)\subseteq
\Prez(\Xt\cup(\Wt \cap \FA))$.
Hence, $\langle L_4 \rangle$ simplifies to $\Prez(\Zt)$, making $\langle L_2 \rangle$ redundant.
From $\Xt\subseteq \Wt\subseteq \Yt$ we know $\Prez(\Xt \cup (\Wt \cap
\FA)) \subseteq \Prez(\Yt)$.
From $\Wt\subseteq \Yt$ we know that $\PreA(\Wt)\subseteq \Prez(\Yt)$.
Thus, $\langle L_3 \rangle$ simplifies to $\FGt \cap \OpPreTDc{\Wt,\Xt\cup \FA}$.
Summarizing, we have
\begin{align*}
   ( \FGt \cap \Prez(\Yt) \cap \FA ) 
~\cup~  ( \FGt \cap \OpPreTDc{\Wt,\Xt\cup \FA}) 
~\cup~ ( \Prez(\Zt)),
\end{align*}
so \eqref{equ:negatedFourFP_first} simplifies to
\begin{equation}\label{equ:negatedFourFP}
  \begin{array}{l}
    \overline{\varphi}_4 = \mu \Zt ~.~\nu \Yt~.~\mu \Xt~.~\nu \Wt~.~ \hfill\\
    \multicolumn{1}{r}{ \quad (\Prez(\Zt) ~\cup~
      (\FGt \cap \FA \cap \Prez(\Yt)) ~\cup~  (\FGt\cap \OpPreTDc{\Wt, \Xt \cup \FA}))}.
  \end{array}
\end{equation}

\subsubsection{The induced ranking of $\Zt^\infty$}
Let $\Zt^0=\emptyset$ and 
$\Zt^i$ for $i\geq 1$
denote the set obtained in the $i$th iteration over $\Zt$.
For $i\geq 1$ we denote $\Yt{}^i=\Zt^i$ as the value of the fixpoint on
$\Yt{}$ that computes the $i$-th iteration of $\Zt$ .
Furthermore, let $\Xt{}^i_0=\emptyset$ and denote by $\Xt{}^i_j$ for
$j\geq 1$ the set obtained in the $j$-th iteration over $\Xt{}$
performed while computing $\Yt{}^i$ (i.e., using $\Yt{}^i$ for $\Yt{}$ and
$\Zt^{i-1}$ for $\Zt$).
Then
it follows from the properties of the fixed-point that after the $i$th 
iteration over $\Zt$ has terminated, we have $\Zt^i=\bigcup_j \Xt{}_j^i$
(in particular $\Zt^k=\bigcup_j \Xt{}_j^k$ for $\Zt^\infty=\Zt^k$).
We define the ranking for every state $q\in \Zt^\infty$ s.t.\
\begin{equation}
\rank(q)=(i,j) \iff q\in \Xt{}^i_j\setminus \Xt{}^i_{j-1} \mbox{ for $i,j>0$.}
\end{equation}
%
After termination of the inner fixed-point over $\Wt{}$,
giving $\Wt{}^i_j=\Xt{}^i_j$, we have
\begin{align}\label{equ:negatedFourFP_withranking}
\Xt{}^i_j=\Prez(\Zt^{i-1}) \rs\cup\rs (\FGt\rs\cap\rs \FA \rs\cap\rs \Prez(\Zt^{i})) \rs\cup\rs  (\FGt\rs\cap\rs \OpPreTDc{\Xt{}^i_j, \Xt{}^{i}_{j-1} \rs\cup\rs \FA})).
 \end{align}
Before interpreting this set, we look at the last term of \eqref{equ:negatedFourFP_withranking} separately.  Using the definition of $\PreA$, $\PreE$ and $\OpPreTDc{}$ from \REFsec{sec:prelim} we have
\begin{align}
 &\OpPreTDc{\Xt{}^i_j, \Xt{}^{i}_{j-1} \cup \FA}:= \PreA(\Xt{}^i_j) \cup
  \Prez(\Xt{}^i_j\cap (\Xt{}^{i}_{j-1} \cup \FA))\notag\\ 
 &=\PreA(\Xt{}^i_j) \cup \Prez(\Xt{}^{i}_{j-1}) \cup \Prez(\Xt{}^i_j\cap \FA)\notag\\ 
 &= \SetCompX{\qz\rs\in\rs\Qz}{\rs
 \begin{propDisjA}
  \Trz(\qz)\subseteq \Xt{}^i_j\\
 \Trz(\qz)\cap \Xt{}^{i}_{j-1}\neq\emptyset\\
 \Trz(\qz)\cap (\Xt{}^i_j\rs\cap\rs \FA)\neq\emptyset
 \end{propDisjA}\rs}
 \cup\SetCompSplit{\qo\rs\in\rs\Qo}{\Tro(\qo)\subseteq (\Xt{}^i_j\rs\cup\rs\Xt{}^{i}_{j-1}\rs\cup\rs(\Xt{}^i_j\rs\cap\rs \FA))}\notag\\
 &=\SetCompX{\qz\rs\in\rs\Qz}{\rs
 \begin{propDisjA}
  \Trz(\qz)\subseteq \Xt{}^i_j\\
 \Trz(\qz)\cap \Xt{}^{i}_{j-1}\neq\emptyset\\
 \Trz(\qz)\cap (\Xt{}^i_j\rs\cap\rs \FA)\neq\emptyset
 \end{propDisjA}\rs}
 \cup\SetCompX{\qo\rs\in\rs\Qo}{\Tro(\qo)\subseteq \Xt{}^i_j)}\label{equ:simplifyOpPreTDc}
\end{align}

Using \eqref{equ:negatedFourFP_withranking} and \eqref{equ:simplifyOpPreTDc} we see that for every
system state $q\in\Qo\cap \Zt^\infty$ with $\rank(q)=(i,j)$ holds
\begin{compactenum}[(a)]
\item $\Tro(q)\subseteq\Zt^\infty$ and  for all $q'\in\Tro(q)$
  holds $\rank(q')\leq(i-1,\cdot)$, or  
 \item $q\in\FA\setminus\FG$, $\Tro(q)\subseteq\Zt^\infty$ and for all $q'\in\Tro(q)$ holds $\rank(q')\leq(i,\cdot)$, or 
 \item $q\in\FGt$, $\Tro(q)\subseteq\Zt^\infty$ and for all $q'\in\Tro(q)$ holds $\rank(q')\leq(i,j)$.
\end{compactenum}
Similarly, for every environment state $q\in\Zt^\infty\cap\Qz$ with $\rank(q)=(i,j)$ holds
\begin{compactenum}[(a')]
 \item that there exists $q'\in\Trz(q)\cap\Zt^\infty$ s.t.\ $\rank(q')\leq(i-1,\cdot)$, or
 \item $q\in\FA\setminus\FG$, and there exists $q'\in\Trz(q)\cap\Zt^\infty$ s.t.\ $\rank(q')\leq(i,\cdot)$, or
 \item $q\in\FGt$ and either
 \begin{compactenum}
 \item[(c'1)] $\Trz(q)\subseteq\Zt^\infty$ and for all $q'\in\Trz(q)$ holds $\rank(q')\leq(i,j)$, or
 \item[(c'2)] there exists $q'\in\Trz(q)\cap\Zt^\infty$ s.t.\  $\rank(q')\leq(i,j-1)$, or
 \item[(c'3)] there exists $q'\in\Trz(q)\cap\Zt^\infty$ s.t.\ $\rank(q')\leq(i,j)$ and $q'\in\FA$.
 \end{compactenum}
\end{compactenum}

\subsubsection{Consequences for a game over $\Gg$}
Consider a system strategy $\fo{}$ over $\Gg$ starting in some state $q\in
\Zt^\infty$ and an environment playing in accordance with the properties $(a)-(c)$ and $(a')-(c')$.
We denote by $R_{\fo{}}$ the subset
of $\Zt^\infty$ that is reachable under $\fo{}$ within such a game and construct this region by induction on the distance from $q$ as follows.

By assumption $q\in\Zt^\infty$. Initially, we set $q\in R_{\fo{}}$.
Consider, by induction, a state $q'\in R_{\fo{}}$ with
$\rank(q')=(i,j)$. Then we have two cases.\\
\textbf{(1)} If $q'\in Q^1$, then based on the $(a)$, $(b)$, and $(c)$ above it
follows that either $(a)$ all successors of $q'$ have rank at most 
$(i-1,\cdot)$ and $\delta^1(q')\in \Zt^\infty$, $(b)$ all
successors of $q'$ have rank at most $(i,\cdot)$ and
$\delta^1(q')\in \Zt^\infty$,
or $(c)$ all successors of $q'$ have rank at
most $(i,j)$ and $\Tro(q')\in \Zt{}^\infty$.
In particular, one of these cases holds for the successor $q''$ of
$q'$ that is compatible with $\fo{}$. We add $q''$ to $R_{\fo{}}$.\\
\textbf{(2)} If $q'\in Q^0$, then based on $(a')$, $(b')$, and $(c')$ above it
follows that either $(a')$ there is a successor $q''$ of $q'$ that has
rank at most $(i-1,\cdot)$ and $q''\in\Zt^\infty$,
$(b')$ there is a successor $q''$ of $q'$ that has rank
at most $(i,\cdot)$ and $q''\in\Zt^\infty$, $(c'2)$ there is a
successor $q''$ of $q'$ such that $\rank(q'')\leq (i,j-1)$, or $(c'3)$
there is a successor $q''$ of $q'$ such that $\rank(q'')\leq (i,j)$
and $q''\in \FcA{}$.
In all these cases, we add this identified successor $q''$ to
$R_{\fo{}}$.
As $\fo{}$ is a strategy for the system, the state $q''$ is compatible with
$\fo{}$. 
The remaining case is $(c'1)$ when {\bf all} successors $q''$ of $q'$
satisfy that $q''\in\Zt^\infty$ and that $\rank(q'') \leq (i,j)$.
In that case we add {\bf all} successors $q''$ of $q'$ to
$R_{\fo{}}$.
As $\fo{}$ is a strategy for the system all these successors $q''$ are
compatible with $\fo{}$. 

We denote by $\Lomega_q(\Gg,\fo{},R_{\fo{}})$ the restriction of
$\Lomega_q(\Gg,\fo{})$ to computations that remain within $R_{\fo{}}$.
It is easy to see that the following lemma follows by construction and is therefore stated without proof.
\begin{lemma}\label{lem:propRf1}
  Given the premises of \REFthm{thm:Completeness_single}, it holds that 
  $\Lomega_q(\Gg,\fo{},R_{\fo{}})\neq \emptyset$ and $\Lomega_q(\Gg,\fo{},R_{\fo{}})\subseteq\Zt^\infty$. 
\end{lemma}
Hence, the environment can render $\Zt^\infty$ invariant. Additionally, it can ensure that $\FcG$ is only visited finitely often, as formalized in the following lemma.

\begin{lemma}\label{lem:FGng}
Given the premises of \REFthm{thm:Completeness_single}, it holds 
for all $q\in \Zt{}^\infty$ and for every system strategy $\fo{}$ over $\Gg$ that
$\Lomega_q(\Gg,\fo{},R_{\fo{}}) \cap \Lomega(\Gg,\FcG)=\emptyset$.
\end{lemma}

\begin{proof}
Let $\pi\in\Lomega_q(\Gg,\fo{},R_{\fo{}})$. In particular, $\pi(0)=q\in
\Zt^\infty$.
As $R_{\fo{}}\subseteq \Zt^\infty$, for all $k\in\mathbb{N}$ one of
the cases (a)-(c') holds.
As $\FG$ can only be visited by going through cases (a) and (a'),
every visit of $\pi$ to $\FG$ causes
the first component of the rank to decrease.
As no case causes an increase of the first component of the rank,
$\pi$ ultimately gets trapped in states with some $i$-rank and cannot
visit $\FG$ any more.
Hence, $\FG$ is not visited infinitely often and therefore
$\pi\notin\Lomega(\Gg,\FcG)$.
\end{proof}

Using \REFlem{lem:propRf1} and \REFlem{lem:FGng} we can now show the essence of \REFthm{thm:Completeness_single}, i.e., that  whenever \eqref{equ:LanguageSpecOld_q} holds for an arbitrary system strategy $\fo{}$ starting in $\semantics{\overline{\varphi}_4}$, then \eqref{eq:SCT_nonblocking_q} cannot hold. This is formalized in the following lemma.

\begin{lemma}\label{lem:FAconditioned}
 Given the premises of \REFthm{thm:Completeness_single}, let $q\in \semantics{\overline{\varphi}_4^v}$ and $\fo{}$ be a system
strategy over $\Gg$ s.t.\ \eqref{equ:LanguageSpecOld_q} holds. 
Then $\OpPre{\Lomega_q(\Gg,\fo{})} \neq
\OpPre{\Lomega_q(\Gg,\fo{})\cap\Lomega_q(\Gg,\FcA)}$.
\end{lemma}

\begin{proof}
First observe that the left part of \eqref{equ:LanguageSpecOld_q}  implies $\OpPre{\Lomega_q(\Gg,\fo{})} \neq\emptyset$. 
The claim is therefore proven by showing that $\OpPre{\Lomega_q(\Gg,\fo{})\cap\Lomega_q(\Gg,\FcA)}=\emptyset$.

  Consider the unwinding of the region $R_{\fo{}}$ to an infinite
  tree $T$.
  Label every node in the tree according to the case $(a)-(c)$ or
  $(a')-(c'3)$ that applies to it according to the construction of
  $R_{\fo{}}$.

  By Lemma~\ref{lem:FGng} there are finitely many occurrences of cases
  $(a)$ and $(a')$.
  Assume by contradiction that cases $(b)$, $(b')$ or $(c'3)$ appear
  infinitely often in $T$.
  From K\"onig's lemma it follows that there is a path $\pi$ in $T$
  along which these cases occur infinitely often.
  However, whenever $(b)$, $(b')$, or $(c'3)$ occur, $\pi$ visits
  $\FcA$.
  It follows that $\pi$ visits infinitely many states in $\FcA$ and
  only finitely many states in $\FcG$ (from Lemma~\ref{lem:FGng}).
  This contradicts the assumption that $\fo{}$ satisfies
  \eqref{equ:LanguageSpecOld_q}. 
  It follows that cases $(b)$, $(b')$ and $(c'3)$ occur finitely often
  in $T$.
  
  Now consider a location in $T$ under which there are no appearances of
  cases $(b)$, $(b')$ or $(c'3)$ and restrict attention to the
  sub-tree $T'$ of $T$ under this location.
  Suppose that case $(c'2)$ occurs infinitely often in $T'$.
  As $(c'2)$ leads to a decrease in the second component of the rank,
  and cases $(c)$ and $(c'1)$ do not allow the rank to increase it
  follows that there are finitely many occurrences of $(c'2)$ in $T'$.
  
  This reasoning implies that along every branch of $T$ (enumerated by $k\in\mathbb{N}$) 
  there exists a \emph{finite} prefix $s_k\in\OpPre{\Lomega_q(\Gg,\fo{})}$ leading to a 
  state $q_k$ at which a sub-tree $T''_k$ is rooted in which only cases $(c)$ and $(c'1)$ occur.
  By construction of $R_{\fo{}}$ all sub-trees $T''_k$ are closed under
  environment moves. This implies that $\OpPre{\Lomega_q(\Gg,\fo{})}=\bigcup_k \OpPre{s_k\cdot\Lomega_{q_k}(\Gg,\fo{},T''_k)}$.
  Further, using the same reasoning as before we know that \eqref{equ:LanguageSpecOld_q} implies that $T''_k$ only contains finitely many states in $\FcA$.
  This implies that $\Lomega_{q_k}(\Gg,\fo{},T''_k)\cap\Lomega_q(\Gg,\FcA)=\emptyset$ for all $k$.  
  As $s_k$ also only contains finitely many states in $\FcA$ (from above), combining the last two observations results in $\OpPre{\Lomega_q(\Gg,\fo{})\cap\Lomega_q(\Gg,\FcA)}=\emptyset$.
%
%
\end{proof}

\subsubsection{Proof of \REFthm{thm:Completeness_single}}
It is easy to see that \REFthm{thm:Completeness_single} directly follows from \REFlem{lem:FAconditioned}. If we pick some system strategy $\fo{}$ over $H$ we have that either \eqref{equ:LanguageSpecOld_q} does not hold, or, if  \eqref{equ:LanguageSpecOld_q} holds we know from \REFlem{lem:FAconditioned} that \eqref{eq:SCT_nonblocking_q} does not hold.

\section{Proofs for GR(1) Winning Conditions}\label{sec:proof:vectorized}

\subsection{Soundness}\label{sec:proof:soundness:vec}
We start by recalling that the last iteration of the fixed-point in
\eqref{equ:4FP_vector} results in the sets $\Za{a}^\infty\subseteq Q$
and define 
$\semantics{\varphi_4^v}=Z^\infty=\bigcup_{a\in[1;n]}\Za{}^\infty$.
Now let $\Ya{}^i$ be the set obtained after the $i$-th iteration of
$\Ya{}$ in line $a$ of \eqref{equ:4FP_vector}, let $\Xa{ab}^i$
denote the fixed-point of the iteration over $\Xa{ab}$ resulting in
$\Ya{ab}^i$ and denote by $\Wa{ab}^i_j$ the set obtained in the $j$th
iteration  
over $\Wa{ab}$ performed while computing $\Xa{ab}^i$ in line $a$ of \eqref{equ:4FP_vector}.
With this notation, we see that the computation of $\Wa{ab}^i_j$ as
part of $\Ya{a}^i$ and based on $\Ya{a}^{i-1}$, $\Xa{ab}^i$, and 
$\Wa{ab}^i_{j-1}$ results in the set
\begin{equation}\label{equ:proof:lastinteration:a}
 \Wa{ab}^{i}_j=
(\FGa{}\cap \Preo(\Za{a^+}^\infty))
\cup 
\Preo(\Ya{}^{i-1})
\cup 
\underbrace{(Q\setminus \FAa{b}\cap \OpPreTD{\Wa{ab}^i_{j-1}, \Xa{ab}^i \setminus \FAa{b}})}_{\ps{ab}\Theta^i_j}.
\end{equation}
Using \eqref{equ:proof:lastinteration:a}, we first show that \eqref{equ:arank} holds.

\begin{lemma}
Given the premises of \REFthm{thm:SoundnessCompleteness_vec} it holds that 
\begin{subequations}\label{equ:proof:arankFG}
\begin{align}
&q\in (\FGa{}\cap \Preo(\Za{a^+}^\infty))=:\ps{a}D\label{equ:proof:arankFG:a}\\
\Leftrightarrow~&\AllQ{b\in[1;m]}{\ps{ab}\rank(q)=(1,1)}\label{equ:proof:arankFG:b}
\end{align}
\end{subequations}
\begin{subequations}\label{equ:proof:arankFA}
\begin{align}
 &q\in\Preo(\Ya{}^{i-1})\setminus\Ya{}^{i-1}=:\ps{a}E^i\label{equ:proof:arankFA:a}\\
 \Leftrightarrow~&\AllQ{b\in[1;m]}{\propConj{\ps{ab}\rank(q)=(i,1)}{i>1}}\label{equ:proof:arankFA:b}
\end{align}
\end{subequations}
\begin{subequations}\label{equ:proof:arankNM}
\begin{align}
 &q\in \ps{ab}\Theta^i_j\setminus (\Wa{ab}^i_{j-1}\cup \Ya{}^{i-1}\cup \ps{a}E^i\cup D)=:\ps{ab}R^i_j\neq\emptyset\label{equ:proof:arankNM:a}\\
 \Leftrightarrow ~ &\propConj{\ps{ab}\rank(q)=(i,j)}{j>1}\label{equ:proof:arankNM:b}\\
\Rightarrow~&q\notin\FAa{b}.\label{equ:proof:arankNM:c}
\end{align}
\end{subequations}
\end{lemma}

\begin{proof}
We show each claim separately.

\textbf{Show \eqref{equ:proof:arankFG}:}
Using \eqref{equ:proof:lastinteration:a} it can be easily observed
that for $i=j=1$ we have $\Wa{ab}^1_1=\FGa{}\cap \Preo(\Za{a^+}^\infty)=\ps{a}D$.  
As $\Ya{}^0=\emptyset$ and $\Wa{ab}^1_0=\emptyset$ this implies that every state
$q\in\ps{a}D$ has $\ps{ab}\rank(q)=(1,1)$ for every $b$ (from \eqref{equ:ranking_a}).
By the definition of the $\ps{ab}$rank in \eqref{equ:ranking_a}, this in turn
means that $\ps{ab}\rank(q')>(1,1)$ implies $q'\notin \ps{a}D$. Hence, \eqref{equ:proof:arankFG:a} $\Leftrightarrow$ \eqref{equ:proof:arankFG:b} holds.

\textbf{Show \eqref{equ:proof:arankFA}:}
First observe that for every $q$ s.t.\ $\ps{ab}\rank(q)=(i,1)$ and $i>1$ we know that $\Wa{ab}^i_{j-1}=\Wa{ab}^i_{0}=\emptyset$ and with this $\ps{ab}\Theta^i_j=\emptyset$.
As $q\in \Wa{ab}^{i}_1$ and $\Wa{ab}^i_0=\emptyset$ we conclude $q\in\Preo(\Ya{}^{i-1})$. Now observe from the definition of the ranking that we have $q\notin \Ya{}^{i-1}$. This immediately proves
\eqref{equ:proof:arankFA:b}$\Rightarrow$\eqref{equ:proof:arankFA:a}. 
To see, that the other direction also holds, fix
$q\in\Preo(\Ya{}^{i-1})\setminus\Ya{}^{i-1}$.
If $q\in\Preo(\Ya{}^{i-1})\neq\emptyset$ then $q\in \Wa{ab}^i_j$ for all $b$ (from \eqref{equ:proof:lastinteration:a}) and hence $q\in\Ya{}^{i}$ by construction. Now observe that \eqref{equ:ranking_a} determines the $j$-rank based on $\Wa{ab}^i_j\setminus \Wa{ab}^i_{j-1}$. As we know that $\Wa{ab}^i_1$ contains $\Preo(Y^{i-1})$ (from \eqref{equ:proof:lastinteration:a}), we conclude $j=1$.

\textbf{Show \eqref{equ:proof:arankNM}:}
First observe that for every $q$ s.t.\ $\ps{ab}\rank(q)=(i,j)$ and $j>1$ we know that $q\in\Wa{ab}^i_{j}\setminus\Wa{ab}^i_{j-1}$ where $\Wa{ab}^i_{j-1}\neq\emptyset$ and $q\in\Ya{a}^i\setminus\Ya{a}^{i-1}$. 
As \eqref{equ:proof:arankFG} and \eqref{equ:proof:arankFA} hold, we furthermore know that $q\notin\ps{a}D$ and $q\notin\ps{a}E^i$. With this it follows from \eqref{equ:proof:lastinteration:a} that 
$q\in\ps{ab}\Theta^i_j\setminus (\Wa{ab}^i_{j-1}\cup \Ya{}^{i-1}\cup \ps{a}E^i\cup D)$. This immediately proves \eqref{equ:proof:arankNM:b}$\Rightarrow$\eqref{equ:proof:arankNM:a}.
For the other direction, we see that $q\in\ps{ab}\Theta^i_j$ implies $q\in\Wa{ab}^i_{j}$ from \eqref{equ:proof:lastinteration:a}. As $q\notin \Wa{ab}^i_{j-1}$ and $q\notin\Ya{a}^{i-1}$, we know that $\ps{ab}\rank(q)=(i,j)$. As $q\notin\ps{a}D$ and $q\notin\ps{a}E^i$, it immediately follows from \eqref{equ:proof:arankFG} and \eqref{equ:proof:arankFA} that $j>1$. 

To see that
\eqref{equ:proof:arankNM:a}$\Rightarrow$\eqref{equ:proof:arankNM:c},
observe that \eqref{equ:proof:arankNM:a} and
\eqref{equ:proof:lastinteration:a} imply that $q$ is contained in the
last term of \eqref{equ:proof:lastinteration:a}, from which it is easy
to see that $q\notin\FAa{b}$.%
%
\end{proof}

Even though the proven statements are a bit weaker compared to \REFlem{lem:rank_prop} they are still sufficient to derive the same cases for states within $Z^\infty$ as in case of singleton winning conditions.
\new{In particular, observe that \eqref{equ:proof:arankNM:c} implies that any state in $\FAa{b}\cap \Za{a}^\infty$  needs to have a $\ps{ab}\rank(q)$ with $j=1$.}
Therefore, the remaining proof for soundness follows the same lines as the one discussed in \REFsec{sec:proof:single:soundness} by annotating the used sets with $a$ and $b$ modes. The resulting lemmas and proofs are given in the remainder of this section for the sake of completeness.

We start by observing the different cases for states in $\Za{}^\infty$.
For every system state $q\in\Qo\cap \Za{}^\infty$ one of the following three cases hold:
\begin{compactenum}[(a)]
\item $\AllQ{b\in[1;m]}{\ps{ab}\rank(q)=(1,1)}$, i.e., $q\in\FGa{}$ and there \emph{exists} $q'\in\Tro(q)\cap \Za{a^+}^\infty$ with defined, arbitrary $a^+b'$-rank for some $b'\in[1;m]$, or  
\item $\AllQ{b\in[1;m]}{\ps{ab}\rank(q)=(i,1),i>1}$, i.e.,
  there
  \emph{exists} $q'\in\Tro(q)\cap \Za{a}^\infty$ and some
  $b'\in[1;m]$
  s.t.\  $\ps{ab'}\rank(q')\leq(i-1,\cdot)<\ps{ab'}\rank(q)$, or 
\item $\ExQ{b\in[1;m]}{\ps{ab}\rank(q)=(i,j),j>1}$,i.e.,
  $q\notin \FAa{b}$  and 
  there \emph{exists} $q'\in\Tro(q)\cap \Za{a}^\infty$
  s.t.\ $\ps{ab}\rank(q')\leq(i,j-1)<\ps{ab}\rank(q)$.
\end{compactenum}
Similarly, for every environment state $q\in\Qz\cap \Za{}^\infty$ holds
\begin{compactenum}[(a')]
 \item $\AllQ{b\in[1;m]}{\ps{ab}\rank(q)=(1,1)}$, i.e., $q\in\FGa{}$, further
   $\Trz(q)\subseteq \Za{a^+}^\infty$, and for all $q'\in\Trz(q)$
   exists $b'\in[1;m]$ s.t.\ $q$ has a defined, arbitrary
   $\ps{ab'}$rank, or  
 \item $\AllQ{b\in[1;m]}{\ps{ab}\rank(q)=(i,1),i>1}$, i.e.,
   $\Trz(q)\subseteq \Za{a}^\infty$ and 
   \emph{for all} $q'\in\Trz(q)$ exists some $b'\in[1;m]$
   s.t.\ $\ps{ab'}\rank(q')\leq(i-1,\cdot)<\ps{ab}\rank(q)$, or  
 \item $\ExQ{b\in[1;m]}{\ps{ab}\rank(q)=(i,j),j>1}$, i.e.,
   $q\notin \FAa{b}$, further $\Trz(q)\subseteq
   \Za{a}^\infty$, there 
   \emph{exists} $q'\in\Trz(q)$ with $\ps{ab}\rank(q')\leq(i,j-1)<\ps{ab}\rank(q)$ and \emph{for all} $q'\in\Trz(q)$ holds
\new{
 \begin{compactenum}
 \item[(c'1)] $\ps{ab}\rank(q')=(i,j')<\ps{ab}\rank(q)$,
 \item[(c'2)] $\ps{ab}\rank(q')=(i,\cdot)$ and $q'\notin\FAa{b}$, or
 \item[(c'3)] there exists $b'\in[1;m]$ s.t. $\ps{ab'}\rank(q')=(i',\cdot)$ with $i'< i$ and $q'\notin\FAa{b}$.
 \end{compactenum}
 }
\end{compactenum}
It should be noted that the system strategy $\fo{}$ constructed in
\eqref{equ:rank_new_a} ensures that the transitions that are
existentially quantified in (a)-(c) are actually taken. Hence, case
(a) resets the $a^+$-rank (ignoring $b$), case (b) decreases the first component of
the $a$-rank (ignoring $b$) and case (c) decreases the second component of the
$ab$-rank. 

Based on this insight, we first show that every play over $\Gg$ started in a state $q\in
Z^\infty$ that complies with the system strategy $\fo{}$ and the
environment transition rules stays in $Z^\infty$.
\begin{lemma}\label{lem:stayinZ_vec}
Given the premises of \REFthm{thm:SoundnessCompleteness_vec}, it holds for all $q\in Z^\infty$ that $\Trz(q)\in Z^\infty$ if $q\in\Qz$ and $(q',a',b')=\fo{}(q,a,b)$ implies $q'\in Z^\infty$, otherwise.
\end{lemma}

\begin{proof}
Suppose $q\in \Qz\cap Z^\infty$. Then $\ps{ab}\rank(q)$ is defined and one of the cases (a')-(c') holds. As for all cases holds $\Trz(q)\subseteq \Za{}^\infty\subseteq Z^\infty$, the claim follows.
  Suppose $q\in \Qo\cap Z^\infty$. Then $\ps{ab}\rank(q)$ is defined and one of the cases (a)-(c) holds. 
  If (a) holds, $(q',a^+,b')=\fo{}(q,a,b)$ implies $q'\in \Za{a^+}^\infty\subseteq Z^\infty$ from the first line of \eqref{equ:rank_new_a}. 
  If (b)-(c) holds $(q',a,b')=\fo{}(q,a,b)$ implies $q'\in \Za{a}^\infty\subseteq Z^\infty$ from the second and third line of \eqref{equ:rank_new_a}.
\end{proof}

Next we show that every play $\pi$ on $\Gg$ consistent with $\fo{}$ and starting in $q\in Z^\infty$
satisfies the GR(1) winning condition.

\begin{lemma}\label{lem:GFa imp GFg_vec}
Given the premises of \REFthm{thm:SoundnessCompleteness_vec}, it holds for all $q\in Z^\infty$ that $\Lomega_q(\Gg,\fo{}) \subseteq
\overline{\Lomega_q(\Gg,\FcA)} \cup \Lomega_q(\Gg,\FcG)$.
\end{lemma}

\begin{proof}
  Let $\pi\in\Lomega_q(\Gg,\fo{})$, i.e., $\pi(0)=q\in Z^\infty$. Then it follows from \REFlem{lem:stayinZ_vec} that $\pi(k)\in Z^\infty$ for all $k\in\mathbb{N}$, i.e., one of the cases (a)-(c') holds for every $k$. 
  
  Now assume that $\pi\in\Lomega_q(\Gg,\FcA)$, i.e., $\pi$ visits every $\FAa{b}$ with $b\in[1,m]$ infinitely often. It remains to show that in this case $\pi$ needs to also pass $\FGa{}$ with $a\in[1,n]$ infinitely often. 
  
  Consider some state $q=\pi(k)$ s.t.\ (c) or (c') holds, i.e., there
  exists $a$, $b$ s.t.\  $\ps{ab}\rank(q)=(i,j)$ with $j>1$ and
  $q\notin \FAa{b}$. In order to visit $\FAa{b}$ again, the second
  component of the $\ps{ab}\rank$ has to decrease to $j=1$, entering
  case (a') or (b') \new{(as $q\notin \FAa{b}$ whenever case (c') holds for $q$)}. If we enter case (a'), $\FGa{a}$ is visited and
  the rank gets reset. Then we can re-apply the same reasoning for
  $a^+$ and some $b'$. On the other hand, if we enter case (b'), the
  first component of the $a$-rank gets reduced
  and $b$ possibly changes to some
  $b''$. Re-applying the same reasoning as before shows that case (b')
  always eventually needs to occur in $\pi$, always reducing the first
  component of the rank for every $b$. The only option that allows the
  first component of the rank to increase is by going through case (a)
  or (a'). As $\pi$ is infinite, while the ranking is finite, this
  implies that we eventually need to go through case (a) or (a') for
  $a$, passing $\FGa{a}$. With this, we reach a state $\pi(k')$
  s.t. $\ps{a^+b}\rank(\pi(k'))$ is defined for some $b$. Then we can
  apply the same reasoning to show that we will eventually pass
  $\FGa{a^+}$.
  
  Hence, $\pi$ visits $\FGa{}$ for every $a\in[1,n]$ infinitely often.
\end{proof}

Next we show that there always exists a play $\pi$ on $\Gg$ that 
complies with $\fo{}$, starts in a state $q\in Z^\infty$ and visits every $\FGa{}$ infinitely often. 

\begin{lemma}\label{lem:existE_vec}
Given the premises of \REFthm{thm:SoundnessCompleteness_vec}, it holds for all
$q\in Z^\infty$ that $\Lomega_q(\Gg,\fo{}) \cap \Lomega_q(\Gg,\FcG)\neq \emptyset$.
\end{lemma}

\begin{proof}
  We will construct an infinite computation $\pi$ in
  $\Lomega_q(\fo{})\cap \Lomega_q(\Gg,\FcG)$.
  We construct $\pi$ by induction such that for every $k$ we have
  $\pi(k)\in Z^\infty$.
  As $\pi$ will be consistent with $\fo{}$ this follows from
  \REFlem{lem:stayinZ_vec}.
  Let $\pi(0)=q\in Z^\infty$.

  Let $\pi(k)=q'$, hence, by induction $\pi(k)\in Z^\infty$. Let $\ps{a}\rank(q')=(i,j)$, that is $q'\in \Wa{}^i_j$.
  Then one of the following cases holds:
  \begin{enumerate}
  \item
    $\ps{ab}\rank(q')=(1,1)$ for all $b$ - then $q'\in \FGa{}\cap \Preo(\Za{a^+}^\infty)$.
    We extend $\pi$ by choosing a successor $q''$ of $q'$ compatible
    with $\fo{}$ such that $q''\in \Za{a^+}^\infty$.
  \item
    $\ps{ab}\rank(q')=(i,1)$ with $i>1$  for all $b$ - then $q'\in\Preo(\Ya{}^{i-1})$.
    We extend $\pi$ by choosing a successor $q''$ of $q'$ compatible
    with $\fo{}$ such that $q''\in \Ya{}^{i-1}$. That is, the first
    component in the rank of $q''$ is smaller than $i$.
  \item
    There exists some $b$ s.t. $\ps{ab}\rank(q')=(i,j)$ with $j>1$ - then $q'\in (Q\setminus \FAa{b})\cap
    \OpPreTD{\Wa{ab}^{i}_{j-1},\Xa{ab}^{i}\setminus \FAa{b}}$.
    By definition of $\OpPreTD{}$ we have $q'\in \PreE(\Wa{ab}^i_{j-1})$.
    We extend $\pi$ by choosing a successor $q''$ of $q'$ compatible
    with $\fo{}$ such that $q''\in \Wa{ab}^i_{j-1}$.
    That is, $\ps{ab}\rank(q'')\leq(i,j-1)<\ps{ab}\rank(q')$.

    We note that if $q'\in \Qo$ then the only option compatible with
    $\fo{}$ is $q''$.
    However, if $q'\in \Qz$ then $q''$ is compatible with $\fo{}$ but
    $q''$ is not enforceable by player~$1$.
  \end{enumerate}

  We show that $\pi\in\Lomega_q(\Gg,\FcG)$.
  In option $1$ above, $\FGa{}$ is visited, the mode is changed to $a^+b'$ and both components of the rank are possibly increased.
  In options $2$ and $3$ above, the rank of $\pi$ decreases.
  As $\pi$ is infinite, it follows that infinitely many times option
  $1$ needs to be taken, implying that every mode $a\in[1;n]$ and every $\FGa{}$ is visited infinitely often, hence $\pi\in\Lomega_q(\Gg,\FcG)$. 
\end{proof}

As an immediate consequence of \REFlem{lem:stayinZ_vec} and
\REFlem{lem:existE_vec} we can now show that $\OpPre{\Lomega_q(\Gg,\fo{})}$ is contained
in $\OpPre{\Lomega_q(\Gg,\fo{})\cap \Lomega_q(\Gg,\FcA)}$.
Interestingly, this is only
true if $\Lomega(\Gg,\FcG)\subseteq\Lomega(\Gg,\FcA)$.


\begin{lemma}\label{lem:f_nonconflice_vec}
Given the premises of \REFthm{thm:SoundnessCompleteness_vec}, let $q\in Z^\infty$ and
  $\Lomega_q(\Gg,\FcG)\subseteq\Lomega_q(\Gg,\FcA)$. Then, 
$
  \OpPre{\Lomega_q(\Gg,\fo{})} = \OpPre{\Lomega_q(\Gg,\fo{})\cap\Lomega_q(\Gg,\FcA)}\,.
$
\end{lemma}

\begin{proof}
  Observe that \enquote{$\supseteq$} above always holds.
  We therefore only prove the other direction.
  Pick $\pi \in \OpPre{\Lomega_q(\Gg,\fo{q})}$.
  Let $q'$ be the last state in $\pi$. 
  As $q\in Z^\infty$ it follows  from \REFlem{lem:stayinZ_vec} that $q'\in
  Z^\infty$.
  Then we can use \REFlem{lem:existE_vec} to pick $\beta$
  s.t.\ $\pi\beta\in \Lomega_q(\Gg,\fo{}) \cap \Lomega_q(\Gg,\FcG)$.
  As $\Lomega_q(\Gg,\FcG)\subseteq \Lomega_q(\Gg,\FcA)$ we therefore
  have $\pi\beta\in\Lomega_q(\Gg,\FcA)$ and hence $\pi\beta\in
  \Lomega_q(\Gg,\fo{})\cap \Lomega_q(\Gg,\FcA)$.
  With this we immediately have that $\pi\in
  \OpPre{\Lomega_q(\Gg,\fo{})\cap \Lomega_q(\Gg,\FcA)}$.
\end{proof}

\subsubsection{Proof of \REFthm{thm:SoundnessCompleteness_vec}, part 1}
Combing the above properties of $\fo{}$ we see that
 \eqref{equ:Soundness:infinite} follows from \REFlem{lem:GFa imp GFg_vec}, \eqref{equ:Soundness:nonempty} follows from \REFlem{lem:existE_vec} and 
 \eqref{equ:Soundness:nonblock} follow from \REFlem{lem:f_nonconflice_vec}.

 \subsection{Completeness}\label{sec:proof:completeness:vec}

We first show that the negation of \eqref{equ:4FP_vector} can be over-approximated by negating every line separately. This implies that the reasoning for every line of the negated fixed-point carries over from \REFsec{sec:proof:completeness} by annotating the used sets with $a$ and $b$ modes. The resulting lemmas and proofs are re-stated in this section for the sake of completeness.

\subsubsection{Negating the vectorized fixed-point in \eqref{equ:4FP_vector}}
First observe that negating line $a$ of \eqref{equ:4FP_vector} results
in the formula
\begin{align}\label{equ:negatedFourFP_first_vec}
  \nu\Yta{}~.~
  \bigwedge_{b=1}^m\mu \Xta{ab}~.~\nu \Wta{ab}~.~&(\FGta{}\cup
  \Prez(\Zta{a^+})) \cap \Prez(\Yta{})\\ 
 &\cap (\FAa{b}\cup \OpPreTDc{\Wta{ab}, \Xta{ab} \cup \FAa{b}}).\notag
\end{align}

One assumption that was made in the simplification of
\eqref{equ:negatedFourFP} was that $\Zt\subseteq\Xt\subseteq
\Wt\subseteq \Yt$.
When we consider the vectorized version, the right hand side of
$\Zta{}$ depends on $^{a^+}\Zt$.
Although, ultimately, all the $Z$ variables have the same value (as
arises from our proofs) we cannot rely on this in the simplification
of the fixpoint.
Instead, we use an over-approximation of the fixpoint.
Consider the reorganization of \eqref{equ:negatedFourFP_first}
appearing in \eqref{equ:rearrangeddemorgan}.
The reasoning that simplifies $\langle L_3\rangle$ relies on $\Xt \subseteq \Wt
\subseteq \Yt$. It is easy to see, that we still have $\Xta{ab}
\subseteq \Wta{ab} \subseteq \Yta{a}$, but the simplification of
$\langle L_2\rangle$ and $\langle L_4\rangle$ to $\Prez(\Zt)$ relies 
on $\Zt \subseteq \Yt$.
However, we note that in both cases, $\Prez(\Zt)$ over-approximates
$\langle L_2\rangle$ and $\langle L_4\rangle$.
It follows that if we replace $\langle L_2\rangle$ and $\langle L_4\rangle$ in
\eqref{equ:rearrangeddemorgan} by $\Prez(\Zt)$ we get a formula that
characterizes \emph{more} states.
Applying this reasoning to \eqref{equ:negatedFourFP_first_vec} results in 
\begin{align}\label{equ:negatedFourFP_proof_vec}
  \nu\Yta{}~.~
  \bigwedge_{b=1}^m\mu \Xta{ab}~.~\nu \Wta{ab}~.~&
 (\Prez(\Zta{a^+}))\\
 &\cup (\FGta{} \cap \FAa{b} \cap \Prez(\Yta{}))\notag\\
 &\cup  (\FGta{}\cap \OpPreTDc{\Wta{ab}, \Xta{ab} \cup \FAa{b}}))\notag
\end{align}
which is the mode-annotated version of \eqref{equ:negatedFourFP}.
We denote the vectorized versions of \eqref{equ:negatedFourFP_first_vec} and \eqref{equ:negatedFourFP_proof_vec} by $\overline{\varphi}_4^v$ and
$\overline{\varphi}_5^v$.
That is, include the vector least-fixpoint on the $\Zta{}$ variables, where
each line is either \eqref{equ:negatedFourFP_first_vec} ($\overline{\varphi}_4$)
or \eqref{equ:negatedFourFP_proof_vec} ($\overline{\varphi}_5$). 

We know that $\semantics{\overline{\varphi}_4^v}\subseteq
\semantics{\overline{\varphi}_5^v}$ (point wise containment for the
resulting vector of $\Zta{}$).
We have defined $\semantics{\varphi_4}=\bigcup_{a\in
  [1;n]}\Za{}^\infty$.
It follows, that in order to prove that $\varphi_4^v$ is complete it would
be sufficient to prove that in $\overline{\varphi}_4^v$ the environment
wins the GR(1) game from every state in $\bigcap_{a\in [1;n]}
\Zta{}^\infty$.
However, we are going to show that the environment wins the GR(1) game
from every state in $\bigcup_{a\in [1;n]} \Zta{}^\infty$ as
computed by $\overline{\varphi}_5^v$.
From the soundness argument, as established above, and from the
determinacy of GR(1) games, it follows that
$\semantics{\overline{\varphi}_5^v} \cap
\semantics{\varphi_4}=\emptyset$.
It follows that
$\semantics{\overline{\varphi}_5^v}=\semantics{\overline{\varphi}_4^v}$
and furthermore for every $a,a'\in [1;n]$ we have
$\Za{}^\infty=\ps{a'}Z^\infty$ (in $\varphi_4$) and
$\Zta{}^\infty=\ps{a'}\Zt^\infty$ (in $\overline{\varphi}_4^v$ and
$\overline{\varphi}_5^v$).
We now proceed with the analysis of $\overline{\varphi}_5^v$ by defining $\Zt^\infty:=\semantics{\overline{\varphi}_5^v}$.

\subsubsection{The induced ranking of $\Zt^\infty$}
Let $\Zt^\infty=\bigcup_{a\in [1;n]}\Zta{}^\infty$.
Let $\Zta{}^0=\emptyset$ and 
$\Zta{}^i$ for $i\geq 1$
denote the set obtained in the $i$th iteration over $\Zta{}$.
Notice that the $\Zta{}^i$ is coordinated for all modes $a$.
That is, they are all obtained from the same vector that is co-computed.
For $i\geq 1$ we denote $\Yta{}^i=\Zta{}^i$ as the value of the fixpoint on
$\Yta{}$ that computes the $i$-th approximation of $\Zta{}$ (based on
$^{a^+}\Zt^{i-1}$).
Furthermore, let $\Xta{ab}^i_0=\emptyset$ and denote by $\Xta{ab}^i_j$ for
$j\geq 1$ the set obtained in the $j$-th iteration over $\Xta{ab}$
performed while computing $\Yta{}^i$ (i.e., using $\Yta{}^i$ for $\Yta{}$ and
$^{a^+}\Zt^{i-1}$ for $^{a^+}\Zt$).
Then
it follows from the properties of the fixed-point that after the $i$th 
iteration over $\Zta{}$ has terminated, we have $\Zta{}^i=\bigcup_j \bigcap_b \Xta{ab}_j^i$
(in particular $\Zta{}^k=\bigcup_j \bigcap_b \Xta{ab}_j^k$ for $\Zta{}^\infty=\Zta{}^k$).

We define the $ab$-ranking for every state $q\in \Zta{}^\infty$ s.t.\
\begin{equation}
\ps{ab}\rank(q)=(i,j) \iff q\in \Xta{ab}^i_j\setminus \Xta{ab}^i_{j-1} \mbox{ for $i,j>0$.}
\end{equation}
%
After termination of the inner fixed-point over $\Wta{ab}$,
giving $\Wta{ab}^i_j=\Xta{ab}^i_j$, we have
\begin{align}\label{equ:negatedFourFP_withranking_a}
\Xta{ab}^i_j=&\Prez(\ps{a^+}\Zt^{i-1}) \cup (\FGta{} \cap \FAa{b} \cap \Prez(\Zta{}^{i}))\notag\\
&\cup  (\FGta{}\cap \OpPreTDc{\Xta{ab}^i_j, \Xta{ab}^{i}_{j-1} \cup \FAa{b}})).
 \end{align}
Before interpreting this set, we look at the last term of \eqref{equ:negatedFourFP_withranking_a} separately.  Using the definition of $\PreA$, $\PreE$ and $\OpPreTDc{}$ from \REFsec{sec:prelim} we have
\begin{align*}
 &\OpPreTDc{\Xta{ab}^i_j, \Xta{ab}^{i}_{j-1} \cup \FAa{}}\\
 :=& \PreA(\Xta{ab}^i_j) \cup
  \Prez(\Xta{ab}^i_j\cap (\Xta{ab}^{i}_{j-1} \cup \FAa{}))\\
 =&\PreA(\Xta{ab}^i_j) \cup \Prez((\Xta{ab}^{i}_{j} \cap \Xta{ab}^{i}_{j-1}) \cup (\Xta{ab}^i_j\cap \FAa{}))\\ 
 =&\PreA(\Xta{ab}^i_j) \cup \Prez(\Xta{ab}^{i}_{j-1} \cup (\Xta{ab}^i_j\cap \FAa{}))\\ 
 =& \SetCompX{\qz\in\Q^0}{
 \begin{propDisjA}
  \Trz(\qz)\subseteq \Xta{ab}^i_j\\
 \Trz(\qz)\cap \Xta{ab}^{i}_{j-1}\neq\emptyset\\
 \Trz(\qz)\cap (\Xta{ab}^i_j\cap \FAa{})\neq\emptyset
 \end{propDisjA}}\\
 &\cup\SetCompX{\qo\in\Q^1}{\Tro(\qo)\subseteq \Xta{ab}^i_j \vee
   \Tro(\qo)\subseteq \Xta{ab}^{i}_{j-1} \cup (\Xta{ab}^i_j\cap \FAa{})}\\
 =&\SetCompX{\qz\in\Q^0}{
 \begin{propDisjA}
  \Trz(\qz)\subseteq \Xta{ab}^i_j\\
 \Trz(\qz)\cap \Xta{ab}^{i}_{j-1}\neq\emptyset\\
 \Trz(\qz)\cap (\Xta{ab}^i_j\cap \FAa{})\neq\emptyset
 \end{propDisjA}}
 \cup\SetCompX{\qo\in\Q^1}{\Tro(\qo)\subseteq \Xta{ab}^i_j)}
\end{align*}

Using \eqref{equ:negatedFourFP_withranking_a} and the previous derivation we see that for every
system state $q\in\Q^1\cap \Zta{}^\infty$ with $\ps{ab}\rank(q)=(i,j)$ holds 
\begin{compactenum}[(a)]
\item $\Tre{1}(q)\subseteq\Zta{}^\infty$ and for all $q'\in\Tre{1}(q)$ exists a $b'\in[1;m]$ s.t.\
   $\ps{a^+b'}\rank(q')\leq(i-1,\cdot)$, or  
 \item 
 $q\in\FAa{}\setminus\FGa{}$, $\Tre{1}(q)\subseteq\Zta{}^\infty$ and for all $q'\in\Tre{1}(q)$ exists a $b'\in[1;m]$ s.t. $\ps{ab'}\rank(q')\leq(i,\cdot)$, or 
 \item $q\in\FGta{}$, $\Tre{1}(q)\subseteq\Zta{}^\infty$ and for all $q'\in\Tre{1}(q)$ holds $\ps{ab}\rank(q')\leq(i,j)$.
\end{compactenum}
Similarly, for every environment state $q\in\Zta{}^\infty\cap\Q^{0}$ with $\ps{ab}\rank(q)=(i,j)$ holds
\begin{compactenum}[(a')]
 \item that there exists $q'\in\Tre{0}(q)\cap\Zta{a^+}^\infty$ s.t.\ $\ps{a^+b}\rank(q')\leq(i-1,\cdot)$ for some $b\in[1;m]$, or
 \item $q\in\FAa{}\setminus\FGa{}$, and there exists $q'\in\Tre{0}(q)\cap\Zta{}^\infty$ s.t.\ $\ps{ab'}\rank(q')\leq(i,\cdot)$ for some (possibly different) $b'\in[1;m]$, or
 \item $q\in\FGta{}$ and either
 \begin{compactenum}
 \item[(c'1)] $\Tre{0}(q)\subseteq\Zta{}^\infty$ and for all $q'\in\Tre{0}(q)$ holds $\ps{ab}\rank(q')\leq(i,j)$, or
 \item[(c'2)] there exists $q'\in\Tre{0}(q)\cap\Zta{}^\infty$ s.t.\  $\ps{ab}\rank(q')\leq(i,j-1)$, or
 \item[(c'3)] there exists $q'\in\Tre{0}(q)\cap\Zta{}^\infty$ s.t.\ $\ps{ab}\rank(q')\leq(i,j)$ and $q'\in\FAa{}$.
 \end{compactenum}
\end{compactenum}

\subsubsection{Consequences for a GR(1) game over $\Gg$}
Consider a system strategy $\fo{}$ over $\Gg$ starting in some state $q\in
\Zt{}^\infty$.
We use the properties $(a)-(c)$ and $(a')-(c')$ to identify a subset $R_{\fo{}}$
of $\Zt{}^\infty$ that is reachable under $\fo{}$.
We construct this region by induction on the distance from $q$.

By assumption $q\in\Zt{}^\infty$. Initially, we set $q\in R_{\fo{}}$.
Consider, by induction, a state $q'\in R_{\fo{}}$ with
$\ps{ab}\rank(q')=(i,j)$. Then we have two cases.\\
\textbf{(1)} If $q'\in \Q^{1}$, then based on the $(a)$, $(b)$, and $(c)$ above it
follows that either $(a)$ for all successors $q''$ of $q'$ exists some $b'$ s.t.\ $q''$ has \ps{a^+b'}rank at most 
$(i-1,\cdot)$ and $\Tre{1}(q')\in \Zta{a^+}^\infty$, $(b)$ for all
successors $q''$ of $q'$ exists some $b'$ s.t.\ $q''$ has \ps{ab'}rank at most $(i,\cdot)$ and
$\Tre{1}(q')\in \Zta{}^\infty$,
or $(c)$ all successors of $q'$ have \ps{ab}rank at
most $(i,j)$ and $\Tro(q')\in \Zta{}{}^\infty$.
In particular, one of these cases holds for the successor $q''$ of
$q'$ that is compatible with $\fo{}$. We add $q''$ to $R_{\fo{}}$.\\
\textbf{(2)} If $q'\in \Q^{0}$, then based on $(a')$, $(b')$, and $(c')$ above it
follows that either $(a')$ there is a successor $q''$ of $q'$ that has
\ps{a^+b'}rank at most $(i-1,\cdot)$ for some $b'\in[1;m]$ and $q''\in\Zta{a^+}^\infty$,
$(b')$ there is a successor $q''$ of $q'$ that has \ps{ab'}rank
at most $(i,\cdot)$ for some $b'\in[1;m]$ and $q''\in\Zta{}^\infty$, $(c'2)$ there is a
successor $q''$ of $q'$ such that $\ps{ab}\rank(q'')\leq (i,j-1)$, or $(c'3)$
there is a successor $q''$ of $q'$ such that $\ps{ab}\rank(q'')\leq (i,j)$ and $q''\in \FAa{}$.
In all these cases, we add this identified successor $q''$ to
$R_{\fo{}}$.
As $\fo{}$ is a strategy for the system, the state $q''$ is compatible with
$\fo{}$. 
The remaining case is $(c'1)$ when {\bf all} successors $q''$ of $q'$
satisfy that $q''\in\Zta{}^\infty$ and that $\ps{ab}\rank(q'') \leq (i,j)$.
In that case we add {\bf all} successors $q''$ of $q'$ to
$R_{\fo{}}$.
As $\fo{}$ is a strategy for the system all these successors $q''$ are
compatible with $\fo{}$. 

We denote by $\Lomega_q(\Gg,\fo{},R_{\fo{}})$ the restriction of
$\Lomega_q(\Gg,\fo{})$ to computations that remain within $R_{\fo{}}$.
%
It is easy to see that the following lemma follows by construction and is therefore stated without proof.
\begin{lemma}\label{lem:propRf1_vec}
  Given the premises of \REFthm{thm:SoundnessCompleteness_vec}, it holds that 
  $\Lomega_q(\Gg,\fo{},R_{\fo{}})\neq \emptyset$ and $\Lomega_q(\Gg,\fo{},R_{\fo{}})\subseteq\Zt^\infty$. 
\end{lemma}
Hence, the environment can render $\Zt^\infty$ invariant. Additionally, it can ensure that $\FcG$ is only visited finitely often, as formalized in the following lemma.

\begin{lemma}\label{lem:FGng_vec}
Given the premises of \REFthm{thm:SoundnessCompleteness_vec}, it holds
for all $q\in \Zt{}^\infty$ and for every system strategy $\fo{}$ over $\Gg$ that
$\Lomega_q(\Gg,\fo{},R_{\fo{}}) \cap \Lomega(\Gg,\FcG)=\emptyset$.

\end{lemma}

\begin{proof}
Let $\pi\in\Lomega_q(\Gg,\fo{},R_{\fo{}})$. In particular, $\pi(0)=q\in
\Zt{}^\infty$.
As $R_{\fo{}}\subseteq \Zt{}^\infty$, for all $k\in\mathbb{N}$ one of
the cases (a)-(c') holds.
As $\FGa{}$ can only be visited by going through cases (a) and (a'),
every visit of $\pi$ to $\FGa{}$ causes the mode to change to $a^+$ and
the first component of the $a^+$ rank decreases.
As no case causes an increase of the first component of the rank,
$\pi$ ultimately gets trapped in a single mode $a$ and cannot visit all
$\FGa{}$ infinitely often, i.e., $\pi\notin\Lomega(\Gg,\FcG)$.
\end{proof}


\begin{lemma}\label{lem:FAconditioned_vec}
Let $q\in \Zt^\infty$ and $\fo{}$ be a system
strategy over $\Gg$ s.t.\ \eqref{equ:LanguageSpecOld_q} holds. 
Then $\OpPre{\Lomega_q(\Gg,\fo{})} \neq
\OpPre{\Lomega_q(\Gg,\fo{})\cap\Lomega_q(\Gg,\FAa{})}$.
\end{lemma}

\begin{proof}
First observe that the left part of \eqref{equ:LanguageSpecOld_q}  implies $\OpPre{\Lomega_q(\Gg,\fo{})} \neq\emptyset$. 
The claim is therefore proven by showing that $\OpPre{\Lomega_q(\Gg,\fo{})\cap\Lomega_q(\Gg,\FAa{})}=\emptyset$.

  Consider the unwinding of the region $R_{\fo{}}$ to an infinite
  tree $T$.
  Label every node in the tree according to the case $(a)-(c)$ or
  $(a')-(c'3)$ that applies to it according to the construction of
  $R_{\fo{}}$.

  By Lemma~\ref{lem:FGng} there are finitely many occurrences of cases
  $(a)$ and $(a')$ for every $a$.
  Assume by contradiction that cases $\ps{b}(b)$, $\ps{b}(b')$ or $\ps{b}(c'3)$ appear
  infinitely often in $T$ for every $b$.
  From K\"onig's lemma it follows that there is a path $\pi$ in $T$
  along which these cases occur infinitely often.
  However, whenever $\ps{b}(b)$, $\ps{b}(b')$, or $\ps{b}(c'3)$ occur, $\pi$ visits
  $\FAa{}$.
  As $\fo{}$ satisfies \eqref{equ:LanguageSpecOld_q} and $\pi$ only visits finitely many states in $\FGa{}$ for all $a$, we know that for every path $\pi$ in $T$ there exists at least one $b\in[1,m]$ s.t.\ $\ps{b}(b)$, $\ps{b}(b')$, or $\ps{b}(c'3)$ only occur finitely often along $\pi$. This forms $m$ subtrees $\ps{b}T$ where inside $\ps{b}T$ there are only finitely many occurrences of $(a)$, $(a')$, $\ps{b}(b)$, $\ps{b}(b')$ and $\ps{b}(c'3)$ and we have $T=\bigcup_b \ps{b}T$.

  Now consider a location in $\ps{b}T$ under which there are no appearances of
  cases $\ps{b}(b)$, $\ps{b}(b')$ or $\ps{b}(c'3)$ and restrict attention to the
  sub-tree $\ps{b}T'$ of $\ps{b}T$ under this location.
  Suppose that case $(c'2)$ (for some $b'\in[1;m]$) occurs infinitely often in $\ps{b}T'$ .
  As $(c'2)$ leads to a decrease in the second component of the rank,
  and cases $(c)$ and $(c'1)$ do not allow the rank to increase it
  follows that there are finitely many occurrences of $(c'2)$ in $\ps{b}T'$.
  
  This reasoning implies that along every branch of $\ps{b}T$ (enumerated by $\ps{b}k\in\mathbb{N}$) 
  there exists a \emph{finite} prefix $\ps{b}s_{\ps{b}k}\in\OpPre{\Lomega_q(\Gg,\fo{})}$ leading to a 
  state $\ps{b}q_{\ps{b}k}$ at which a sub-tree $\ps{b}T''_{\ps{b}k}$ is rooted in which only cases $(c)$ and $(c'1)$ occur.
  By construction of $R_{\fo{}}$ all sub-trees $\ps{b}T''_{\ps{b}k}$ are closed under
  environment moves. This implies that $\OpPre{\Lomega_q(\Gg,\fo{})}=\bigcup_b\bigcup_{\ps{b}k} \OpPre{\ps{b}s_{\ps{b}k}\cdot\Lomega_{\ps{b}q_{\ps{b}k}}(\Gg,\fo{},\ps{b}T''_{\ps{b}k})}$.
  Further, using the same reasoning as before we know that \eqref{equ:LanguageSpecOld_q} implies that $\ps{b}T''_{\ps{b}k}$ only contains finitely many states in $\FAa{}$. 
  This implies that for all $b\in[1;m]$ we have that  $\Lomega_{\ps{b}q_{\ps{b}k}}(\Gg,\fo{},\ps{b}T''_{\ps{b}k})\cap\Lomega_q(\Gg,\FAa{})=\emptyset$ holds for all ${\ps{b}k}$.  
  Combining the last two observations we have $\OpPre{\Lomega_q(\Gg,\fo{})\cap\Lomega_q(\Gg,\FAa{})}=\emptyset$.
%
%
\end{proof}

\subsubsection{Proof of \REFthm{thm:SoundnessCompleteness_vec}, part 2}
It is easy to see that the claim directly follows from \REFlem{lem:FAconditioned_vec}. If we pick some system strategy $\fo{}$ over $H$ we have that either \eqref{equ:LanguageSpecOld_q} does not hold, or, if  \eqref{equ:LanguageSpecOld_q} holds we know from \REFlem{lem:FAconditioned_vec} that \eqref{eq:SCT_nonblocking_q} does not hold.

\end{document}